\documentclass[a4paper,USenglish,numberwithinsect]{lipics}

\usepackage{etex} 
\usepackage[]{graphicx}
\usepackage{etoolbox} 

\newif\ifignore 
\ignorefalse
\newcommand{\auxproof}[1]{
\ifignore\mbox{}\newline
\textbf{BEGIN: AUX-PROOF} \dotfill\newline
{#1}\mbox{}\newline
\textbf{END: AUX-PROOF}\dotfill\newline
\fi}

\usepackage{hhline}
\usepackage{fancybox,amssymb,amstext,amsmath,stmaryrd,wasysym,cite,proof,mathtools}
\usepackage[pdftex,all]{xy}
\usepackage{xspace}
\allowdisplaybreaks[1] 
\usepackage{pstricks}

\xyoption{v2}
\xyoption{curve}
\xyoption{2cell}
\SelectTips{cm}{}  
\UseAllTwocells
\SilentMatrices
\def\labelstyle{\scriptstyle}

\newdir{ >}{{}*!/-8pt/@{>}}  
\newdir{|>}{%
!/1.6pt/@{|}*:(1,-.2)@^{>}*:(1,+.2)@_{>}}
\newdir{pb}{:(1,-1)@^{|-}}
  \def\pb#1{\save[]+<20 pt,0 pt>:a(#1)\ar@{pb{}}[]\restore}

\usepackage{wrapfig}
\newcommand{\true}{\mathtt{True}}

\newcommand{\X}{\mathsf{X}}
\newcommand{\F}{\mathsf{F}}
\newcommand{\Fdisc}[1]{\mathop{\F\raisebox{-.2ex}[0ex][0ex]{$\scriptstyle #1$}}\nolimits}
\newcommand{\G}{\mathsf{G}}
\newcommand{\Gdisc}[1]{\mathop{\G\raisebox{-.2ex}[0ex][0ex]{$\scriptstyle #1$}}\nolimits}
\newcommand{\U}{\mathbin{\mathsf{U}}}
\newcommand{\R}{\mathbin{\mathsf{R}}}
\newcommand{\N}{\mathbb{N}}

\newcommand{\Q}{\mathbb{Q}}

\newcommand{\Lang}{\mathcal{L}}

\newcommand{\A}{\mathcal{A}}
\newcommand{\Ap}{\mathcal{A}^{\mathrm{p}}}
\newcommand{\Ana}{\mathcal{A}^{\mathrm{na}}}
\newcommand{\K}{\mathcal{K}}
\newcommand{\D}{\mathcal{D}}
\newcommand{\Dexp}{\D_{\expo}}
\newcommand{\Fcal}{\mathcal{F}}
\newcommand{\Fcalmc}{\Fcal_{\mathrm{mc}}}
\newcommand{\Pcal}{\mathcal{P}}

\newcommand{\expo}{\mathrm{exp}}

\newcommand{\LTLd}[1]{\textrm{LTL}^\mathrm{disc}[#1]}
\newcommand{\LTLdD}{\LTLd{\D,\Fcal}}
\newcommand{\Sub}{\mathop{\mathrm{Sub}}\nolimits}
\newcommand{\place}{\underline{\phantom{n}}\,} 
\newcommand{\run}{\mathop{\mathrm{run}}\nolimits}
\newcommand{\pathrm}{\mathop{\mathrm{path}}\nolimits}
\newcommand{\ttrue}{\mathtt{t{\kern-1.5pt}t}}
\newcommand{\ffalse}{\mathtt{f{\kern-1.5pt}f}}

\newcommand{\sem}[1]{\llbracket #1 \rrbracket} 
\newcommand{\dotminus}{\mathbin{\scriptstyle\dot{\smash{\textstyle-}}}}
\newcommand{\AP}{\mathit{AP}}
\newcommand{\concatseq}{\mathbin{:}}

\newcommand{\deletewords}[1]{}    
\newcommand{\movewordsfrom}[2]{}  

\theoremstyle{definition}
\newtheorem{defi}{Definition}[section]

\theoremstyle{remark}
\newtheorem{rem}[defi]{Remark}
\newtheorem{exam}[defi]{Example}

\theoremstyle{plain}
\newtheorem{prop}[defi]{Proposition}
\newtheorem{cor}[defi]{Corollary}
\newtheorem{lem}[defi]{Lemma}
\newtheorem{sublem}[defi]{Sublemma}
\newtheorem{thm}[defi]{Theorem}

\AtEndEnvironment{thebibliography}{
}

\title{Near-Optimal Scheduling  for
 LTL with Future Discounting}

\author{Shota Nakagawa}
\author{Ichiro Hasuo}

\affil{Department of  Computer Science, The University of Tokyo
  }

\authorrunning{S. Nakagawa and I. Hasuo} 

\Copyright{Shota Nakagawa and Ichiro Hasuo} 



\serieslogo{}
\volumeinfo
  {Billy Editor and Bill Editors}
  {2}
  {Conference title on which this volume is based on}
  {1}
  {1}
  {1}
\EventShortName{}
\DOI{10.4230/LIPIcs.xxx.yyy.p}

\begin{document}


%
%
%


\maketitle              

 \begin{abstract}
  We study the search problem for optimal schedulers for the \emph{linear temporal
  logic (LTL) with future discounting}.  The logic, introduced by
  Almagor, Boker and Kupferman, is a quantitative variant of LTL in
  which an event in the far future has only discounted contribution to a
  truth value (that is a real number in the unit interval $[0,1]$). The
  precise problem we study---it naturally arises e.g.\ in search for a
  scheduler that recovers from an internal error state as soon as
  possible---is the following: given a Kripke frame, a formula and a
  number in $[0,1]$ called a \emph{margin}, find a path of the Kripke
  frame that is optimal with respect to the formula up to the prescribed
  margin (a truly optimal path may not exist). We present an algorithm
  for the problem; it works even in the extended setting with 
  propositional quality operators, a setting where (threshold)
  model-checking is known to be undecidable.

 \end{abstract}

\section{Introduction}\label{sec:intro}
In the field of \emph{formal methods} where a mathematical approach is
taken to modeling and verifying systems, the conventional theory
is built around the Boolean notion of truth: if a given system
satisfies a given specification, or not. This \emph{qualitative} theory
has produced an endless list of notable achievements from hardware
design to communication protocols. Among many techniques,
\emph{automata-based} ones for verification and synthesis have been
particularly successful in serving engineering needs, by offering 
a specification method by temporal logic and  push button-style
algorithms. See e.g.~\cite{Vardi96anautomata-theoretic,PnueliR89}.

However, trends today in the use of computers---computers as
part of more and more \emph{heterogeneous} systems---have pushed
researchers to turn to \emph{quantitative} consideration of systems,
too. For example, in an \emph{embedded system} where
a microcomputer 
controls a bigger system  with  mechanical/electronic components, 
concerns include  \emph{real-time properties}---if an expected task
is finished within the prescribed deadline---and \emph{resource consumption}
e.g.\ with respect to electricity, memory, etc. 

Quantities in formal methods can thus arise from a specification (or
 an \emph{objective}) that is quantitative in nature. Another
source of quantities are systems that are themselves quantitative, such
as one with
probabilistic 
behaviors.

Besides, quantities  can arise simply via \emph{refinement} of the Boolean notion of
 satisfaction. For example, consider 
the usual interpretation of the \emph{linear
 temporal logic (LTL)} formula
 $\F\varphi$---it 
is satisfied by a sequence 
$s_{0}s_{1}\dotsc$ 
 if there exists $i$ such that $s_{i}\models\varphi$. 
 It has the following natural
 quantitative refinement, where the modality $\F$ is replaced with 
 a \emph{discounted} modality $\Fdisc{\expo_{\frac{1}{2}}}$:
\begin{equation}\label{eq:introDiscounting}
 \sem{s_{0}s_{1}\dotsc,\,\Fdisc{\expo_{\frac{1}{2}}}\varphi}
 \;=\;
 \textstyle(\frac{1}{2})^{i}\enspace,
 \quad\text{where $i$ is the least index such that
 $s_{i}\models \varphi$.}
\end{equation}
This value $\sem{s_{0}s_{1}\dotsc,\,\Fdisc{\expo_{\frac{1}{2}}}\varphi}\in[0,1]$ is a
quantitative \emph{truth value} and is like \emph{utility} in the
game-theoretic terminology.
Such  refinements allow quantitative  reasoning about
so-called \emph{quality of service (QoS)}, specifically ``how soon $\varphi$
becomes true'' in this example. Another example is a quantitative
variation $\Gdisc{\expo_{\frac{1}{2}}}\varphi$ of $\G\varphi$, where
\begin{math}
 \sem{s_{0}s_{1}\dotsc,\,\Gdisc{\expo_{\frac{1}{2}}}\varphi}
 =
 1-\textstyle(\frac{1}{2})^{i}
 \end{math}---where $i$ is the least index such that
  $s_{i}\not\models \varphi$---meaning that violation of $\varphi$ in the far future only has a small
negative
impact.

\noindent\textbf{$\LTLdD$: LTL with Future Discounting}\quad
The last examples are about quantitative refinement of temporal
specifications. An important step in this direction is taken
in the recent work~\cite{AlmagorBK14}.
There various useful quantitative refinements in
LTL---including the last examples---are unified under the notion of
\emph{future discounting}, an idea first presented in~\cite{AlfaroHM03}
in the field of formal methods. They introduce a clean syntax of
the logic  $\LTLdD$---called
\emph{LTL with discounting}---that combines: 1) a ``discounting until'' operator
$\U_{\eta}$; 2) the usual features of LTL such as the non-discounting
one $\U$;  and 3) so-called propositional
quality operators such as the (binary) average operator $\oplus$, in
addition to $\land$ and $\lor$. In~\cite{AlmagorBK14} they define its semantics; and importantly, they
show that  usual automata-theoretic techniques for verification and
synthesis (e.g.\ from~\cite{Vardi96anautomata-theoretic,PnueliR89})
mostly remain applicable.  

Probably the most important algorithm in~\cite{AlmagorBK14} is for the
\emph{threshold model-checking problem}: given a Kripke structure $\K$,
a formula $\varphi$ and a \emph{threshold} $v\in[0,1]$, it asks if 
$\sem{\K,\varphi}>v$, i.e.\ 
the
worst case truth value of a path of $\K$ is above $v$ or
not.
The core idea of
 the algorithm
is what we call an \emph{event horizon}: assuming that a discounting
function $\eta$ in $\U_{\eta}$ tends to $0$ as time goes by, and
that $v>0$, there exists a time beyond which nothing is 
significant enough
to change the answer to the threshold model-checking problem.
In this case we can approximate an infinite path by its finite prefix.

\noindent
\textbf{Our Contribution: Near-Optimal Scheduling for
$\LTLdD$}\quad
Now that a temporal formula $\varphi$ assigns quantitative \emph{truth}
or \emph{utility} $\sem{\xi,\varphi}$ to each path $\xi$, a natural task is
to find a path $\xi_{0}$ in a given Kripke structure $\K$ that achieves the
optimal. On the ground that the logic $\LTLdD$ from~\cite{AlmagorBK14} is capable of 
expressing many common specifications encountered in real-world
problems, finding an optimal path---i.e.\ resolving nondeterminism in
the best possible way---must have numerous applications. The situation
is similar to one with \emph{timed automata}, for which optimal
scheduling problems are studied e.g.\ in~\cite{AbdeddaimAM06}. 

It turns out, however, that 
a (truly) optimal path need not exist
(Example~\ref{exam:optimalityNotAchievable}):
$v_{0}=\sup_{\xi\in\pathrm(\K)}\sem{\xi,\varphi}$ is obviously a limit point but no
$\xi_{0}$ achieves $\sem{\xi_{0},\varphi}=v_{0}$. This leads us to the
following \emph{near}-optimal scheduling problem:
\begin{quote}
 \textbf{Near-optimal scheduling.}
 Given a Kripke structure $\K$, an $\LTLdD$ formula $\varphi$ and a
 \emph{margin}
 $\varepsilon\in(0,1)$, find a path $\xi_{0}\in\pathrm(\K)$ that is
 \emph{$\varepsilon$-optimal}, that is,
 \begin{math}
  \sup_{\xi\in\pathrm(\K)}\sem{\xi,\varphi} -\varepsilon
  \le
  \sem{\xi_{0}, \varphi}\enspace.
 \end{math}
\end{quote}
We study automata-theoretic algorithms
for this problem. In the basic setting where there are no propositional
quality operators, we can find  a straightforward algorithm  that conducts
binary search using the model-checking algorithm
from~\cite{AlmagorBK14}. 
Our main contribution, however,  is
 an alternative algorithm  that
takes the usual workflow: it
constructs, from
a formula $\varphi$ and a margin $\varepsilon$, 
 an automaton $\A_{\varphi,\varepsilon}$
 with which we combine a system model $\K$;
running a nonemptiness check-like algorithm to the resulting automaton then
yields an answer. 

On  the one hand, our (alternative) algorithm  resembles the one in~\cite{AlmagorBK14}. 
In particular it  relies on the idea of event horizon:
a margin $\varepsilon$ in our setting plays the role
of a threshold $v$ in~\cite{AlmagorBK14} and
enables us to ignore events in the far future.

On the other hand, a major difference from~\cite{AlmagorBK14} 
is that we translate a specification
$(\varphi,\varepsilon)$ into an automaton that is itself quantitative
(what we call a \emph{$[0,1]$-acceptance
automaton}, with  Boolean
branching and $[0,1]$-acceptance values). This is 
unlike~\cite{AlmagorBK14} where the target automaton
is totally Boolean.
An advantage of  $[0,1]$-acceptance
automata is that they allow optimal path search 
much like
emptiness of B\"uchi automata is checked (via
lasso computations). 
Applied to our current problem, 
this enables us to directly find a
near-optimal path
for $\LTLdD$
without knowing the optimal value
$\sup_{\xi\in\pathrm(\K)}\sem{\xi,\varphi}$.

\noindent
\textbf{Presence of $\oplus$ and Other Propositional Quality Operators}\quad
Notably,  our  (alternative) algorithm is shown to work even in the presence
 of
 any propositional quality operators that are \emph{monotone} and
 \emph{continuous} (in the sense we will define later; an example is the
 average operator $\oplus$). 
Those operators makes the logic  more complex: indeed~\cite{AlmagorBK14}
shows that, in presence of 
 the average operator $\oplus$, the model-checking problem for the logic
 $\LTLdD$ becomes
 undecidable.
 The binary-search algorithm mentioned earlier (that
 repeats model checking) ceases to work for this reason; our
 alternative algorithm works, nevertheless.


We analyze the complexity of the proposed algorithm, focusing on a certain
subclass of the logic $\LTLdD$ (\S\ref{subsec:complexity}). Furthermore
we present our prototype implementation and some experimental
results (\S\ref{sec:experiments}). They all seem to suggest the following: addition of propositional
quality
operators (like the average operator $\oplus$) does incur substantial
computational costs---as is expected from the fact that $\oplus$ makes
model checking undecidable; still our automata-theoretic approach is a
viable approach, potentially applicable to optimization 
problems in the field of model-based system design. 

The significance of the average operator $\oplus$ in
envisaged applications is that it allows one to \emph{superpose} multiple
objectives. For example, one would want an event $\varphi$ as soon as
possible, but at the same time avoiding a different event $\psi$ as long as
possible. This is a trade-off situation and the formula
$\Fdisc{\eta}\varphi\oplus \Gdisc{\eta'}\lnot \psi$---with suitable discounting
functions $\eta,\eta'$---represents a 50-50 trade-off. Other
 trade-off ratios can be represented as (monotone and continuous) proportional quality operators, too,
and our algorithm accommodates them.



\noindent
\textbf{Related Work}\quad
Quantitative  temporal logics and their decision procedures
have  been a very active research
topic~\cite{AlmagorBK14,AlmagorBK13,AlfaroHM03,BouyerMM14,FaellaLS08}. 
 We shall lay them out along a basic taxonomy. We denote by $\K$  (the model of) the
system against which a specification formula $\varphi$ is verified (or tested, synthesized,
etc.).

\begin{itemize}
\item
  \textit{Quantitative vs.\ Boolean system models.} Sometimes we need
       quantitative considerations just because the system
        $\K$ itself is quantitative. This is the case e.g.\ when $\K$ is a Markov chain, a Markov decision
       process, a timed or hybrid automaton, etc. In the current work
       $\K$ is a Kripke structure and is Boolean.

\item
 \textit{Quantitative vs.\ Boolean truth values.} The previous
       distinction is quite orthogonal to whether a formula $\varphi$
       has truth values from $[0,1]$ (or another continuous domain), or from
       $\{\ttrue,\ffalse\}$. For example, the temporal logic PCTL~\cite{HanssonJ94} for
       reasoning about probabilistic systems has modalities
       like $\mathcal{P}_{>v}\psi$ (``$\psi$ with a probability 
       $>v$'') and has Boolean interpretation.  In $\LTLdD$ studied  here,
       truth values are from $[0,1]$.

\item
 \textit{Linear time vs.\ branching time.} This distinction is
       already there in the qualitative/Boolean
       setting~\cite{vanGlabbeek01}---its probabilistic variant is
       studied in~\cite{CheungSV07}---and gives rise to temporal logics
       with the corresponding flavors (LTL vs.\ CTL, $\textrm{CTL}^{*}$).
       In fact the idea of future discounting is first introduced to 
       a branching-time logic in~\cite{AlfaroHM03}, where an
       approximation algorithm for truth values is presented.


\item
 \textit{Future discounting vs.\ future averaging.} The temporal
       quantitative operators in $\LTLdD$ are \emph{discounting}---an
     event's significance tends to $0$ as time proceeds---a fact that
       benefits model checking  via event horizons. 
        Different  temporal
       quantitative operators are studied in~\cite{BouyerMM14}, including the
       \emph{long-run average} operator $\widetilde{\G}\psi$. 
     Presence of $\widetilde{G}$, however, makes most
       common decision problems undecidable~\cite{BouyerMM14}.

\end{itemize}

In~\cite{FaellaLS08} LTL (without additional quantitative operators)
is interpreted over the unit interval $[0,1]$, and its model-checking
problem against quantitative systems $\K$ is shown to be decidable. 
In this setting---where the LTL connectives are interpreted by
   idempotent operators $\min$
and $\max$---the variety of truth values arises only from a finite-state
quantitative system $\K$, hence is finite.


In~\cite[Thm.~4]{AlmagorBK14} it is proved that the \emph{threshold
synthesis} problem for the logic $\LTLd{\D,\emptyset}$ (see Def.~\ref{def:LTLdDSyntax}) is feasible. This problem asks:
given a partition of atomic propositions into the input and output
signals, an $\LTLd{\D,\emptyset}$ formula $\varphi$ and $v\in[0,1]$, to come up with a transducer (i.e.\ a finite-state strategy)
 that makes the truth value of $\varphi$ at least
$v$. 
We remark that this is
different from the near-optimal scheduling problem that we
solve
in this paper. The \emph{synthesis} problem
in~\cite[\S{}2.2]{AlmagorBK13}, without a threshold, is closer to ours.

Automata- (or game-) theoretic approaches are  taken 
in~\cite{BloemCHJ09,CernyCHRS11}
to the synthesis of controllers or programs with better
quantitative performance, too. In these papers, a specification is 
given itself as an automaton, instead of a temporal formula in the
current work. Another difference is that, in~\cite{BloemCHJ09,CernyCHRS11}, utility is computed
along a path by limit-averaging, not future discounting. The algorithms in~\cite{BloemCHJ09,CernyCHRS11} therefore rely on 
those which are known for mean-payoff games, including the ones
in~\cite{ChatterjeeHJ05}.

More and more
diverse quantitative measures of systems' QoS are studied recently:
from best/worst case probabilities and costs, to quantiles,
conditional probabilities and ratios. See~\cite{BaierDK14} and the
references
therein.  Study of such in $\LTLdD$ is future work.

In~\cite{ChatterjeeDH10} so-called \emph{cut-point languages} of 
weighted automata are studied.
Let $L:\Sigma^{\omega} \to \mathbb{R}$ be the quantitative 
language of a weighted automata $\A$.
For a threshold $\eta$, the cut-point language of $\A$ is 
the set consisting of all words $w$ such that $L(w) \geq \eta$.
In~\cite{ChatterjeeDH10} it is proved that the cut-point languages of 
deterministic limit-average automata and those of discounted-sum automata 
are $\omega$-regular if the threshold $\eta$ is \emph{isolated}, that is, 
there is no word $w$ such that $L(w)$ is close to $\eta$. 
We expect that similar properties for the logic $\LTLdD$ are not hard to
establish, although details are yet to be worked out.




\noindent
\textbf{Organization of the Paper} \quad
In~\S{}\ref{sec:syntaxAndThresholdProblem} we review the logic
$\LTLdD$ and known results on threshold model checking and
satisfiability,
all from~\cite{AlmagorBK14}. We introduce  quantitative
variants of (alternating) B\"uchi automata, called
(alternating)  $[0,1]$-acceptance automata,
in~\S{}\ref{sec:zeroOneAutomata}, with  auxiliary observations on their
relation to \emph{fuzzy automata}~\cite{Rahonis05}.
These automata
play a central role in~\S{}\ref{sec:nearOptimalSchedulerSynth} where we
formalize and solve the near-optimal scheduling problem for the logic
$\LTLdD$ (under certain assumptions on $\D$ and $\Fcal$).
We also study complexities, focusing on the average operator $\oplus$ as
the only propositional quality operator. In~\S{}\ref{sec:experiments} we
present our implementation and some experimental results;
in~\S{}\ref{sec:conclFutureWork} we conclude, citing some future work.
Omitted proofs are found in Appendix~\ref{appendix:omittedproofs}.

\noindent
\textbf{Notations and Terminologies} \quad
We shall fix some notations and terminologies, mostly
following~\cite{AlmagorBK14}.  They are all standard.

The powerset of a set $X$ is denoted by $\mathcal{P}X$.
We fix the set $\AP$  of \emph{atomic propositions}.
 A
\emph{computation} (over $\AP$) is an infinite sequence $\pi= \pi_{0} \pi_{1} \ldots \in
 (\mathcal{P}(\AP))^{\omega}$ over the alphabet
 $\mathcal{P}(\AP)$. 
For $i\in\N$, 
$\pi^{i} = \pi_{i} \pi_{i+1} \ldots$ 
denotes the suffix of $\pi$ starting from its $i$-th element.

A \emph{Kripke structure} over $\AP$ is a tuple $\K=(W,R,\lambda)$ of:  
  a finite set $W$ of states; a  transition
 relation 
 $R \subseteq W^{2}$ 
that is left-total (meaning that $\forall s\in W.\,\exists
 s'\in W.\,(s,s')\in R$), and a labeling function  $\lambda : W
  \rightarrow \mathcal{P}(\AP)$. 
We follow~\cite{KupfermanVW00} and call an infinite sequence  $\xi=s_{0}s_{1}\dotsc$ of states $s_{i}\in W$, such that
 $(s_{i},s_{i+1})\in R$ for each $i\in\N$,  a \emph{path} of
  a Kripke structure $\K$.
The set of paths of $\K$ is
 denoted by $\pathrm(\K)$. A path
$\xi=s_{0}s_{1}\dotsc\in W^{\omega}$ gives rise to a computation
$\lambda(s_{0})\,\lambda(s_{1})\dotsc\in(\mathcal{P}(\AP))^{\omega}$; the
  latter is  denoted by $\lambda(\xi)$.

Given a set $X$, 
 $\mathcal{B}^{+}(X)$ denotes, as
 usual, the set of positive propositional formulas (using
 $\land,\lor,\top,\bot$) over $x\in X$  as atomic
 propositions.


\section{The Logic $\LTLd{\D,\Fcal}$, and Its Threshold Problems}
\label{sec:syntaxAndThresholdProblem}
Here we recall from~\cite{AlmagorBK13,AlmagorBK14}  our target
logic, and some existing (un)decidability results.

The logic $\LTLd{\D,\Fcal}$  extends LTL with: 1) 
propositional quality operators~\cite{AlmagorBK13} like the 
average operator $\oplus$; and 2) discounting in temporal
operators~\cite{AlmagorBK14}. In~\cite{AlmagorBK14} the two extensions
have been studied separately because their coexistence leads to undecidability
of the (threshold) model-checking problem; here we put them altogether.

\auxproof{The logic is quantitative---a truth value is a
number in the unit interval $[0,1]$ rather than a Boolean value from
$\{\ttrue,\ffalse\}$. The logic features fixed point operators with
future discounting. 
}

 The logic $\LTLd{\D,\Fcal}$ has two parameters: 
a set $\D$ of discounting
functions; and a set $\Fcal$ of propositional connectives,
called propositional quality operators. 
\begin{defi}[discounting function~\cite{AlmagorBK14}]\label{def:discountFunction}
 A \emph{discounting function} is a strictly decreasing function $\eta :
 \N \rightarrow [0,1]$ such that $\lim_{i \to \infty} \eta(i) = 0$. A
 special case is an \emph{exponential discounting function}
 $\expo_{\lambda}$, where $\lambda \in (0,1)$, that is defined by
 $\expo_{\lambda}(i) = \lambda^{i}$. 

 The set
\begin{math}
 \Dexp 
= \{ \expo_{\lambda} \mid \lambda \in (0,1) \cap \Q \}
\end{math} is that of exponential discounting
 functions.
\end{defi}

\auxproof{
\begin{rem}
 An extension of the framework is proposed in~\cite{AlmagorBK14} where a discounting functions $\eta$ need
 not tend to $0$. It is claimed that
 such an extension, e.g.\ when $\eta$ tends to $\frac{1}{2}$, is suited for
 the situation where we are (not totally pessimistic but) ambivalent
 about the future. This extension does not change the algorithmic
 results in~\cite{AlmagorBK14}, nor here. It is
 enough that the limit value $\lim_{k\to\infty}\eta(k)$ is statically known so that
 we can use the value in construction of automata.
\end{rem}
}


\begin{defi}[(monotone and continuous) propositional quality
 operator~\cite{AlmagorBK13}]\label{def:propositionalQualityOperator}
 Let $k \in \N$ be a natural number.
 A $k$-ary \emph{propositional quality operator} is a 
 function
 $f :
 [0,1]^{k} \rightarrow [0,1]$. 

 We will eventually restrict to  propositional quality operators that are 
 \emph{monotone} (wrt.\ the usual order between real numbers)
 and \emph{continuous} (wrt.\ the usual Euclidean topology). 
 The set of monotone and continuous propositional quality operators
is denoted by $\Fcalmc$.
\end{defi}

\begin{exam}\label{exam:propositionalQualityOperator}
 A prototypical example of a propositional quality operator is 
the \emph{average operator} $\oplus\colon [0,1]^{2}\to [0,1]$, defined
 by $v_{1} \oplus v_{2} = (v_{1} + v_{2})/2$. (Note that $\oplus$ is a
 ``propositional'' average operator and is different from the
 ``temporal'' average operator $\widetilde{\U}$
 in~\cite{BouyerMM14}). The operator $\oplus$ is monotone and
 continuous. 
 Other (unary) examples from~\cite{AlmagorBK13ExtendedPreprint} include: 
 $\triangledown_{\lambda}(v)=\lambda\cdot v$ and
 $\blacktriangledown_{\lambda}(v)=\lambda\cdot v + (1-\lambda)$ (they
 are explained in~\cite{AlmagorBK13ExtendedPreprint} to express \emph{competence} and \emph{necessity}, respectively).
 The conjunction and disjunction connectives $\land,\lor$, interpreted 
 by  infimums and supremums in $[0,1]$, can also be regarded as binary
 propositional quality operators. They are  monotone and
 continuous, too. 
\end{exam}

Recall that the set $\AP$  is that of atomic propositions. 
\begin{defi}[$\LTLd{\D,\Fcal}$]\label{def:LTLdDSyntax}
Given a set $\D$ of discounting functions and a set $\Fcal$ of propositional quality operators, the \emph{formulas} of $\LTLd{\D,\Fcal}$
are defined by the  grammar:
 \begin{displaymath}
  \varphi ::= \true \mid p \mid \lnot \varphi \mid \varphi \land \varphi \mid \X \varphi \mid \varphi \U \varphi \mid \varphi \U_{\eta} \varphi \mid f(\varphi,\ldots,\varphi)\enspace,
 \end{displaymath}
 where $p \in \AP$, $\eta \in \D$ is a discounting function and
 $f \in \Fcal$ is a propositional quality operator (of a suitable arity). 
 We adopt the usual notation conventions: 
\begin{math}
{\F \varphi} 
=
{\true \U \varphi}
\end{math}
and 
\begin{math}
 \G \varphi
=
\lnot \F \lnot \varphi
\end{math}. The same goes for discounting operators:
\begin{math}
{\Fdisc{\eta} \varphi} 
=
{\true \U_{\eta} \varphi}
\end{math}
and 
\begin{math}
 \Gdisc{\eta} \varphi
=
\lnot \Fdisc{\eta} \lnot \varphi
\end{math}.
\end{defi}
%
%
As we have already discussed, 
the logic $\LTLd{\D,\Fcal}$ extends the usual LTL with: 1) discounted
temporal operators like $\U_{\eta}$ (cf.~(\ref{eq:introDiscounting}));
and 2) propositional quality operators like $\oplus$ that operate, on
 truth values from $[0,1]$ that arise from the discounted modalities, in
 the ways other than $\land$ and $\lor$ do. The  precise
 definition below
closely follows~\cite{AlmagorBK13,AlmagorBK14}.

\begin{defi}[semantics of $\LTLd{\D,\Fcal}$~\cite{AlmagorBK13,AlmagorBK14}]\label{def:LTLdDSemantics}
Let $\pi= \pi_{0} \pi_{1} \ldots \in
 (\mathcal{P}(\AP))^{\omega}$ be a computation (see~\S{}\ref{sec:intro}), and
 $\varphi$ be an $\LTLd{\D,\Fcal}$ formula.
 The \emph{truth value} $\sem{ \pi, \varphi }$ of $\varphi$ in $\pi$
 is a real number in $[0,1]$ defined
 as follows. Recall that $\pi^{i} = \pi_{i} \pi_{i+1} \ldots$ is a
 suffix of $\pi$.
 \begin{displaymath}\small
\renewcommand*{\arraystretch}{1}
\begin{array}{ll}
 \sem{ \pi,\true } \;=\; 1
\qquad
& 
\sem{ \pi,p } \;=\; 
1 \quad\text{(if $p\in\pi_{0}$);}
\qquad
0 \quad\text{(if $p\not\in\pi_{0}$)}
\\
\sem{ \pi, \lnot \varphi}
\;=\; 1 - \sem{ \pi,\varphi }
\phantom{hogehoge}
 &
\sem{ \pi,\varphi_{1} \land \varphi_{2} } \;=\; \min \bigl\{\, \sem{ \pi,\varphi_{1} }, \sem{ \pi,\varphi_{2} } \,\bigr\}
 \\
\sem{\pi, \X \varphi} \;=\; \sem{ \pi^{1},\varphi }
\\
\multicolumn{2}{l}{
\textstyle
\sem{ \pi, \varphi_{1} \U \varphi_{2} } 
\;=\;  \sup_{i \in \N} \bigl\{\, \min \bigl\{ \sem{ \pi^{i}, \varphi_{2} }, \min_{0 \leq j < i} \sem{ \pi^{j}, \varphi_{1} } \bigr\}
\,\bigr\}
}
\\
\multicolumn{2}{l}{
\textstyle
\sem{ \pi, \varphi_{1} \U_{\eta} \varphi_{2} } 
\;=\;  \sup_{i
\in \N} \bigl\{\, \min \bigl\{\, \eta(i)\sem{ \pi^{i}, \varphi_{2} },\, \min_{0 \leq j
< i} \eta(j)\sem{ \pi^{j}, \varphi_{1} } \,\bigr\}
\,\bigr\}
}
\\
\sem{\pi, f(\varphi_{1},\ldots,\varphi_{k})} \;=\; f \bigl(\sem{ \pi,\varphi_{1} },\ldots,\sem{ \pi,\varphi_{k} }\bigr)
\end{array}  
 \end{displaymath}
\end{defi}
Compare the semantics of $\varphi_{1} \U \varphi_{2}$
and that of $\varphi_{1} \U_{\eta} \varphi_{2}$. The former is a
straightforward quantitative analogue of the usual Boolean semantics;
the latter additionally includes ``discounting'' by
$\eta(i),\eta(j)\in[0,1]$. Recall that a discounting function $\eta$ is
deemed to be strictly decreasing; this allows us to express intuitions
like in~(\ref{eq:introDiscounting}).

\begin{prop}\label{lem:maxTruthValOfDiscountedUntil}
 The truth value $\sem{ \pi, \varphi_{1} \U_{\eta} \varphi_{2} } $
 lies between $0$ and $\eta(0)$. \qed
\end{prop}

We extend the semantics to Kripke structures (see~\S{}\ref{sec:intro}). 
\begin{defi}
 Let $\K$ be a Kripke structure and $\xi$ be a path of $\K$. 
 The truth value $\sem{ \xi, \varphi }$ of $\varphi$ in the path $\xi$ is defined by
 $\sem{ \xi, \varphi } =  \sem{ \lambda(\xi), \varphi }$, where
 $\lambda(\xi)\in(\mathcal{P}(\AP))^{\omega}$ is the computation 
 induced by $\xi$ (see~\S{}\ref{sec:intro}). 
 The truth value $\sem{ \K, \varphi }$ of $\varphi$ in $\K$ is defined by $\sem{ \K, \varphi } = \inf_{\xi \in \pathrm(\K)} \sem{ \xi, \varphi }$.
\end{defi}
\auxproof{
The choice of infimum (instead of $\sup$) in the definition of $\sem{ \K,
\varphi }$ is due to~\cite{AlmagorBK14} and is more natural, modeling
the worst case value, when we consider
the threshold model-checking problem (see
below).
}


\begin{rem} 
 Later in this paper we will restrict to propositional quality operators 
 that are monotone and
 continuous, i.e.\ $\LTLd{\D,\Fcal}$  with $\Fcal\subseteq\Fcalmc$. 
 Such a logic can nevertheless express some non-monotonic 
 operators with the help of negation. For example, the function
 $f_{0}\colon [0,1]\to[0,1], f_{0}(v) = |v-\frac{1}{2}|$ can be expressed as 
 a combination $f_{0}(v) = \max \{ 1 - f_{1} (v), f_{2} (v) \}$, using
$f_{1} (v) = \min \{v + \frac{1}{2}, 1\}$ and $f_{2} (v) = \max \{ v- \frac{1}{2}, 0 \}$
 (note that  $f_{1},f_{2}\in\Fcalmc$)---i.e.\ as the semantics of the formula
 $(\lnot f_{1}\varphi)\lor (f_{2}\varphi)$. A nonexample is the function 
 $f_{3}(v)=v\cdot\sin \frac{1}{v}$ that oscillates infinitely often in $[0,1]$.
\end{rem}

The following ``threshold'' problems are studied
in~\cite{AlmagorBK14,AlmagorBK13ExtendedPreprint}. It is shown that the logic
$\LTLd{\D,\emptyset}$---i.e.\ without propositional quality operators
other than $\land, \lor$---has those problems decidable. Adding the
average operator $\oplus$ makes them undecidable~\cite{AlmagorBK14},
while adding $\triangledown_{\lambda}$
(Example~\ref{exam:propositionalQualityOperator}) maintains decidability~\cite{AlmagorBK13ExtendedPreprint}.
Here the complexities are  in terms of a suitable notion
 $|\langle \varphi \rangle|$ of the size of $\varphi$ 
 (see~\cite{AlmagorBK14}).


\begin{thm}[\cite{AlmagorBK14}]\label{thm:thresholdModelCheckingDecidable}
 The \emph{threshold model-checking problem} for $\LTLd{\D,\emptyset}$ is:
given a
 Kripke structure $\K$, an $\LTLd{\D,\emptyset}$ formula $\varphi$ and a threshold $v
 \in [0,1]$,  decide whether $\sem{ \K, \varphi } \geq v$. It
 is decidable; when restricted to $\LTLd{\Dexp,\emptyset}$ and $v\in\Q$, the problem is in PSPACE in $|\langle
 \varphi \rangle|$ and in the description of  $v$, and in NLOGSPACE in the
 size of $\K$.

 The \emph{threshold satisfiability problem} for $\LTLd{\D,\emptyset}$ is:
given an $\LTLd{\D,\emptyset}$ formula $\varphi$, a threshold $v \in
 [0,1]$
and   $\mathord{\sim} \in \{ <, > \}$, decide whether there exists a computation $\pi \in
 (\mathcal{P}(\AP))^{\omega}$ such that $\sem{ \pi, \varphi }
 \sim v$. This is decidable; when restricted to $\LTLd{\Dexp,\emptyset}$ 
 and $v\in\Q$, the problem is in PSPACE in $|\langle
 \varphi \rangle|$ and in the description of  $v$. \qed
\end{thm}
%
%
\begin{thm}[\cite{AlmagorBK14}]\label{thm:averageUndecidable}
  For $\LTLd{\D,\{ \oplus \}}$ where $\D \neq \emptyset$, both the
 threshold model-checking problem and the threshold satisfiability problem are undecidable.
 \qed
\end{thm}





\section{$[0,1]$-Acceptance B\"uchi Automata}
\label{sec:zeroOneAutomata}

Our algorithm for near-optimal scheduling relies on a certain
notion of quantitative automaton---called \emph{$[0,1]$-acceptance B\"uchi
automaton}, see Def.~\ref{def:zeroOneAcceptanceBuchi}---and an algorithm for its optimal value problem
(Lem.~\ref{lem:lassoOptimalityForQuantitativeAcceptAutom}).
The notion is not
extensively studied in the literature,   to the best of our knowledge.

In a $[0,1]$-acceptance B\"uchi
automaton 
 each state has a real value
$v\in[0,1]$,
instead of 
a Boolean value
 $b\in\{\ttrue,\ffalse\}$, of acceptance.
Note that branching is Boolean (i.e.\ nondeterministic) and not
$[0,1]$-weighted. 
 In Appendix~\ref{appendix:fuzzyAndZeroOne}
we study a relationship to so-called \emph{fuzzy automata} (see
e.g.~\cite{Rahonis05}) and
show
that adding weights to branching does not increase expressivity when it
comes to (weighted) languages.

 \begin{defi}[{$[0,1]$}-acceptance automaton]
  \label{def:zeroOneAcceptanceBuchi}
  A \emph{$[0,1]$-acceptance B\"uchi automaton}---or simply a
  \emph{$[0,1]$-acceptance automaton} henceforth---is $\A =
 (\Sigma,Q,I,\delta,F)$, where $\Sigma$ is a finite  alphabet,
 $Q$ is a finite set of states, $I \subseteq Q$ is a set of initial
 states, $\delta : Q \times \Sigma \rightarrow \left(\mathcal{P}(Q) \setminus \{ \emptyset \}\right)$ is a
 transition function and $F : Q \rightarrow [0,1]$ is 
 a function that assigns an \emph{acceptance value} to each state.
  We define the (weighted) language $\Lang(\A) : \Sigma^{\omega} \rightarrow [0,1]$ of $\A$ by
  \begin{equation}\label{eq:acceptanceZeroOneAutom}\textstyle
     \Lang(\A)(w) \;=\; \max \{ F(q) \mid \exists \rho \in \run(w).\, q
    \in \mathrm{Inf}(\rho) \} \quad\text{for each $w\in \Sigma^{\omega}$}\enspace,
  \end{equation}
 where the sets $\run(w)$ and $\mathrm{Inf}(\rho)$ are defined as
  usual. Precisely:
\begin{itemize}
 \item For an infinite word $w\in\Sigma^{\omega}$, a \emph{run} over $w$
       of $\A$ is an infinite alternating sequence
       $\rho=q_{0}a_{0}q_{1}a_{1}\dotsc$ such that: 1) $q_{i}\in Q$ is a
       state and $a_{i}\in \Sigma$ is a letter, for all $i\in \N$; 2)
       $q_{0}\in I$; and 3)
       $q_{i+1}\in \delta(q_{i},a_{i})$ for all $i\in \N$. The set of
       runs over $w$ is denoted by  $\run(w)$.
 \item Given a run $\rho$, the set
       $\mathrm{Inf}(\rho)$ is defined by
       $\mathrm{Inf}(\rho)=\{q\in Q\mid \text{$q$ occurs infinitely
       often in $\rho$}\}$.
\end{itemize}
\end{defi}
 Note that, when we restrict to Boolean acceptance values (i.e.\
$F(q)\in \{0,1\}$), 
the acceptance value in~(\ref{eq:acceptanceZeroOneAutom})
precisely coincides with the one in the usual notion of B\"{u}chi
automaton.
Note also that, in~(\ref{eq:acceptanceZeroOneAutom}), we take the
maximum of finitely many values (the state space $Q$ is finite).

\auxproof{
{$[0,1]$}-acceptance B\"uchi automata are quantitative extension of B\"uchi
automata. The extension, however, does not incur too much additional 
complexity. For example its threshold acceptance problem can be solved
by a B\"uchi automaton.

\begin{lem}
  Let $\A = (\Sigma,Q,I,\delta,F)$ be a {$[0,1]$}-acceptance B\"uchi automaton,
 $v\in[0,1]$, and $\triangleright\in\{>,\geq\}$.
 \begin{enumerate}
  \item There exists an (ordinary) B\"uchi automaton $\A_{\triangleright
	v}$ such that 
  \begin{math}
    \Lang(\A_{\triangleright v}) = \{ w \in \Sigma^{\omega} \mid \Lang(\A)(w) \triangleright  v \}
  \end{math}.
  \item Let  $F_{v} = \{ q \in Q \mid F(q) = v \}$, and
 $\A_{v} = (\Sigma,Q,I,\delta,F_{v})$  be an (ordinary) B\"uchi
	automaton. We have
  \begin{math}
    \Lang(\A)(w) = \max \{ v \in [0,1] \mid w \mbox{ is accepted by } \A_{v} \}
  \end{math}.
 \end{enumerate}
\end{lem}
\begin{proof}
 For 1., let $F_{\triangleright v} = \{ q \in Q \mid F(q) \triangleright  v \}$ and let
 $\A_{\triangleright v} = (\Sigma,Q,I,\delta,F_{\triangleright v})$. \qed
\end{proof}
}



The following observation, though not hard, is a key fact for
our search algorithm. 
It is a quantitative
 analogue
of emptiness check in usual (Boolean) automata.
\begin{lem}[the optimal value problem for {$[0,1]$}-acceptance automata]\label{lem:lassoOptimalityForQuantitativeAcceptAutom}
  Let $\A = (\Sigma,Q,I,\delta,F)$ be a {$[0,1]$}-acceptance B\"uchi automaton.
There exists the maximum
$\max_{w \in \Sigma^{\omega}} \Lang(\A)(w)$ 
 of $\Lang(\A)$.
Moreover, 
there is an algorithm that computes 
the value $\max_{w \in \Sigma^{\omega}} \Lang(\A)(w)$ 
as well as a run $\rho_{\mathrm{max}}
=
q_{0}a_{0}q_{1}a_{1}\dotsc \in
 (\Sigma\times Q)^{\omega}
$ 
that realizes the maximum.

\end{lem}
\auxproof{Note that
existence of 
$\max_{w \in \Sigma^{\omega}} \Lang(\A)(w)$ (as opposed to a supremum)
is nontrivial.
}
\begin{proof}
 The algorithm is much like the one for emptiness check of (ordinary)
B\"uchi automata, searching for a suitable lasso computation. More
concretely: consider those states $q$ which are both reachable from some
initial state and reachable from $q$ itself. Let $s$ be one, among those
 states, with the
greatest acceptance value $F(s)$. It is easy to show that a lasso
computation with the state $s$ as a ``knot'' gives the run
$\rho_{\mathrm{max}} $ that we seek for.
\end{proof}


 Our algorithm first translates a formula    into an
\emph{alternating} $[0,1]$-acceptance automata. 
\begin{defi}[alternating {$[0,1]$}-acceptance automaton]\label{def:alternative_buchi}
  An \emph{alternating {$[0,1]$}-acceptance (B\"uchi) automaton} is a tuple $\A =
 (\Sigma,Q,I,\delta,F)$, where $\Sigma$ is a finite alphabet,
 $Q$ is a finite set of states, $I \subseteq Q$ is a set of initial
 states, $\delta : Q \times \Sigma \rightarrow \mathcal{B}^{+} (Q \cup
 [0,1])$ is a transition function and $F : Q \rightarrow [0,1]$ gives
 acceptance values.  Recall (\S{}\ref{sec:intro})
 that $\mathcal{B}^{+}(Q \cup [0,1])$ is the set of positive
 propositional combinations of $q\in Q$ and $v\in[0,1]$.

  We define the (weighted) language $\Lang(\A) : \Sigma^{\omega} \rightarrow [0,1]$ of $\A$ by
  \begin{equation}\label{eq:langOfAFBA}\textstyle
    \Lang(\A)(w) \;=\; \max_{\tau \in \run_{\A}(w)} \min_{\rho \in \pathrm_{\A,w} (\tau)} F^{\infty}(\rho)\enspace,
  \end{equation}
  where \emph{runs}, \emph{paths} and the function $F^{\infty}$ are
 formally defined much like with the usual alternating
 automata. Precisely:
 \begin{itemize}
  \item A run is much like with the usual alternating
 automata. Precisely, let $\A=(\Sigma,Q,I,\delta,F)$ be an alternating
	$[0,1]$-acceptance automaton and 
 $w=a_{0}a_{1}\dotsc\in \Sigma^{\omega}$ be an infinite word.
	A \emph{run} $\tau$ of  $\A$ over 
 $w$ is a (possibly infinite-depth) tree subject to the following.
	\begin{itemize}
	 \item Each node $t$ of the tree $\tau$ is labeled from $Q\cup
	       [0,1]$.
\auxproof{	       That is, either by a state $q\in Q$
	       or a number $v\in[0,1]$.
}
	 \item The root of $\tau$ is labeled with an initial state
	       $q_{0}\in I$.
	 \item Any node $t$ labeled with a number $v\in[0,1]$ is a leaf.
	 \item Consider an arbitrary node $t$ that is labeled 
	with a state $q\in Q$. Assume that $t$ is of depth $i\in\N$; and 
	       let the labels of $t$'s children be $l_{1},\dotsc, l_{k}\in Q\cup[0,1]$.
We require  
	       $l_{1},\dotsc,
	l_{k}\models
	\delta(q,a_{i})$, where: $\delta(q,a_{i})\in\mathcal{B}^{+} (Q \cup
 [0,1])$ is the
	       $a_{i}$-successor of $q$ in $\A$; and $\models$
	       designates the obvious Boolean notion of satisfaction
	       (where we think of elements of $Q \cup
 [0,1]$ as atomic variables).
	\end{itemize}
	The set $\run_{\A}(w)$ is that of all runs of $\A$ over the word
	$w$.
  \item A \emph{path} $\rho$ of a run $\tau\in\run_{\A}(w)$ is simply a (finite or
	infinite) path in the tree $\tau$, from the root of $\tau$. 
	A path $\rho$ is finite only
	when its last state is a leaf of $\tau$.
	The set of paths of $\tau\in\run_{\A}(w)$ is denoted by $\pathrm_{\A,w} (\tau)$.
  \item The function $F^{\infty}\colon \pathrm_{\A,w}(\tau)\to [0,1]$
 in~(\ref{eq:langOfAFBA}) is defined as follows. 
	If $\rho\in\pathrm_{\A,w}(\tau)$ is an
	infinite path,  each node $t$ in $\rho$ is labeled with
	a state $q$ of $\A$. We define
 \begin{equation}\label{eq:langOfAFBA1}
  F^{\infty}(\rho)
  \;=\;
  \max\{\,F(q)\,\mid\,\text{$q\in Q$ occurs infinitely often, as labels,
  in $\rho$}\,\}\enspace.
 \end{equation}
	Assume now that $\rho\in\pathrm_{\A,w}(\tau)$ is finite, say
	$\rho=t_{0}t_{1}\dotsc t_{i}$. Then the last node $t_{i}$
	is labeled either by $v\in[0,1]$ or $q\in [0,1]$.
	In the former case we define
	\begin{math}
	 F^{\infty}(\rho)
	 =
	 v
	\end{math} 
	(i.e.\ $F^{\infty}$ returns the label of $t_{i}$).
 In the latter
	case,
	we have that $\delta(t_{i}, a_{i})$ is
	propositionally equivalent to $\top$ (``truth'') by the
	definition of run. We define
	\begin{math}
	 F^{\infty}(\rho)
	 =
	 1
	\end{math}.
 \end{itemize}
%
\end{defi}
In the above we used $\max$ and $\min$ (not $\sup$
or $\inf$) since $\{F(q)\mid q\in Q\}$ is a finite set.

\begin{prop}\label{prop:ABAtoNBA}
  Let $\A = (\Sigma,Q,I,\delta,F)$ be an alternating
 {$[0,1]$}-acceptance  automaton. There exists a {$[0,1]$}-acceptance
 automaton $\A'$ such that $\Lang(\A) = \Lang(\A')$. 
\qed
\end{prop}
The  construction of $\A'$ is a quantitative adaptation of the
 one in~\cite{MiyanoH84} that turns an alternating $\omega$-automaton into
a  nondeterministic one. In our adaptation we use what we call
 \emph{exposition flags}, an idea that is potentially useful in other
 settings with B\"uchi-type acceptance conditions, too. 
 See Appendix~\ref{pf:lemABAtoNBA} for details of the proof and the
 construction therein.

Later we will also use the fact that
$[0,1]$-acceptance automata are closed under monotone propositional
quality operators (Def.~\ref{def:propositionalQualityOperator}).  

\begin{prop}\label{prop:ClosedUnderIncOperator}
  Let $f \colon [0,1]^{k} \to [0,1]$ be monotone, and
$\A_{1},\dotsc,\A_{k}$ be
 $[0,1]$-acceptance automata
 over a common alphabet $\Sigma$. There is a $[0,1]$-acceptance
 automaton $f(\A_{1},\ldots,\A_{k})$ such that
 $\Lang\bigl(f(\A_{1},\ldots,\A_{k})\bigr) (w) =
 f\bigl(\Lang(\A_{1})(w),\ldots,\Lang(\A_{k})(w)\bigr)$ for each $w \in
 \Sigma^{\omega}$. \qed
\end{prop}

\begin{rem}
 Prop.~\ref{prop:ABAtoNBA} and~\ref{prop:ClosedUnderIncOperator} are
 essentially two separate constructions that deal with: the connectives
 $\land$ and $\lor$; and the other propositional quality operators, respectively. 
 One can alternatively think of $\land$ and $\lor$ as special cases of the latter 
 (Example~\ref{exam:propositionalQualityOperator}) and  use
 Prop.~\ref{prop:ClosedUnderIncOperator} altogether. 
This however results in a worse
complexity: the powerset-like construction in Prop.~\ref{prop:ABAtoNBA}
 exploits the commutativity, idempotency and associativity of $\land$ to 
 suppress the number of states, while such  is not done in the
 product-like
 construction in Prop.~\ref{prop:ClosedUnderIncOperator}.
\end{rem}

A generalization of $[0,1]$-acceptance automaton
  is naturally obtained by making transitions also
$[0,1]$-weighted. The result is called \emph{fuzzy automaton}
  and studied e.g.\ in~\cite{Rahonis05}.
  In Appendix~\ref{appendix:fuzzyAndZeroOne} we show that this generalization does not add expressivity. In
  fact we prove a more general result there, parametrizing $[0,1]$
  into a suitable semiring 
  $\mathbb{K}$.

\section{Near-Optimal Scheduling for $\LTLd{\D,\Fcalmc}$ }
\label{sec:nearOptimalSchedulerSynth}
In~\cite{AlmagorBK14,AlmagorBK13ExtendedPreprint} the threshold
model-checking problem for the logic $\LTLd{\D,\Fcal}$
is studied. 
In this paper, instead, we are interested in the following problem:  what
path of a given Kripke structure $\K$ is the best for a given
$\LTLd{\D,\Fcal}$ formula $\varphi$.

\begin{wrapfigure}[2]{r}{0pt}%
\begin{math}
\entrymodifiers={+[Fo]}
\def\labelstyle{\textstyle}
  \vcenter{\xymatrix@1@C-1em{
  { s_{0} } 
    \ar[r]
	\ar@(ur,ul) 
\save[]+<0cm,-.45cm>*{\lnot p}
 \restore
 &
   { s_{1} } %
     \ar[r]
\save[]+<0cm,-.45cm>*{p}\restore
 &
   { s_{2} } %
     \ar@(ur,ul)
\save[]+<0cm,-.45cm>*{\lnot p}\restore
}}
\end{math}
\end{wrapfigure}
 In general, however, there does not exist an optimal path $\xi_{0}$
 of $\K$, i.e.\ one that achieves $\sem{
 \xi_{0}, \varphi } = \sup_{\xi \in \pathrm(\K)} \sem{
 \xi, \varphi }$. 

\begin{exam}[optimality not achievable]\label{exam:optimalityNotAchievable}
Take a formula $\varphi = \Gdisc{\eta} \F p$
 and the Kripke structure $\K$ shown in the above.  
This example illustrates that the existence of an optimal path is 
not guaranteed in general: indeed, whereas $\sup_{\xi' \in \pathrm(\K)} \sem{
 \xi', \varphi } = 1$ in this example, there is no path $\xi$  that achieves $\sem{
 \xi, \varphi }=1$.

 More specifically: we first
 note that, in each 
 path $\xi$ of the Kripke structure,  $p$ is true
 at most once. The later the state $s_{1}$ occurs in a path $\xi$, the bigger the
 truth value $\sem{
 \xi, \varphi }$ is;  moreover the value $\sem{
 \xi, \varphi }$ tends to $1$ (since $\eta$ tends to $0$). 
 However there is no path $\xi$ that achieves exactly $\sem{
 \xi, \varphi }=1$: if $p$ is postponed indefinitely, no state in
 $\xi$ satisfies $p$, in which case $\F p$ is 
 everywhere false and hence $\sem{
 \xi, \varphi }=0$.
\end{exam}

 We thus strive for \emph{near}-optimality, allowing
a prescribed margin $\varepsilon$.


\begin{defi}
\label{def:nearOptimalPathSynthesis}
 The \emph{near-optimal scheduling} problem for $\LTLd{\D,\Fcal}$ is:
 given a Kripke structure $\K = (W,R,\lambda)$, an $\LTLd{\D,\Fcal}$ formula $\varphi$ and a positive real number $\varepsilon \in (0,1)$, to
find a path $\xi_{0} \in \pathrm(\K)$ such that $\sem{ \xi_{0}, \varphi } \geq \sup_{\xi \in \pathrm(\K)} \sem{ \xi, \varphi} - \varepsilon$. 
\end{defi}
Ultimately we will show that the problem in the above is decidable (Thm.~\ref{thm:main}), when all
the propositional quality operators are monotone and continuous
($\Fcal\subseteq\Fcalmc$).

We first note that, in the special case for 
$\LTLd{\D,\emptyset}$ (i.e.\ no propositional quality operators), there is a straightforward binary search algorithm
that 
relies on the (threshold) model-checking algorithm
in~\cite{AlmagorBK14} (Thm.~\ref{thm:thresholdModelCheckingDecidable}). Specifically, the binary search algorithm
repeatedly conducts threshold model-checking for:
 the threshold $v=\frac{1}{2}$ in the first round;  $v=\frac{1}{4}$ or
 $\frac{3}{4}$ in the second round,
       depending on the outcome of the first round;  then
       for $v=\frac{1}{8},\dotsc,\frac{6}{8}$ or $
       \frac{7}{8}$, depending on the outcome of the second round; and so on.  
Given a margin
       $\varepsilon\in(0,1)$,  this way, we need  $-\log
       \varepsilon$ rounds. 
This binary search algorithm is rather effective (see~\S{}\ref{sec:experiments}).

However the binary search algorithm does not work in presence of the
average operator $\oplus$, simply because the threshold model-checking
problem is undecidable (Thm.~\ref{thm:averageUndecidable}). Our main
contribution is a novel algorithm for near-optimal scheduling
that works even in this case (and more generally for the logic
$\LTLd{\D,\Fcalmc}$). 
Our algorithm  first translates a formula $\varphi$ and a margin
$\varepsilon\in (0,1)$ to an
alternating
$[0,1]$-acceptance automaton $\A_{\varphi,\varepsilon}$, which is further
turned into a $[0,1]$-acceptance automaton 
(Prop.~\ref{prop:ABAtoNBA}). The resulting automaton---after  taking the
product  
with $\K$---is amenable to optimal value search
(Lem.~\ref{lem:lassoOptimalityForQuantitativeAcceptAutom}), yielding
a solution to the original problem.



In the rest of the section we describe our algorithm.
We shall however  first
  restrict to the logic 
$\LTLd{\D,\emptyset}$ for the sake of presentation (although this basic fragment allows binary
search). After describing
the basic algorithm for $\LTLd{\D,\emptyset}$
in~\S{}\ref{subsec:algorithmWithoutPropositional},
in~\S{}\ref{subsec:algorithmWithPropositional} we explain how it can be
modified to accommodate propositional quality operators.

\subsection{Our Algorithm, When Restricted to
  $\LTLd{\D,\emptyset}$}\label{subsec:algorithmWithoutPropositional}
Our translation of  $\varphi$ and
$\varepsilon\in (0,1)$ to an automaton 
$\A_{\varphi,\varepsilon}$ is  an extension of the standard
translation from LTL formulas to alternating B\"{u}chi automata (e.g.\
in~\cite{Vardi96anautomata-theoretic}), with: 
\begin{itemize}
 \item  incorporation of
quantities---accumulation of discount factors, more specifically---by
	means of what we call \emph{discount sequences}; and
 \item cutting off those events which are  far in the future---the idea of \emph{event horizon} from~\cite{AlmagorBK14}.
\end{itemize}
The extension is not complicated on the conceptual level. Its
details need care, however, especially in handling negations and alternation of
greatest and least fixed points. 

As preparation, we recall some definitions and notations from~\cite{AlmagorBK14}. 
\begin{defi}[$\eta^{+k}$, $\mathit{xcl} (\varphi)$~\cite{AlmagorBK14}]
 Let $\eta : \N
 \rightarrow [0,1]$ be a discounting function. We define a discounting
 function $\eta^{+k} : \N \rightarrow [0,1]$ by
 \begin{math}
   \eta^{+k} (i) =
  \eta(i+k)
 \end{math} 
for each $k\in\N$.

For an $\LTLdD$ formula $\varphi$, the
 \emph{extended closure} $\mathit{xcl} (\varphi)$ of
 $\varphi$~\cite{AlmagorBK14} is defined by 
\begin{equation}\label{eq:defXcl}
 \mathit{xcl}(\varphi) \;=\; \Sub(\varphi) \cup \{ \varphi_{1} \U_{\eta^{+k}}
 \varphi_{2} \mid k \in \N, \varphi_{1} \U_{\eta} \varphi_{2} \in
 \Sub(\varphi) \}\enspace,
\end{equation}
where $\Sub(\varphi)$ denotes the set of subformulas of $\varphi$.
\end{defi}

\subsubsection{Discounting Sequences}\label{subsubsec:discountSeq}
 We go on to technical details.
 In the alternating $[0,1]$-acceptance automaton
 $\A_{\varphi,\varepsilon}$ that we shall construct, a state is a pair 
 $(\psi,\vec{d})$ of a formula $\psi$ and a \emph{discount sequence} 
 $\vec{d}\in[0,1]^{+}$.
\begin{defi}[discount sequence]\label{def:discountSequence}
 A \emph{discount sequence} is a sequence $\vec{d}=d_{1}d_{2}\dotsc d_{n}\in
 [0,1]^{+}$ of real numbers with a nonzero length ($d_{i}\in [0,1]$ for
 each $i$). 
\end{defi}
The notion of discount sequence is a quantitative extension of that of
\emph{priority} in parity automata. Specifically, the length $n$ of a
discount sequence $\vec{d}=d_{1}d_{2}\dotsc d_{n}$ corresponds to a
priority---i.e.\ the alternation depth of greatest and least fixed
points.
Each real number $d_{i}$ in the sequence, in turn, stands for the accumulated 
discount factor in each level of fixed-point alternation. 
For example,  the formula 
 $
\Fdisc{\expo_{\frac{1}{2}}}
\Gdisc{\expo_{\frac{2}{3}}}
\Fdisc{\expo_{\frac{3}{4}}}
p
$  will induce a discount sequence
$
(\frac{1}{2})^{n_{1}},
(\frac{2}{3})^{n_{2}},
(\frac{3}{4})^{n_{3}}
$
 of
 length 3---where $n_{1},n_{2}$ and $n_{3}$ are the numbers of steps for
 which 
the three discounting temporal operators
$\Fdisc{\eta_{1}}$, 
$\Gdisc{\eta_{2}}$ and 
$\Fdisc{\eta_{3}}$
 ``have waited,'' respectively.

We use three operators $\odot, \concatseq, \boxtimes$ that act on discount
sequences; the 
intuitions are as follows. The first two are for accumulating discount
factors: we use $\odot$
in case there is no alternation of greatest and least fixed points; and
we use
$\concatseq$ in case there is. Examples are:
\begin{displaymath}\textstyle
\begin{array}{lcl}
 \bigl(\,
  (\frac{1}{2})^{2},
 (\frac{2}{3})^{3},
 \frac{3}{4}
 \,\bigr)
 \;\odot\; 
 \color{blue}
\frac{4}{5}
 \color{black}
 &=&
 \bigl(\,
 (\frac{1}{2})^{2},
 (\frac{2}{3})^{3},
 \frac{3}{4}\cdot
 \color{blue}
 \frac{4}{5}
 \color{black}
 \,\bigr)
 \;=\;
 \bigl(\,
 (\frac{1}{2})^{2},
 (\frac{2}{3})^{3},
 \frac{3}{5}
 \,\bigr)
 \quad\text{and}
 \\
 \bigl(\,(\frac{1}{2})^{2},
 (\frac{2}{3})^{3},
 \frac{3}{4}
 \,\bigr)\;\concatseq\; 
 \color{blue}
\frac{4}{5}
 \color{black}
 &=&
\bigl(\,
 (\frac{1}{2})^{2},
 (\frac{2}{3})^{3},
 \frac{3}{4},\,
 \color{blue}
 \frac{4}{5}
 \color{black}
 \,\bigr)\enspace.
\end{array}
\end{displaymath}
Note that in the former the length is preserved, while in the latter the
sequence gets longer by one.
\begin{defi}[$\vec{d}\odot d'$,  $\vec{d}\concatseq d'$]\label{def:operationsOnDiscountSeqOdotAndCdot}
 The  operator $\odot$ 
takes a discount sequence $\vec{d}$
 and a discount factor $d'\in[0,1]$ as arguments, and multiplies
 the last element of  $\vec{d}$  by $d'$. That is,
\begin{displaymath}
 (d_{1}d_{2}\dotsc d_{n})\odot d'
 \;=\;
 d_{1}d_{2}\dotsc d_{n-1} (d_{n}\cdot d')
\quad\in [0,1]^{+}\enspace.
\end{displaymath}
The operator $\concatseq$ is simply the concatenation operator: given 
$\vec{d}=d_{1}d_{2}\dotsc d_{n}$
 and  $d'\in[0,1]$, the sequence $\vec{d}\concatseq d'$ is $d_{1}d_{2}\dotsc
 d_{n} d'$ of length $n+1$.
\end{defi}

We use the operator $\boxtimes$ in $\vec{d}\boxtimes v$ to let a
discount sequence $\vec{d}$ act on a truth value $v\in[0,1]$. 
\begin{defi}[$\vec{d}\boxtimes v$]\label{def:operationsOnDiscountSeq}
 The operator $\boxtimes$ takes $\vec{d}\in[0,1]^{+}$
 and $v\in[0,1]$ as
 arguments. The value $\vec{d}\boxtimes v\in[0,1]$ is defined inductively by:
 \begin{align}\label{eq:inductiveDefOfAction}
  &d\boxtimes v
  \;=\; d v\enspace, 
  \qquad
  \vec{d}d'\boxtimes v
  \;=\;
  \vec{d}\boxtimes (1-d' v)\enspace. \quad\text{Explicitly:}
 \\
\label{eq:actionExplicitly}
&\footnotesize
\begin{array}{l}
   (d_{1}d_{2}\dotsc d_{n})\boxtimes v
 \;=\;
 d_{1}\,- d_{1}d_{2} \,+ d_{1}d_{2}d_{3} \,- \cdots 
 \,+ (-1)^{n} d_{1}d_{2}\dotsc d_{n-1}
 \,+ (-1)^{n+1} d_{1}d_{2}\dotsc d_{n} v\,.
\end{array}
 \end{align}
\end{defi}

The intuition behind the action $\vec{d}\boxtimes v$ is most visible
in~(\ref{eq:inductiveDefOfAction}), where $dv$ and $d'v$ denote
multiplication of real numbers. Given a discount sequence
$\vec{d}d'$:  1) we apply the final discount factor $d'$ to 
the truth value $v$, obtaining $d'v$; 2) the alternation between 
greatest and least fixed points is taken into account, by taking
the negation $1-d' v$ (cf.\ Def.~\ref{def:LTLdDSemantics}); and
 3)
we apply the remaining sequence $\vec{d}$ inductively and obtain $\vec{d}\boxtimes
(1-d' v)$. An example is
\begin{math}\textstyle
\bigl(
\frac{3}{4},\,
 \frac{1}{3},\,
 \frac{2}{5}
\bigr)\boxtimes 1
=
\bigl(
\frac{3}{4},\,
 \frac{1}{3}
\bigr)\boxtimes 
\bigl(1-\frac{2}{5}\cdot 1\bigr)
=
\bigl(
\frac{3}{4},\,
 \frac{1}{3}
\bigr)\boxtimes \frac{3}{5}
=
\bigl(
\frac{3}{4}
\bigr)\boxtimes
\bigl(
 1-\frac{1}{3}\cdot\frac{3}{5}
\bigr)
=
\bigl(
\frac{3}{4}
\bigr)\boxtimes \frac{4}{5}
=
\frac{3}{5}
\end{math}.

 The following  relationship between $\odot$
and $\boxtimes$ is easily seen to hold:
\begin{equation}\label{eq:odotAndBoxtimes}
 (\vec{d}\odot d')\boxtimes v
 \;=\; \vec{d}\boxtimes (d'\cdot v)\enspace.
\end{equation}

The three operators $\odot, \concatseq, \boxtimes$ defined in the above 
will be used shortly, in the construction of the automaton
$\A_{\varphi,\varepsilon}$. Their roles are briefly discussed after Def.~\ref{def:ABAforLTL}.

\subsubsection{Construction of $\A_{\varphi,\varepsilon}$}
 We  describe the construction of $\A_{\varphi,\varepsilon}$,
for a formula $\varphi$ of $\LTLd{\D,\emptyset}$ and a margin
$\varepsilon$.  We subsequently discuss  ideas
behind it, comparing the definition with other known constructions.

We first define 
$\Ap_{\varphi,\varepsilon}$ that is infinite-state, and obtain 
$\A_{\varphi,\varepsilon}$ as the reachable part. The latter will be shown to
be finite-state (Lem.~\ref{lem:AphiEpsIsFiniteState}).

\begin{defi}[the automata $\Ap_{\varphi,\varepsilon}, \A_{\varphi,\varepsilon}$]\label{def:ABAforLTL}
  Let $\varphi$ be an $\LTLd{\D,\emptyset}$  formula and $\varepsilon\in (0,1)$. We define an alternating {$[0,1]$}-acceptance automaton
 $\Ap_{\varphi, \varepsilon} = (\mathcal{P}(\mathit{AP}),Q,I,\delta,F)$ as
 follows. 
 Its state space $Q$ is $\mathit{xcl}(\varphi) \times
 [0,1]^{+}$; hence a state is a pair $(\psi,\vec{d})$ of a formula and a discount sequence.
 The transition function $\delta\colon Q\times
 \mathcal{P}(\mathit{AP})\to \mathcal{B}^{+}(Q\cup [0,1])$ is defined as 
in Table~\ref{table:transitionFuncOfAp}, where we let
 $\vec{d} = d_{1} d_{2} \ldots d_{n}\in [0,1]^{+}$
 and  $\sigma \in \mathcal{P}(\mathit{AP})$.
\begin{table}[tbp]\centering
\small
  \begin{align}
   \delta \bigl((\true,\vec{d}),\sigma\bigr) &=
	  \vec{d}\boxtimes 1 
   \nonumber
   \\
\delta \bigl((p,\vec{d}),\sigma\bigr) &= 
      \begin{cases}
        \vec{d}\boxtimes 1
       & \mbox{if } p \in \sigma, \\
        \vec{d}\boxtimes 0
& \mbox{otherwise.}
      \end{cases}
   \nonumber
   \\
\delta \bigl((\lnot \psi,\vec{d}),\sigma\bigr) &= 
	      \delta \bigl((\psi,\vec{d}\concatseq 1),\sigma\bigr)
   \label{eq:201507281451}
\\
\delta \bigl((\psi_{1} \land \psi_{2},\vec{d}),\sigma\bigr) &= 
      \begin{cases}
        {\delta \bigl((\psi_{1},\vec{d}),\sigma\bigr)} \land {\delta
       \bigl((\psi_{2},\vec{d}),\sigma\bigr)} & \mbox{if $|\vec{d}|$ is odd},\\
        {\delta \bigl((\psi_{1},\vec{d}),\sigma\bigr)} \lor {\delta \bigl((\psi_{2},\vec{d}),\sigma\bigr)} & \mbox{otherwise.}
      \end{cases}
   \nonumber
   \\
   \delta \bigl((\X \psi,\vec{d}),\sigma\bigr) &= (\psi,\vec{d})
   \nonumber
   \\
\delta \bigl((\psi_{1} \U \psi_{2},\vec{d}),\sigma\bigr) &= 
      \begin{cases}
        {\delta \bigl((\psi_{2},\vec{d}),\sigma\bigr)} \lor
       \Bigl(\,{\delta \bigl((\psi_{1},\vec{d}),\sigma\bigr)} \land
       {(\psi_{1} \U \psi_{2},\vec{d})}\,\Bigr) & 
\mbox{if $|\vec{d}|$ is odd},\\
        {\delta \bigl((\psi_{2},\vec{d}),\sigma\bigr)} \land \Bigl(\,{\delta \bigl((\psi_{1},\vec{d}),\sigma\bigr)} \lor {(\psi_{1} \U \psi_{2},\vec{d})}\,\Bigr) & \mbox{otherwise.}
      \end{cases}
   \nonumber
  \end{align}
\begin{minipage}{\textwidth}\small
 For $\delta\bigl((\psi_{1} \U_{\eta} \psi_{2},
	 \vec{d}),\sigma\bigr)$ we make cases. Let 
	 $\vec{d}=d_{1}\dotsc d_{n}$. If
	 $\eta(0)\cdot \prod_{i = 1}^{n} d_{i} \le \varepsilon$:
	  \begin{equation}\label{eq:defDeltaBeyondEventHorizon}
	   \small
	    \delta \bigl((\psi_{1} \U_{\eta} \psi_{2},\vec{d}),\sigma\bigr) = 
	    \begin{cases}
	      \vec{d} \boxtimes 0 & 
 \mbox{if $|\vec{d}|$ is odd},\\
		 \vec{d}\boxtimes \eta(0) & \mbox{otherwise;}
	    \end{cases} 
	  \end{equation}
	  otherwise, i.e.\ if 	 $\eta(0)\cdot \prod_{i = 1}^{n} d_{i} > \varepsilon$:
 \begin{equation}\label{eq:defDeltaWithinEventHorizon}
 \small
 	  \begin{aligned}
	   &\delta \bigl((\psi_{1} \U_{\eta} \psi_{2},\vec{d}),\sigma\bigr) =
        \begin{cases}
          {\delta \bigl((\psi_{2},\vec{d}\odot \eta(0)),\sigma\bigr)}
	 \lor \Bigl({\delta \bigl((\psi_{1},\vec{d}\odot
	 \eta(0)),\sigma\bigr)} \land {(\psi_{1} \U_{\eta^{+1}}
	 \psi_{2},\vec{d})}\Bigr) & 
 \mbox{if $|\vec{d}|$ is odd},\\
		  {\delta \bigl((\psi_{2},\vec{d}\odot \eta(0)),\sigma\bigr)} \land \Bigl({\delta \bigl((\psi_{1},\vec{d}\odot \eta(0)),\sigma\bigr)} \lor {(\psi_{1} \U_{\eta^{+1}} \psi_{2},\vec{d})}\Bigr) & \mbox{otherwise.}
        \end{cases}
	  \end{aligned}
\end{equation}  
\end{minipage}
 \caption{Transition function $\delta$ of $\Ap_{\varphi, \varepsilon}$}
 \label{table:transitionFuncOfAp}
\end{table}

  The set $I$ of the initial states of $\Ap_{\varphi, \varepsilon}$ is $\{ (\varphi,1) \}$. The acceptance function $F$ is
  \begin{equation}\label{eq:acceptanceFunc}
  \small	F (\psi,\vec{d}) = 
    \begin{cases}
      1 & \mbox{if } {\psi = \psi_{1} \U \psi_{2}} \mbox{ and $|\vec{d}|$ is even}
\\
	  0 & \mbox{otherwise.}
    \end{cases}
  \end{equation}

 The alternating $[0,1]$-acceptance automaton $\A_{\varphi,\varepsilon}$
 is defined to be the restriction of  $\Ap_{\varphi,\varepsilon}$ to the states that
 are
 reachable from the initial state $(\varphi, 1)$. 
\end{defi}
\auxproof{To slightly optimize the automaton one can take
 $\delta \bigl((\lnot \psi,\vec{d}),\sigma\bigr) = 
      \begin{cases}
        \delta \bigl((\psi,\vec{d}'),\sigma\bigr) & \mbox{if } {\vec{d} = \vec{d}'1} \land {\vec{d}' \in [0,1]^{+}} \\
        \delta ((\psi,\vec{d}1),\sigma) & \mbox{otherwise.}
      \end{cases}$.
}
 Examples of $\A_{\varphi,\varepsilon}$ are
in Fig.~\ref{fig:exampleTranslation}--\ref{fig:exampleTranslation2},
where
$(\varphi,\varepsilon)=
(
\Gdisc{\expo_{\frac{1}{2}}}
\Fdisc{\expo_{\frac{2}{5}}}
p
,
\frac{1}{3}
)
$
and
$(
\Fdisc{\expo_{\frac{1}{2}}}
\G
p
,
\frac{1}{3}
)$.
There a discount sequence
$d_{1}\dotsc d_{n}$ is denoted by $\langle d_{1},\dotsc, d_{n}
\rangle$ for readability.

\begin{figure}[tbp]\centering
\begin{minipage}[]{.5\textwidth}
 \includegraphics[width=\textwidth]{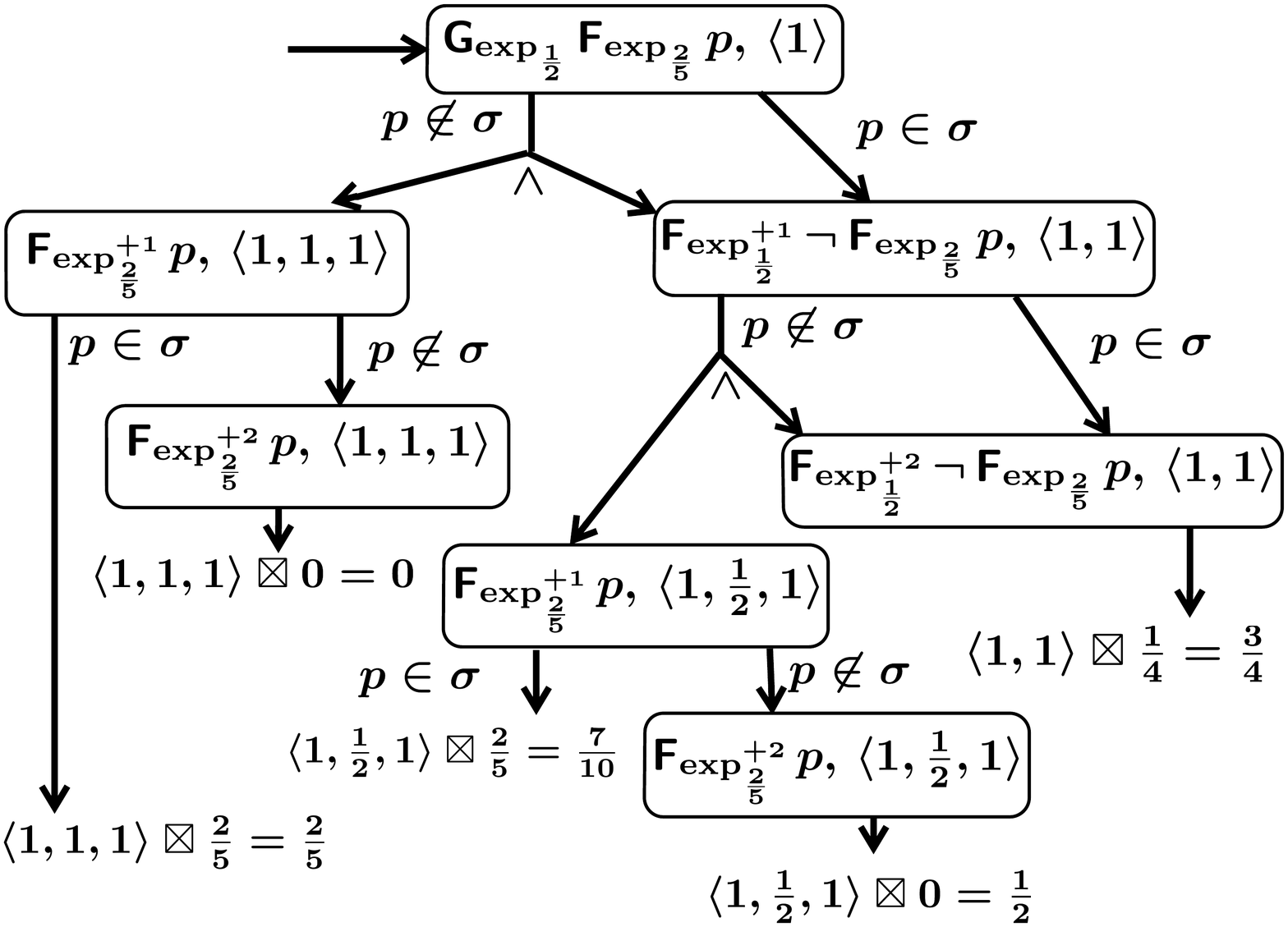}
 \caption{The automaton 
 $\A_{\varphi,\varepsilon}$ for\\ $\varphi
 =
 \Gdisc{\expo_{\frac{1}{2}}}
 \Fdisc{\expo_{\frac{2}{5}}}
 p
 $ and $\varepsilon = \frac{1}{3}$
 }
 \label{fig:exampleTranslation}
\end{minipage}
\begin{minipage}[]{.48\textwidth}
\includegraphics[width=\textwidth,clip,trim=0cm 4.5cm 3cm 0cm]{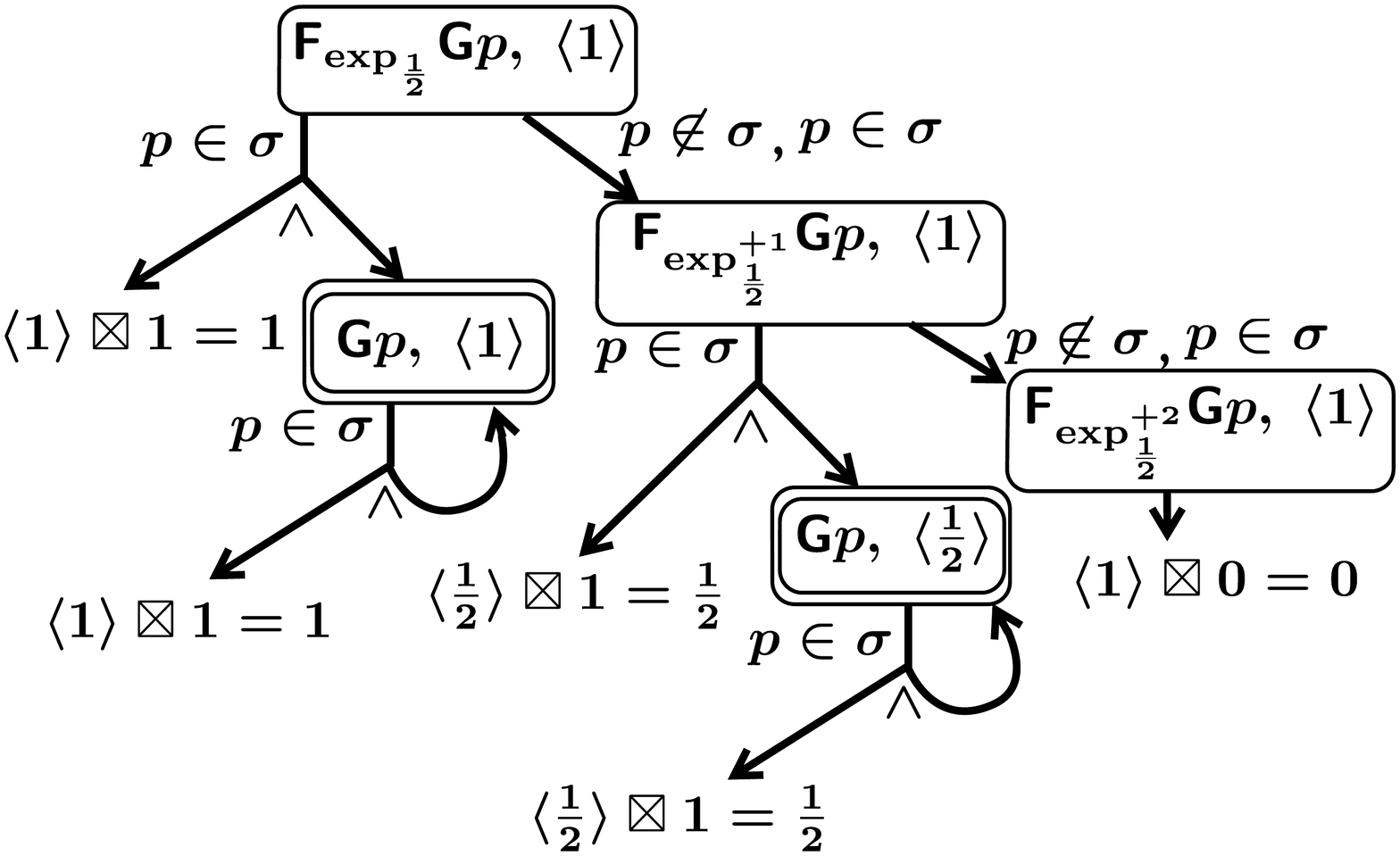}
\caption{The automaton 
$\A_{\varphi,\varepsilon}$ for\\ $\varphi
=
\Fdisc{\expo_{\frac{1}{2}}}
\G
p
$ and $\varepsilon = \frac{1}{3}$. The double-lined nodes  have
 the acceptance value $1$.
}
\label{fig:exampleTranslation2}
\end{minipage}
\end{figure}

Some remarks on Def.~\ref{def:ABAforLTL} are in order.

\noindent
\textbf{In Absence of Discounting (Sanity Check)}\quad 
If the formula $\varphi$ contains no discounting operator $\U_{\eta}$,
then the construction essentially coincides the usual one
in~\cite{Vardi96anautomata-theoretic} that translates a (usual) LTL
formula to an alternating B\"uchi automaton. To see it, 
recall that the length $|\vec{d}|$ of a discount sequence plays the
role of a priority in parity automata (\S{}\ref{subsubsec:discountSeq}).
Therefore in the first case of~(\ref{eq:acceptanceFunc}),  $|\vec{d}|$ being
even means that we are in fact dealing with a greatest fixed point. This
makes the state accepting (in the B\"{u}chi sense), much like in~\cite{Vardi96anautomata-theoretic}.

\noindent 
\textbf{$\A_{\varphi,\varepsilon}$ is Quantitative} \quad
The acceptance values of the states of $\A_{\varphi,\varepsilon}$
       are Boolean (see~(\ref{eq:acceptanceFunc})). Nevertheless the automaton
       is quantitative, in that non-Boolean values from $[0,1]$
       appear
       as atomic propositions in the range $\mathcal{B}^{+}(Q\cup
       [0,1])$
       of the transition $\delta$ (they occur at the leaves in Fig.~\ref{fig:exampleTranslation}--\ref{fig:exampleTranslation2}). Once we transform
       $\A_{\varphi,\varepsilon}$ to a
       non-alternating
       automaton (Prop.~\ref{prop:ABAtoNBA}), these non-Boolean
       propositional values
       give rise to non-Boolean acceptance values.

\noindent 
\textbf{Event Horizon} \quad 
A fundamental idea from~\cite{AlmagorBK14} is the following. A discounting operator, in presence of
a threshold (in~\cite{AlmagorBK14}) or a nonzero margin (here), allows an exact
representation by a (finitary) formula without a fixed point
operator. The latter means, for example:
\begin{align}\label{eq:finReprOfDiscountOpr1}
\textstyle  \sem{\pi,\Fdisc{\expo_{\frac{1}{2}}}\varphi}\ge \frac{1}{4}
 &\quad\Longleftrightarrow
 \quad\pi\models \varphi\lor \X\varphi\lor \X\X\varphi\enspace,
  \quad\text{and}
\\
\label{eq:finReprOfDiscountOpr2}
\textstyle
  \sem{\pi,\Gdisc{\expo_{\frac{1}{2}}}\varphi}\ge \frac{3}{4}
 &\quad\Longleftrightarrow\quad
\pi\models \varphi\land \X\varphi\land \X\X\varphi\enspace,
\end{align}
and so on. Note that in~(\ref{eq:finReprOfDiscountOpr1}), whatever happens after two time units
has contributions less than $(\frac{1}{2})^{2}=\frac{1}{4}$ and
therefore never enough to make up the threshold. The example~(\ref{eq:finReprOfDiscountOpr2})
is similar, with events in the future having only negligible negative
contributions. In other words:  fixed point operators with discounting have
an \emph{event horizon}---in the above examples~(\ref{eq:finReprOfDiscountOpr1}--\ref{eq:finReprOfDiscountOpr2}) it lies between $t=2$  and
$3$---nothing beyond which  matters.

This idea of event horizon is used in the distinction
between~(\ref{eq:defDeltaBeyondEventHorizon})
and~(\ref{eq:defDeltaWithinEventHorizon}). The value $\eta(0)\cdot
\prod_{i = 1}^{n} d_{i}$ is, as we shall see, the greatest contribution
to a truth value that the events henceforth potentially have.  In case it is smaller
than the margin $\varepsilon$ we can safely ignore the \emph{positive}
contribution henceforth and take
 the smallest possible truth value
$0$---much like the disjunct
$\X^{3}\varphi\lor\X^{4}\varphi\lor\cdots$ is truncated
in~(\ref{eq:finReprOfDiscountOpr1}). This is what is done in the first case
in~(\ref{eq:defDeltaBeyondEventHorizon}). 
The second case 
in~(\ref{eq:defDeltaBeyondEventHorizon}) is about a greatest fixed point
and
we truncate the \emph{negative} contributions of the events beyond the event
horizon---this is much like the obligation
$\X^{3}\varphi\land\X^{4}\varphi\land\cdots$ is lifted
in~(\ref{eq:finReprOfDiscountOpr2}). In this case we use the greatest truth
value possible, namely $\eta(0)$. This is what is done
in~(\ref{eq:defDeltaBeyondEventHorizon}).


\noindent 
\textbf{Use of Discount Sequences} \quad
Discount sequences $\vec{d}$ are used for two purposes. Firstly, as
we already described,
its length
$|\vec{d}|$ indicates the alternation between positive and negative views on a
formula---observe that a discount sequence gets longer in~(\ref{eq:201507281451}).  Consequently many clauses in the definition of $\delta$
distinguish
cases according to the parity of $|\vec{d}|$. Secondly it records all
the discount factors that have been
encountered. See~(\ref{eq:defDeltaWithinEventHorizon}), where the last
element of $\vec{d}$ is multiplied by the newly encountered factor
$\eta(0)$ and updated to $\vec{d}\odot \eta(0)$. 
Such accumulation $\vec{d}$ of 
discount factors acts on a truth value via the $\boxtimes$  operator, like
in~(\ref{eq:defDeltaBeyondEventHorizon}) and in the definition of $\delta
\bigl((\true,\vec{d}),\sigma\bigr)$.

\begin{lem}\label{lem:AphiEpsIsFiniteState}
 The automaton $\A_{\varphi,\varepsilon}$ has only finitely many states.
 \qed
\end{lem}

The following ``correctness lemma'' claims that $\A_{\varphi,\varepsilon}$ 
conducts the expected task. 
\begin{lem}\label{lem:correctnessOfA}
  Let $\varphi$ be an $\LTLd{\D,\emptyset}$ formula and $\varepsilon\in(0,1)$ be a
 positive real number. 
For each computation $\pi \in (\mathcal{P}(\mathit{AP}))^{\omega}$, we
 have
\begin{math}
     \sem{ \pi, \varphi } - \varepsilon \;\leq\; \Lang(\A_{\varphi,
   \varepsilon})(\pi)
 \; \leq\; \sem{ \pi, \varphi }
\end{math}\qed
\end{lem}

\subsubsection{The Algorithm}\label{subsubsec:theAlgorithmWithoutPropositional}
 After the construction of
 $\A_{\varphi,\varepsilon}$, 
the algorithm proceeds in the following manner.
We first translate $\A_{\varphi,\varepsilon}$ to a (non-alternating)
{$[0,1]$}-acceptance automaton (relying on Prop.~\ref{prop:ABAtoNBA}).

\begin{cor}\label{cor:NBAforLTL}
  Let $\varphi$ be an $\LTLd{\D,\emptyset}$  formula and $\varepsilon \in (0,1)$ be a
 positive real number. There exists a (non-alternating) {$[0,1]$}-acceptance automaton $\Ana_{\varphi,\varepsilon}$ such that 
  \begin{math}
    \sem{ \pi, \varphi } - \varepsilon
 \leq \Lang(\Ana_{\varphi,\varepsilon})(\pi) 
\leq \sem{ \pi, \varphi }
  \end{math}
  for each computation $\pi \in (\mathcal{P}(\mathit{AP}))^{\omega}$.
\qed
\end{cor}

Towards the solution of the near-optimal 
scheduling problem (Def.~\ref{def:nearOptimalPathSynthesis}), we
construct the \emph{product} of
$\Ana_{\varphi,\varepsilon}$
in
Cor.~\ref{cor:NBAforLTL}
and the given Kripke structure $\K$. Since transitions of
{$[0,1]$}-acceptance automata are nondeterministic, this product can be
defined just as usual.
\begin{defi}\label{def:productAutomata}
  Let $\A = (\mathcal{P}(\mathit{AP}),Q,I,\delta,F)$ be a
 {$[0,1]$}-acceptance automaton and $\K = (W,R,\lambda)$ be a Kripke
 structure.  Their \emph{product} $\A \times \K$ 
is a 
 {$[0,1]$}-acceptance automaton
 $(1,Q',I',\delta',F')$---over a singleton alphabet
 $1=\{\bullet\}$---defined by:
 $Q'= Q$; $I'= I \times W$; $\delta' \bigl(\,(q,s),\,\bullet\,\bigr)
=     \bigl\{\,
   (q',s')
  \;\bigl|\bigr.\;
   q'\in \delta(q,\lambda(s)),
   (s,s') \in R
   \,\bigr\}$; and $F'(q,s) = F(q)$.
\end{defi}
%
\begin{lem}\label{lem:fromOptimalInATimesKToOptimalInA}
 Let 
\begin{math}
 (q_{0}, s_{0})\,\bullet\,
 (q_{1}, s_{1})\,\bullet\,
 \dotsc
\end{math}
be an optimal run of the automaton $\A\times\K$ (that necessarily
 exists by Lem.~\ref{lem:lassoOptimalityForQuantitativeAcceptAutom}).
The path $s_{0}s_{1}\dotsc\in\pathrm(\K)$ realizes the optimal
       value of $\A$, that is,
\begin{math}
        \Lang(\A)\bigl(\lambda(s_{0})\lambda(s_{1})\dotsc\bigr) =
 \max_{\xi\in\pathrm(\K)}\Lang(\A)\bigl(\lambda(\xi)\bigr)
\end{math}.
\qed
\end{lem}


\begin{thm}[optimal scheduling for $\LTLd{\D,\emptyset}$]\label{thm:mainWithoutPropositional}
 Assume the setting of  Def.~\ref{def:nearOptimalPathSynthesis}, 
 and that $\Fcal=\emptyset$ (i.e.\ the formula $\varphi$ contains no
 propositional quality operators).
Let 
\begin{math}
 (q_{0}, s_{0})\,\bullet\,
 (q_{1}, s_{1})\,\bullet\,
 \dotsc
\end{math}
 be an optimal run (computed by
 Lem.~\ref{lem:lassoOptimalityForQuantitativeAcceptAutom}) for the
 $[0,1]$-acceptance automaton $\Ana_{\varphi,\varepsilon}\times\K$
 constructed as in Def.~\ref{def:ABAforLTL}, Cor.~\ref{cor:NBAforLTL}
 and Def.~\ref{def:productAutomata}. Then the path
 $s_{0}s_{1}\dotsc\in\pathrm(\K)$ is a solution to the near-optimal scheduling problem (Def.~\ref{def:nearOptimalPathSynthesis}). 

Moreover,
 the solution  $s_{0}s_{1}\dotsc$ can be chosen to be ultimately
 periodic.
\qed
\end{thm}



\auxproof{ We can search for a near-worst scheduler, too, by seeking for an optimal
 path for $\lnot \varphi$.
}




\subsection{Our General Algorithm for
  $\LTLd{\D,\Fcalmc}$}\label{subsec:algorithmWithPropositional}
Our general algorithm works in the setting of $\LTLd{\D,\Fcalmc}$---i.e.\ in
the presence of monotone and continuous propositional quality operators like $\oplus$---where 
threshold model checking is potentially undecidable~\cite{AlmagorBK14}
and therefore the binary-search algorithm (described after
Def.~\ref{def:nearOptimalPathSynthesis}) may not work. 

The general algorithm is a (rather straightforward)
adaptation of the one we described for
$\LTLd{\D,\emptyset}$
(\S{}\ref{subsec:algorithmWithoutPropositional}). Here
 we construct the 
  alternating $[0,1]$-acceptance automaton $\A_{\varphi,\varepsilon}$
  \emph{inductively} on the construction on the formula $\varphi$:
       \begin{itemize}
	\item When the outermost connective is other than a
	      propositional quality operator, the construction is 
	      much like in Def.~\ref{def:ABAforLTL}.
	\item When the outermost connective is  a
	      propositional quality operator, we rely on
	      Prop.~\ref{prop:ClosedUnderIncOperator}. 
       \end{itemize}
 The rest of the algorithm (i.e.\ the part described
 in~\S{}\ref{subsubsec:theAlgorithmWithoutPropositional}) remains
 unchanged. 
 An extensive description of the details of the construction
 is deferred to Appendix~\ref{appendix:generalAlgorithmDetails}.

\begin{thm}[main theorem, optimal scheduling for $\LTLd{\D,\Fcalmc}$]\label{thm:main}
 In the setting of Def.~\ref{def:nearOptimalPathSynthesis}, 
 assume that $\Fcal\subseteq\Fcalmc$ (i.e.\ all the propositional
 quality operators in $\varphi$ are monotone and continuous).
 Then the near-optimal scheduling problem is decidable.
\qed
\end{thm}

\subsection{Complexity}
\label{subsec:complexity}
The two parameters $\D$ and $\Fcal$ in $\LTLd{\D,\Fcal}$---i.e.\ discounting functions
(Def.~\ref{def:discountFunction}) and propositional quality operators
(Def.~\ref{def:propositionalQualityOperator})---are both relevant to
 the  complexity of our algorithm. Formulating a
complexity result is hard when these parameters are left open.  We therefore
restrict to:
\begin{itemize}
 \item exponential discounting functions
(Def.~\ref{def:discountFunction}), i.e.\ $\D=\Dexp = \{ \expo_{\lambda} \mid \lambda
\in (0,1) \cap \Q \}$, as is
       done
in~\cite{AlmagorBK14}; and
 \item the average operator $\oplus$, i.e.\ $\Fcal=\{\oplus\}$.
\end{itemize} 
%
We use the definition
 $|\langle \varphi \rangle|$ 
of the size of a formula $\varphi$, which is from~\cite{AlmagorBK14}: it
reflects the description length of $\lambda\in\Q$ that appears in
 discounting
functions, as well as the length of $\varphi$
as an expression.


\begin{prop}[size of  $\A_{\varphi,\varepsilon}$]\label{prop:sizeOfAltAutom}
  Let $\varphi$ be an $\LTLd{\Dexp,\{\oplus\}}$ formula and $\varepsilon\in(0,1)\cap\Q$
 be a positive rational number.  The size of the state space of
 the alternating $[0,1]$-acceptance automaton
 $\A_{\varphi,\varepsilon}$ is singly exponential in $|\langle \varphi
 \rangle|$ and in the length of the description of $\varepsilon$. \qed
\end{prop}
\begin{thm}[complexity for $\LTLd{\Dexp,\{\oplus\}}$] \label{thm:mainComplexity}
The near-optimal scheduling problem for $\LTLd{\Dexp,\{\oplus\}}$ is: in
 EXPSPACE in $|\langle\varphi\rangle|$ and in the description length of
 $\varepsilon$; and in NLOGSPACE in the size of $\K$. 
\qed
\end{thm}

In case of absence of propositional quality operators
(i.e.\ $\LTLd{\Dexp,\emptyset}$), we can further 
optimize the complexity by using a heuristic and  avoiding the exponential blowup  from 
$\A_{\varphi,\varepsilon}$ to  $\Ana_{\varphi,\varepsilon}$. This yields 
the following complexity result, which is also achievable by the
binary-search
algorithm.

\begin{thm}[complexity for $\LTLd{\Dexp,\emptyset}$] \label{thm:mainComplexityWithoutPropositional}
The near-optimal scheduling problem for $\LTLd{\Dexp,\emptyset}$ is: in
 PSPACE in $|\langle\varphi\rangle|$ and in the description length of
 $\varepsilon$; and in NLOGSPACE in the size of $\K$. 
\qed
\end{thm}

\section{Experiments}\label{sec:experiments}
We implemented our algorithm
in~\S{}\ref{sec:nearOptimalSchedulerSynth} that solves the near-optimal
scheduling for $\LTLd{\Dexp,\{\oplus\}}$. The implementation is
in OCaml.
The following  experiments were on
a MacBook Pro laptop
 with a Core i5 processor (2.7 GHz) and 16 GB RAM. 
\begin{table}[tb]
\scriptsize
\begingroup
\renewcommand{\arraystretch}{1.5}
\begin{tabular}{c|c|c|c|c|c|c}
 & \multicolumn{2}{c|}{$\varepsilon = \frac{1}{10}$} & \multicolumn{2}{c|}{$\varepsilon = \frac{1}{50}$} & \multicolumn{2}{c}{$\varepsilon = \frac{1}{100}$} \\ \cline{2-7}
formula $\varphi$ \textbackslash\; \#(states) & \multicolumn{1}{c|}{$\A_{\varphi,\varepsilon}$} & \multicolumn{1}{c|}{$\Ana_{\varphi,\varepsilon}$} & \multicolumn{1}{c|}{$\A_{\varphi,\varepsilon}$} & \multicolumn{1}{c|}{$\Ana_{\varphi,\varepsilon}$} & \multicolumn{1}{c|}{$\A_{\varphi,\varepsilon}$} & \multicolumn{1}{c}{$\Ana_{\varphi,\varepsilon}$} 
 \\ \hhline{=|=|=|=|=|=|=}
${\F_{\expo_{\frac{1}{2}}} p_{1}} $
 & 5 & 10 & 7 & 14 & 8 & 16 \\ \hline
${\F_{\expo_{\frac{99}{100}}} p_{1}} $
 & 231 & 462 & 391 & 782 & 460 & 920 \\ \hline
${\F_{\expo_{\frac{1}{2}}}\G_{\expo_{\frac{1}{2}}} p_{1}}$ & 15 & 36 & 28 & 85 & 36 & 121
\\ \hline\hline
${\F_{\expo_{\frac{1}{2}}} p_{1}} \oplus 
 {\F_{\expo_{\frac{1}{2}}} p_{2}}$
 & 33 & 128 & 61 & 1859 & 78 & 7421 \\ \hline
${\F_{\expo_{\frac{1}{2}}} p_{1}} \oplus 
 {\G_{\expo_{\frac{1}{2}}} p_{2}}$
 & 29 & 272 & 55 & 6659 & 71 & 32703 \\ \hline
${\F_{\expo_{\frac{3}{5}}} p_{1}} \oplus 
 {\F_{\expo_{\frac{3}{5}}} p_{2}}$
 & 46 & 477 & 97 & 29655 & 141 & timeout (2 min.) \\ \hline\hline
${\F(\G p_{1}\oplus\F_{\expo_{\frac{1}{2}}} p_{2})}$ & 14 & 19 & 20 & 27 & 23 & 31
\end{tabular}
\endgroup
\caption{Size of the alternating $[0,1]$-acceptance automaton $\A_{\varphi,\varepsilon}$, and $[0,1]$-acceptance automaton $\A_{\varphi,\varepsilon}$}
\label{table:numberOfStates}

\begingroup
\renewcommand{\arraystretch}{1.5}
\begin{tabular}{c|c|c|c|c}
 margin $\varepsilon$ & \#(states of $\K$) & max.\ outgoing degree of $\K$ & 
 time (sec)  &  space (MB) \\ \hhline{=====}
 $\frac{1}{10}$ & 100 & 3 & 0.085508 & 5.861 \\ \cline{3-5}
  &  & 10 & 0.114427 & 9.368 \\ \cline{2-5}
  & 200 & 3 & 0.186989 & 10.586 \\ \cline{3-5}
  &  & 10 & 0.249392 & 18.216 \\ \cline{1-5}
 $\frac{1}{50}$ & 100 & 3 & 5.928842 & 199.782 \\ \cline{3-5}
  & & 10 & 8.108335 & 405.884 \\ \cline{2-5}
 & 200 & 3 & 10.750703 & 405.313 \\ \cline{3-5}
 &  & 10 & 18.250345 & 851.255 
\end{tabular}
\endgroup
\caption{Time and space consumption of our algorithm for near-optimal scheduling, for the formula ${\Gdisc{\expo_{\frac{1}{2}}} p_{1}} \oplus
 {\Gdisc{\expo_{\frac{1}{2}}} p_{2}}$ and a randomly generated
 Kripke structure $\K$. For each choice of the number of states ($100$ or $200$) and of
 the maximum outgoing degree ($3$ or $10$), we randomly generated $100$
 instances of $\K$ and the above shows the average}
\label{table:timeSpaceConsumption}
 
\begingroup
\renewcommand{\arraystretch}{1.5}
\begin{tabular}{c|c|c|c|c}
 & \multicolumn{2}{c|}{$\G_{\expo_{\frac{1}{2}}}\F p$} & \multicolumn{2}{c}{$\F_{\expo_{\frac{1}{2}}} \G p$} \\ \cline{2-5}
 &  time (sec)  &  space (MB) &  time (sec)  &  space (MB) 
 \\ \hhline{=====}
our algorithm (\S{}\ref{subsec:algorithmWithoutPropositional}) & 18.918600 & 897.111 & 0.019800 & 4.629 \\ \hline
binary search & 0.047200 & 5.140390 & 0.069500 & 5.567
\end{tabular}
\endgroup
\caption{(Comparison with binary search, in absence of $\oplus$) Time and space consumption for near-optimal scheduling, for the margin $\varepsilon=\frac{1}{100}$ and a randomly
 generated Kripke structure $\K$ ($500$ states, max.\ outgoing degree $10$, average over
 $100$ instances)}
\label{table:comparisonWithBinarySearch}
\end{table}
In Table~\ref{table:numberOfStates}, for each choice of $\varphi$ and
$\varepsilon$,  we show the size of the alternating automaton
$\A_{\varphi,\varepsilon}$,
and the non-alternating $\Ana_{\varphi,\varepsilon}$ that results
from $\A_{\varphi,\varepsilon}$. 
The first three rows have no $\oplus$, in which case the implementation
scales well for bigger bases (i.e.\ discount functions that decrease more
slowly). We observe that presence of $\oplus$ incurs substantial
computational costs: the small increase of bases from $\frac{1}{2}$ (the
fourth row) to
 $\frac{3}{5}$ (the sixth row) makes $\A_{\varphi,\varepsilon}$ much
 bigger, resulting in  one timeout. This is as expected, however: 
$\oplus$ makes other problems  harder too, such as model checking (undecidable).

In Table~\ref{table:timeSpaceConsumption} we fix a formula $\varphi={\Gdisc{\expo_{\frac{1}{2}}} p_{1}} \oplus
  {\Gdisc{\expo_{\frac{1}{2}}} p_{2}}$ and measure time and space
 consumption, for various choices of a margin $\varepsilon$ and a Kripke
 structure $\K$. 
 Kripke structures $\K$ were randomly generated: 
we first set the number of states (100 or 200) and the maximum outgoing
  degree of $\K$ (3 or 10); for each state we fixed its outgoing degree,
  from the uniform distribution from $1$ to the maximum (that we had
  already fixed); and then, for each outgoing edge, its target state is chosen
  from the uniform distribution over the set of states.
We observe that time and space consumption grows significantly as the
 problem becomes more difficult. However, for problem instances of a
 considerable size we still see manageable costs: a
 margin $\varepsilon=\frac{1}{50}$ (2\%) is fairly small, and 
 a Kripke structure $\K$ with $200$ states is likely to be capable of
 modeling many communication protocols.

In Table~\ref{table:comparisonWithBinarySearch}, for reference, we
compare our algorithm in~\S{}\ref{subsec:algorithmWithoutPropositional}
with the binary-search algorithm that exploits the model-checking
algorithm in~\cite{AlmagorBK14} (we also implemented the latter). We
emphasize again that the latter does not work in presence of
$\oplus$. Our experience shows that the binary-search algorithm can in
some
cases be faster by a magnitude (e.g.\ for the first formula here), but not
always (for the second formula our algorithm is a few times faster). 

 Those experimental results indicate that, although presence of the
 average operator $\oplus$ incurs significant computational cost (as
 expected), 
 automata-based optimal scheduling for
 $\LTLd{\Dexp,\{\oplus\}}$ is potentially a
 viable approach. It is not that
 our algorithm scales up to huge problem instances, but systems of
 hundreds of states can be handled without difficulties.
Identification of concrete real-world challenges,
 and enhancement of the tool's efficiency to match up to them, is an important
 direction of future work.

\section{Conclusions and Future Work}\label{sec:conclFutureWork}
For the quantitative logic $\LTLd{\Dexp,\Fcal}$ with future
discounting~\cite{AlmagorBK14}, we formulated a natural problem of
synthesizing near-optimal schedulers,
 and presented an algorithm. The latter relies on: the existing 
idea of \emph{event horizon}  exploited in~\cite{AlmagorBK14}  for the
threshold model checking problem, as well as a supposedly widely-applicable 
technique of translation to $[0,1]$-acceptance automata and a
lasso-style optimal value algorithm for them.

Here are several directions of future work.

\noindent
\textbf{Controller Synthesis for Open Systems}
We note that the current results are focused on \emph{closed} systems.
For \emph{open} or \emph{reactive} systems (like a server that responds
to requests that come from the environment) we would wish to synthesize a \emph{controller}---formally  a
\emph{strategy} or a \emph{transducer}---that achieves a near-optimal
performance. 

An envisaged workflow,
following the one in~\cite{Vardi96anautomata-theoretic}, is as follows. 
We will use the same automaton $\A_{\varphi,\varepsilon}$
(Def.~\ref{def:ABAforLTL}). It is then: 1) determinized,
2) transformed into a tree automaton that accepts the desired strategies,
and 3) the optimal value of the tree automaton is checked, much like 
in
Lem.~\ref{lem:lassoOptimalityForQuantitativeAcceptAutom}. While the
step 2) will be straightforward, the steps 1) and 3)
(namely: determinization of $[0,1]$-acceptance automata, and the optimal
value
problem for ``$[0,1]$-acceptance Rabin automata'') are yet to be
investigated.
Another possible workflow is by an adaptation of the Safraless algorithm~\cite{KupfermanPV06}.

\noindent
\textbf{Probabilistic Systems and $\LTLd{\Dexp,\Fcal}$} 
Here and in~\cite{AlmagorBK14} the system model is a Kripke structure
that is nondeterministic. Adding probabilistic branching will gives us 
a set of new problems to be solved: for Markov chains
the threshold
model-checking problem can be formulated; for Markov decision processes, 
we have both the threshold
model-checking problem and the near-optimal scheduling problem.
Furthermore, another axis of variation is given by
whether we consider the expected value or the worst-case value. In the
latter case we would wish to exclude truth values that arise
with probability $0$.
 All these variations have important
applications in various areas.

\paragraph*{Acknowledgments}
Thanks are due to
Shaull Almagor,
  Shuichi Hirahara, and
the anonymous referees,
 for useful discussions and comments.
The authors are supported by
   Grants-in-Aid No.\
24680001, 15KT0012 and 15K11984, JSPS.

\bibliographystyle{plain}  
\bibliography{myref}

\newpage
\appendix

\section{Our General Algorithm for $\LTLd{\D,\Fcalmc}$, Further Details}
\label{appendix:generalAlgorithmDetails}
In this section, we extend~\S{}\ref{subsec:algorithmWithPropositional}
and describe details of the construction of $\A_{\varphi,\varepsilon}$
for a formula $\varphi$ of $\LTLd{\D,\Fcalmc}$. 
We inductively construct an alternating $[0,1]$-acceptance automaton 
$\A_{\varphi,\varepsilon}^{\vec{d}}$---that is also parametrized by 
a discount sequence $\vec{d}$. Then the automaton 
 $\A_{\varphi,\varepsilon}$ is defined by 
$\A_{\varphi,\varepsilon}^{\langle 1 \rangle}$ (for the sequence
$\langle 1 \rangle$ of length one).
\begin{lem}\label{lem:correctnessWithPropositional}
  Let $\varphi$ be an $\LTLd{\D,\Fcalmc}$ formula, $\varepsilon\in(0,1)$ be a
 positive real number, and $\vec{d}=d_{1}d_{2}\dotsc d_{n}$ be a discount sequence (Def.~\ref{def:discountSequence}). There exists an
alternating
$[0,1]$-acceptance automaton $\A_{\varphi,\varepsilon}^{\vec{d}}$ such that,
for each computation $\pi \in (\mathcal{P}(\AP))^{\omega}$, 
  \begin{equation}
    \bigl(\vec{d}\boxtimes\sem{ \pi, \varphi }\bigr) - \varepsilon \;\leq\; \Lang(\A^{\vec{d}}_{\varphi,
   \varepsilon})(\pi)
 \; \leq\; \vec{d}\boxtimes\sem{ \pi, \varphi }\enspace.
  \end{equation}
\end{lem}

\begin{proof}
  The proof is inductive on the construction of $\varphi$.\footnote{To
 be precise, we have two nested induction: the outer one is with respect
 to the number of propositional quality operators occurring in
 $\varphi$; and the inner one is with respect to the size of a formula $\varphi$.} In this proof
 we assume without loss of generality that an alternating
 $[0,1]$-acceptance automaton has exactly one initial state, and
 consequently, the initial state of $\A_{\varphi,\varepsilon}^{\vec{d}}$
 shall be denoted by $q_{\varphi,\varepsilon}^{\vec{d}}$. For the case
 where the outermost connective of $\varphi$ is other than a
 propositional quality operator, we only describe the construction of
 $\A_{\varphi,\varepsilon}^{\vec{d}}$. The correctness of this automaton
 can be proved in a similar way to the proof of
 Lem.~\ref{lem:correctnessOfA}:
recall that Lem.~\ref{lem:correctnessOfA} is also proved by induction on the construction of a formula. 
  
  Suppose that $\varphi = \true$. We define $\A_{\true,\varepsilon}^{\vec{d}} = (\mathcal{P}(\AP),\{q_{\true,\varepsilon}^{\vec{d}} \},\{q_{\true,\varepsilon}^{\vec{d}} \},\delta,F)$ 
  where $\delta (q_{\true,\varepsilon}^{\vec{d}},\sigma) = \vec{d}\boxtimes 1$ and 
  $F(q_{\true,\varepsilon}^{\vec{d}}) = 0$. 
  
  Suppose that $\varphi = p \in \AP$. We define $\A_{p,\varepsilon}^{\vec{d}} = (\mathcal{P}(\AP),\{q_{p,\varepsilon}^{\vec{d}} \},\{q_{p,\varepsilon}^{\vec{d}} \},\delta,F)$ where 
\begin{displaymath}
  \delta (q_{p,\varepsilon}^{\vec{d}},\sigma) = 
  \begin{cases}
    \vec{d}\boxtimes 1 & \mbox{if } p \in \sigma \\
	\vec{d}\boxtimes 0 & \mbox{otherwise}
  \end{cases}
  \quad\mbox{ and }\quad
  F(q_{p,\varepsilon}^{\vec{d}}) = 0\enspace.
\end{displaymath}
 
 Suppose that $\varphi = \varphi_{1} \land \varphi_{2}$ and that
 $|\vec{d}|$ is odd. By the induction hypothesis, for each of $i \in \{1,2\}$, there exists $\A_{\varphi_{i},\varepsilon}^{\vec{d}} = (\mathcal{P}(\AP),Q_{i},\{q_{\varphi_{i},\varepsilon}^{\vec{d}} \},\delta_{i},F_{i})$ that satisfies the postulated condition. Then we define $\A_{\varphi,\varepsilon}^{\vec{d}} = (\mathcal{P}(\AP),Q,\{q_{\varphi,\varepsilon}^{\vec{d}} \},\delta,F)$ as follows. Its state space $Q$ is $\{q_{\varphi,\varepsilon}^{\vec{d}} \}\cup Q_{1}\cup Q_{2}$. The transition function $\delta$ is
\begin{displaymath}
  \delta (q,\sigma) = 
  \begin{cases}
    \delta_{1} (q_{\varphi_{1},\varepsilon}^{\vec{d}},\sigma) \land \delta_{2} (q_{\varphi_{2},\varepsilon}^{\vec{d}},\sigma) & \mbox{if } q = q_{\varphi,\varepsilon}^{\vec{d}} \\
    \delta_{i} (q,\sigma) & \mbox{if } q \in Q_{i}.
  \end{cases}
\end{displaymath}
  The acceptance function $F$ is
\begin{displaymath}
  F (q) = 
  \begin{cases}
    0 & \mbox{if } q = q_{\varphi,\varepsilon}^{\vec{d}} \\
    F_{i} (q) & \mbox{if } q \in Q_{i}.
  \end{cases}
\end{displaymath}

 Suppose that $\varphi = \varphi_{1} \land \varphi_{2}$ and that
 $|\vec{d}|$ is even. By the induction hypothesis, for each of $i \in \{1,2\}$, there exists $\A_{\varphi_{i},\varepsilon}^{\vec{d}} = (\mathcal{P}(\AP),Q_{i},\{q_{\varphi_{i},\varepsilon}^{\vec{d}} \},\delta_{i},F_{i})$ that satisfies the postulated condition. Then we define $\A_{\varphi,\varepsilon}^{\vec{d}} = (\mathcal{P}(\AP),Q,\{q_{\varphi,\varepsilon}^{\vec{d}} \},\delta,F)$ as follows. Its state space $Q$ is $\{q_{\varphi,\varepsilon}^{\vec{d}} \}\cup Q_{1}\cup Q_{2}$. The transition function $\delta$ is
\begin{displaymath}
  \delta (q,\sigma) = 
  \begin{cases}
    \delta_{1} (q_{\varphi_{1},\varepsilon}^{\vec{d}},\sigma) \lor \delta_{2} (q_{\varphi_{2},\varepsilon}^{\vec{d}},\sigma) & \mbox{if } q = q_{\varphi,\varepsilon}^{\vec{d}} \\
    \delta_{i} (q,\sigma) & \mbox{if } q \in Q_{i}.
  \end{cases}
\end{displaymath}
  The acceptance function $F$ is
\begin{displaymath}
  F (q) = 
  \begin{cases}
    0 & \mbox{if } q = q_{\varphi,\varepsilon}^{\vec{d}} \\
    F_{i} (q) & \mbox{if } q \in Q_{i}.
  \end{cases}
\end{displaymath}

Suppose that $\varphi = \lnot\varphi'$. By the induction hypothesis, there exists $\A_{\varphi',\varepsilon}^{\vec{d}\concatseq 1}$ that satisfies the postulated condition. Let $\A_{\varphi,\varepsilon}^{\vec{d}} = \A_{\varphi',\varepsilon}^{\vec{d}\concatseq 1}$.

 Suppose that $\varphi = \X\varphi'$. By the induction hypothesis, there exists $\A_{\varphi',\varepsilon}^{\vec{d}} = (\mathcal{P}(\AP),Q',\{q_{\varphi',\varepsilon}^{\vec{d}} \},\delta',F')$ that satisfies the postulated condition. Then we define $\A_{\varphi,\varepsilon}^{\vec{d}} = (\mathcal{P}(\AP),Q,\{q_{\varphi,\varepsilon}^{\vec{d}} \},\delta,F)$ as follows. Its state space $Q$ is $\{q_{\varphi,\varepsilon}^{\vec{d}} \}\cup Q'$. The transition function $\delta$ is
\begin{displaymath}
  \delta (q,\sigma) = 
  \begin{cases}
    q_{\varphi',\varepsilon}^{\vec{d}} & \mbox{if } q = q_{\varphi,\varepsilon}^{\vec{d}} \\
    \delta' (q,\sigma) & \mbox{otherwise. }
  \end{cases}
\end{displaymath}
  The acceptance function $F$ is
\begin{displaymath}
  F (q) = 
  \begin{cases}
    0 & \mbox{if } q = q_{\varphi,\varepsilon}^{\vec{d}} \\
    F' (q) & \mbox{otherwise. }
  \end{cases}
\end{displaymath}

 Suppose that $\varphi = \varphi_{1} \U \varphi_{2}$ and that
 $|\vec{d}|$ is odd. By the induction hypothesis, for each of $i \in \{1,2\}$, there exists $\A_{\varphi_{i},\varepsilon}^{\vec{d}} = (\mathcal{P}(\AP),Q_{i},\{q_{\varphi_{i},\varepsilon}^{\vec{d}} \},\delta_{i},F_{i})$ that satisfies the postulated condition. Then we define $\A_{\varphi,\varepsilon}^{\vec{d}} = (\mathcal{P}(\AP),Q,\{q_{\varphi,\varepsilon}^{\vec{d}} \},\delta,F)$ as follows. Its state space $Q$ is $\{q_{\varphi,\varepsilon}^{\vec{d}} \}\cup Q_{1}\cup Q_{2}$. The transition function $\delta$ is
\begin{displaymath}
  \delta (q,\sigma) = 
  \begin{cases}
    \delta_{2} (q_{\varphi_{2},\varepsilon}^{\vec{d}},\sigma) \lor (\delta_{1} (q_{\varphi_{1},\varepsilon}^{\vec{d}},\sigma) \land q_{\varphi,\varepsilon}^{\vec{d}}) & \mbox{if } q = q_{\varphi,\varepsilon}^{\vec{d}} \\
    \delta_{i} (q,\sigma) & \mbox{if } q \in Q_{i}.
  \end{cases}
\end{displaymath}
  The acceptance function $F$ is
\begin{displaymath}
  F (q) = 
  \begin{cases}
    0 & \mbox{if } q = q_{\varphi,\varepsilon}^{\vec{d}} \\
    F_{i} (q) & \mbox{if } q \in Q_{i}.
  \end{cases}
\end{displaymath}

 Suppose that $\varphi = \varphi_{1} \U \varphi_{2}$ and that
 $|\vec{d}|$ is even. By the induction hypothesis, for each of $i \in \{1,2\}$, there exists $\A_{\varphi_{i},\varepsilon}^{\vec{d}} = (\mathcal{P}(\AP),Q_{i},\{q_{\varphi_{i},\varepsilon}^{\vec{d}} \},\delta_{i},F_{i})$ that satisfies the postulated condition. Then we define $\A_{\varphi,\varepsilon}^{\vec{d}} = (\mathcal{P}(\AP),Q,\{q_{\varphi,\varepsilon}^{\vec{d}} \},\delta,F)$ as follows. Its state space $Q$ is $\{q_{\varphi,\varepsilon}^{\vec{d}} \}\cup Q_{1}\cup Q_{2}$. The transition function $\delta$ is
\begin{displaymath}
  \delta (q,\sigma) = 
  \begin{cases}
    \delta_{2} (q_{\varphi_{2},\varepsilon}^{\vec{d}},\sigma) \land (\delta_{1} (q_{\varphi_{1},\varepsilon}^{\vec{d}},\sigma) \lor q_{\varphi,\varepsilon}^{\vec{d}}) & \mbox{if } q = q_{\varphi,\varepsilon}^{\vec{d}} \\
    \delta_{i} (q,\sigma) & \mbox{if } q \in Q_{i}.
  \end{cases}
\end{displaymath}
  The acceptance function $F$ is
\begin{displaymath}
  F (q) = 
  \begin{cases}
    1 & \mbox{if } q = q_{\varphi,\varepsilon}^{\vec{d}} \\
    F_{i} (q) & \mbox{if } q \in Q_{i}.
  \end{cases}
\end{displaymath}

  Suppose that $\varphi = \varphi_{1} \U_{\eta^{+k}} \varphi_{2}$ and
 that $|\vec{d}|$ is odd. Since $\lim_{i \to \infty}\eta(i) = 0$, there
 exists a natural number $k_{\max} \in \N$ such that
 $\eta(k_{\max})\cdot \prod_{i = 1}^{n} d_{i} \leq \varepsilon$ (i.e.\
 $k_{\max}$ is beyond the event
horizon). We construct $\A_{\varphi,\varepsilon}^{\vec{d}}$ by induction
 on $k$ backwards, that is, starting from $k = k_{\max}$ and
 decrementing $k$ one by one until $k = 0$.
 If $\eta(k)\cdot \prod_{i = 1}^{n} d_{i} \leq \varepsilon$, we define $\A_{\varphi,\varepsilon}^{\vec{d}} = (\mathcal{P}(\AP),\{q_{\varphi,\varepsilon}^{\vec{d}} \},\{q_{\varphi,\varepsilon}^{\vec{d}} \},\delta,F)$ 
  where $\delta (q_{\varphi,\varepsilon}^{\vec{d}},\sigma) = \vec{d}\boxtimes 0$ and 
  $F(q_{\varphi,\varepsilon}^{\vec{d}}) = 0$.
  Otherwise, we define $\A_{\varphi,\varepsilon}^{\vec{d}} =
 (\mathcal{P}(\AP),Q,\{q_{\varphi,\varepsilon}^{\vec{d}} \},\delta,F)$
 as follows. By the induction hypothesis, for each of $i \in \{1,2\}$, there exists $\A_{\varphi_{i},\varepsilon}^{\vec{d}\odot\eta(k)} = (\mathcal{P}(\AP),Q_{i},\{q_{\varphi_{i},\varepsilon}^{\vec{d}\odot\eta(k)} \},\delta_{i},F_{i})$ that satisfies the postulated condition. Moreover, there exists $\A_{\varphi_{1} \U_{\eta^{k+1}} \varphi_{2},\varepsilon}^{\vec{d}} = (\mathcal{P}(\AP),Q_{3},\{q_{\varphi_{1} \U_{\eta^{k+1}} \varphi_{2},\varepsilon}^{\vec{d}} \},\delta_{3},F_{3})$. We define the state space $Q$ of $\A_{\varphi,\varepsilon}^{\vec{d}}$ by $\{q_{\varphi,\varepsilon}^{\vec{d}} \}\cup Q_{1}\cup Q_{2}\cup Q_{3}$. The transition function $\delta$ is
\begin{displaymath}
  \delta (q,\sigma) = 
  \begin{cases}
    \delta_{2} (q_{\varphi_{2},\varepsilon}^{\vec{d}\odot\eta(k)},\sigma) \lor (\delta_{1} (q_{\varphi_{1},\varepsilon}^{\vec{d}\odot\eta(k)},\sigma) \land q_{\varphi_{1} \U_{\eta^{k+1}} \varphi_{2},\varepsilon}^{\vec{d}}) & \mbox{if } q = q_{\varphi,\varepsilon}^{\vec{d}} \\
    \delta_{i} (q,\sigma) & \mbox{if } q \in Q_{i}.
  \end{cases}
\end{displaymath}
  The acceptance function $F$ is
\begin{displaymath}
  F (q) = 
  \begin{cases}
    0 & \mbox{if } q = q_{\varphi,\varepsilon}^{\vec{d}} \\
    F_{i} (q) & \mbox{if } q \in Q_{i}.
  \end{cases}
\end{displaymath}

  Suppose that $\varphi = \varphi_{1} \U_{\eta^{+k}} \varphi_{2}$ and that $|\vec{d}|$ is even. Similarly to the case where $|\vec{d}|$ is odd, we construct $\A_{\varphi,\varepsilon}^{\vec{d}}$ by induction on $k$ backwards.
 If $\eta(k)\cdot \prod_{i = 1}^{n} d_{i} \leq \varepsilon$, we define $\A_{\varphi,\varepsilon}^{\vec{d}} = (\mathcal{P}(\AP),\{q_{\varphi,\varepsilon}^{\vec{d}} \},\{q_{\varphi,\varepsilon}^{\vec{d}} \},\delta,F)$ 
  where $\delta (q_{\varphi,\varepsilon}^{\vec{d}},\sigma) = \vec{d}\boxtimes \eta(k)$ and 
  $F(q_{\varphi,\varepsilon}^{\vec{d}}) = 0$.
  Otherwise, we define $\A_{\varphi,\varepsilon}^{\vec{d}} =
 (\mathcal{P}(\AP),Q,\{q_{\varphi,\varepsilon}^{\vec{d}} \},\delta,F)$
 as follows. By the induction hypothesis, for each of $i \in \{1,2\}$, there exists $\A_{\varphi_{i},\varepsilon}^{\vec{d}\odot\eta(k)} = (\mathcal{P}(\AP),Q_{i},\{q_{\varphi_{i},\varepsilon}^{\vec{d}\odot\eta(k)} \},\delta_{i},F_{i})$ that satisfies the postulated condition. Moreover, there exists $\A_{\varphi_{1} \U_{\eta^{k+1}} \varphi_{2},\varepsilon}^{\vec{d}} = (\mathcal{P}(\AP),Q_{3},\{q_{\varphi_{1} \U_{\eta^{k+1}} \varphi_{2},\varepsilon}^{\vec{d}} \},\delta_{3},F_{3})$. We define the state space $Q$ of $\A_{\varphi,\varepsilon}^{\vec{d}}$ by $\{q_{\varphi,\varepsilon}^{\vec{d}} \}\cup Q_{1}\cup Q_{2}\cup Q_{3}$. The transition function $\delta$ is
\begin{displaymath}
  \delta (q,\sigma) = 
  \begin{cases}
    \delta_{2} (q_{\varphi_{2},\varepsilon}^{\vec{d}\odot\eta(k)},\sigma) \land (\delta_{1} (q_{\varphi_{1},\varepsilon}^{\vec{d}\odot\eta(k)},\sigma) \lor q_{\varphi_{1} \U_{\eta^{k+1}} \varphi_{2},\varepsilon}^{\vec{d}}) & \mbox{if } q = q_{\varphi,\varepsilon}^{\vec{d}} \\
    \delta_{i} (q,\sigma) & \mbox{if } q \in Q_{i}.
  \end{cases}
\end{displaymath}
  The acceptance function $F$ is
\begin{displaymath}
  F (q) = 
  \begin{cases}
    0 & \mbox{if } q = q_{\varphi,\varepsilon}^{\vec{d}} \\
    F_{i} (q) & \mbox{if } q \in Q_{i}.
  \end{cases}
\end{displaymath}

  Suppose that $\varphi = f(\varphi_{1},\ldots, \varphi_{k})$ where $f
 \in \Fcalmc$ and that $|\vec{d}|$ is odd. Since $f$ is continuous and
 its domain $[0,1]^{k}$ is bounded and closed in the Euclidean space $\mathbb{R}^{k}$, this function $f$ is uniformly continuous by the Heine–Cantor theorem. By the monotonicity and the uniform continuity, there exists $\varepsilon' \in (0,1)$ such that, for each $\mathbf{x} = (x_{1},\ldots,x_{k}) \in [0,1]^{k}$, 
  \begin{equation}\label{eq:newEpsilon}
    f(x_{1} \dotminus \varepsilon',\ldots,x_{k} \dotminus \varepsilon') \geq f(\mathbf{x}) - {\varepsilon / {(d_{1} \cdot d_{2} \cdot \cdots \cdot d_{n}}})\enspace, 
  \end{equation}
  where $a \dotminus b$ is defined by $\max \{ a-b, 0 \}$. By the  induction hypothesis, there exist $\A_{\varphi_{1}, \varepsilon'}^{\langle 1 \rangle},\ldots,\A_{\varphi_{k}, \varepsilon'}^{\langle 1 \rangle}$ such that, for $i \in \{1,\ldots,k\}$, 
  \begin{displaymath}
    {\sem{ \pi, \varphi_{i} }} - \varepsilon'
\; \leq\; \Lang(\A_{\varphi_{i}, \varepsilon'}^{\langle 1 \rangle})(\pi) 
\;\leq\; {\sem{ \pi, \varphi_{i} }}
  \end{displaymath}
  for each $\pi \in (\mathcal{P}(\AP))^{\omega}$. Since $|\vec{d}|$ is
 odd, the function $g\colon [0,1]^{k}\to [0,1]$ defined by $g(\mathbf{x}) = \vec{d} \boxtimes f(\mathbf{x})$ is
 monotone in $\mathbf{x}$. Since the class of languages of alternating
 $[0,1]$-acceptance automata and that of $[0,1]$-acceptance automata are
 the same by Prop.~\ref{prop:ABAtoNBA}, the closure property in
 Prop.~\ref{prop:ClosedUnderIncOperator} remains true even if $[0,1]$-acceptance automata are replaced by alternating $[0,1]$-acceptance automata. Hence, there exists $g(\A_{\varphi_{1}, \varepsilon'}^{\langle 1 \rangle},\ldots,\A_{\varphi_{k}, \varepsilon'}^{\langle 1 \rangle})$ defined in Prop.~\ref{prop:ClosedUnderIncOperator}. 
 By~(\ref{eq:newEpsilon}) and the definition~(\ref{eq:actionExplicitly}) of the operator $\boxtimes$, we have
  \begin{displaymath}
    g(x_{1} \dotminus \varepsilon',\ldots,x_{1} \dotminus \varepsilon') \geq g(\mathbf{x}) - \varepsilon = \vec{d} \boxtimes f(\mathbf{x}) - \varepsilon\enspace.
  \end{displaymath}
 Hence, if we define $\A_{\varphi, \varepsilon}^{\vec{d}}$ by
 $g(\A_{\varphi_{1}, \varepsilon'}^{\langle 1
 \rangle},\ldots,\A_{\varphi_{k}, \varepsilon'}^{\langle 1 \rangle})$, 
it
satisfies the postulated condition.
  
  Suppose that $\varphi = f(\varphi_{1},\ldots, \varphi_{k})$ where $f
 \in \Fcalmc$ and that $|\vec{d}|$ is even. Let $\vec{d'} =
 d_{1}d_{2}\dotsc d_{n-1}$ be a prefix of $\vec{d}$. Then we have
 $\vec{d} \boxtimes v = \vec{d'} \boxtimes (1 - d_{n} \cdot v)$. We
 define a function $(d_{n}\cdot f)^{*} \colon [0,1]^{k} \to [0,1]$ by
 $(d_{n}\cdot f)^{*} (x_{1},\ldots,x_{k}) = 1 - d_{n}\cdot
 f(1-x_{1},\ldots,1-x_{k})$. Let $\varphi' = (d_{n}\cdot f)^{*}
 (\lnot\varphi_{1},\ldots,\lnot\varphi_{k})$. It is obvious that
 $(d_{n}\cdot f)^{*} \in \Fcalmc$. Moreover, we have $\vec{d} \boxtimes
 \sem{\pi,\varphi} = \vec{d'} \boxtimes \sem{\pi, \varphi'}$ for each
 $\pi \in (\mathcal{P}(\AP))^{\omega}$, and $\vec{d'}$ is odd. Therefore
 there exists $\A_{\varphi', \varepsilon}^{\vec{d'}}$ because of the
 previous case (i.e.\ when $|\vec{d}|$ is odd),\footnote{Recall that
 we are currently running two nested induction, with the outer one being with respect
 to the number of propositional quality operators.} 
and we take this as $\A_{\varphi, \varepsilon}^{\vec{d}}$. Then $\A_{\varphi, \varepsilon}^{\vec{d}}$ satisfies the postulated condition.
\end{proof}

Once $\A_{\varphi,\varepsilon}$ is constructed, the procedure described in \S{}\ref{subsubsec:theAlgorithmWithoutPropositional} works regardless of the presence of propositional quality operators.

\section{Omitted Proofs}\label{appendix:omittedproofs}

\subsection{Proof of Prop.~\ref{prop:ABAtoNBA}
}
\label{pf:lemABAtoNBA}
 \begin{proof}
\auxproof{  (Miyano,Hayashi (1984) Alternating Finite Automata on
  omega-Words.)}
  We first describe the formal construction; intuitions follow shortly.
  
  Without loss of generality, we can assume that a positive Boolean formula $\delta(q,a)$ is a
 disjunctive normal form; 
 therefore the transition function is of the type
 $\delta : Q \times \Sigma \rightarrow
  \mathcal{P}(\mathcal{P}(Q \cup [0,1]))$. 
 More concretely,
 for each $q\in Q$ and $a\in \Sigma$, 
 the formula $\delta(q,a)$ is a disjunction of formulas of the form
\begin{displaymath}
   (q_{1}\land\cdots\land q_{k})
   \land (v_{1}
   \land\cdots\land
   v_{l})
\end{displaymath}
where $q_{j}\in Q$ and $v_{j}\in [0,1]$ are
 atomic propositions (we  changed their order suitably). Moreover, since the conjunction 
\begin{math}
v_{1}
   \land\cdots\land
   v_{l}
\end{math}
 is equivalent to a single atomic proposition $\min\{
v_{1},
   \dotsc,
   v_{l}
\}$, we  assume that any disjunct of the DNF formula $\delta(q,a)$ is of
 the form 
\begin{displaymath}
   (q_{1}\land\cdots\land q_{k})
   \land v\enspace.
\end{displaymath}

Let $V_{Q} = \{ F(q) \mid q
 \in Q \}$ be the set of acceptance values that occur in $\A$, and
 $V_\delta$ be the set of values from $[0,1]$  (i.e.\ atomic
 propositions from $[0,1]$) that
 occur in the transition function $\delta$, that is,
\begin{displaymath}
 V_{\delta}
 \;=\;
 \bigcup_{q\in Q, a\in \Sigma}
 \bigl\{\,
  v
  \,\bigl|\bigr.\,
  \bigl(\,(q_{1}\land\cdots\land q_{k})\land v\,\bigr) \in \delta(q,a)
\,\bigr\}\enspace.
\end{displaymath}
We define $\A' = (\Sigma,Q',I',\delta',F')$ as follows.
  \begin{align*}
 Q' &\;=\; {\mathcal{P} (Q \times V_{Q})} \times V_{\delta} \times \{ \ffalse, \ttrue
   \}\enspace, 
	\\ 
I' &\;=\; \bigl\{\, \bigl(\, \bigl\{ (q_{0}, F(q_{0}))
   \bigr\},1,\ffalse\,\bigr) 
   \;\bigl|\bigr.\;
   q_{0} \in I\, 
   \,\bigr\}\enspace, \\
   F'(Y,v,b) &\;=\; 
    \begin{cases}
      \min \bigl\{\, v, \min \{ v' \mid \exists q \in Q.\, (q,v') \in Y\} \,\bigr\} & \mbox{if } b = \ttrue \\
      0 & \mbox{otherwise.}
    \end{cases}
  \end{align*}
The transition function $\delta'$ is defined as follows.
Let
 $\widetilde{q}=\bigl(\,\bigl\{\,(q^{1},v^{1}),\dotsc,(q^{n},v^{n})\,\bigr\}, v,
 b\,\bigr)$ be a state
in $Q'$, and $a\in \Sigma$. Then $\delta'(\widetilde{q},a)$ is defined,
  in case $b=\ffalse$, by:
\begin{equation}\label{eq:defDeltaPrimeNotExternalizing}
\begin{aligned}
  \left\{\quad
 \left.
  \left(
 \begin{array}{c}
  \left\{
 \begin{array}{c}
  \bigl(\,q^{1}_{1},\, \max\{v^{1},F(q^{1}_{1})\}\,\bigr)\,,\; \dotsc,\;
    \bigl(\,q^{1}_{l_{1}},\, \max\{v^{1},F(q^{1}_{l_{1}})\}\,\bigr),
 \\
  \vdots
 \\
  \bigl(\,q^{n}_{1},\, \max\{v^{n},F(q^{n}_{1})\}\,\bigr)\,,\; \dotsc,\;
    \bigl(\,q^{n}_{l_{n}},\, \max\{v^{n},F(q^{n}_{l_{n}})\}\,\bigr)
 \end{array}
 \right\}\,,
 \\[+.7cm]
 \min\{v,u^{1},\dotsc,u^{n}\}\,,
 \\
 b'
 \end{array}
 \right)
 \;\right.\;
 \right.
 \qquad\qquad
 \\
 \Bigl.
 \Bigl.
 \Bigr|\quad
 \bigl(\,(q^{i}_{1}\land\cdots\land q^{i}_{l_{i}})\land
 u^{i}\,\bigr)\in\delta(q^{i},a)\,,\quad b'\in\{\ttrue,\ffalse\}
 \quad
 \Bigr\}\enspace;
\end{aligned}
\end{equation}
in case $b=\ttrue$, 
\begin{equation}\label{eq:defDeltaPrimeExternalizing}
 \left\{\quad
\left.
  \left(
\begin{array}{c}
  \left\{
 \begin{array}{c}
  \bigl(\,q^{1}_{1},\, F(q^{1}_{1})\,\bigr)\,,\; \dotsc,\;
    \bigl(\,q^{1}_{l_{1}},\, F(q^{1}_{l_{1}})\,\bigr),
 \\
  \vdots
 \\
  \bigl(\,q^{n}_{1},\, F(q^{n}_{1})\,\bigr)\,,\; \dotsc,\;
    \bigl(\,q^{n}_{l_{n}},\, F(q^{n}_{l_{n}})\,\bigr)
 \end{array}
 \right\}\,,
\\[+.7cm]
\min\{v,u^{1},\dotsc,u^{n}\}\,,
\\
 b'
\end{array}
\right)
\;\right|\;
\begin{array}{l}
\bigl(\,(q^{i}_{1}\land\cdots\land q^{i}_{l_{i}})\land
 u^{i}\,\bigr)
 \\
\qquad\qquad\in\delta(q^{i},a)\enspace,
 \\
b'\in\{\ttrue,\ffalse\}
\end{array}
\right\}\enspace.
\end{equation}
  In each case ($b=\ffalse$ or $\ttrue$), different $a$-successors of
  $\widetilde{q}$ arise from: 1) different choices of a disjunct
of a DNF formula $\delta(q^{i},a)$, for $i\in [0,n]$; and 2) different choices of 
$b'$ (it can always be chosen from $\ttrue$ and $\ffalse$).

In the setting of~\cite{MiyanoH84} (that is Boolean instead of
 quantitative), the state space $Q'$ of the nondeterministic automaton obtained as a translation of
  an alternating one is $\mathcal{P}(Q\times \{0,1\})$.
  Its quantitative adaptation $\mathcal{P}(Q\times V_{Q})$ occurs as
  the first component of $Q'$
  in our above
  quantitative construction;
  the rest 
  $V_{\delta}\times\{\ffalse,\ttrue\}$ of $Q'$ is there
  for handling
   quantitative 
acceptance.

It is not hard to see that $\A$ and $\A'$ have the same language.\footnote{ A  more rigorous proof can be given via formulating an
 acceptance game for an alternating $[0,1]$-acceptance automaton.
}
For example, in a state
 $\widetilde{q}=\bigl(\,\bigl\{\,(q^{1},v^{1}),\dotsc,(q^{n},v^{n})\,\bigr\}, v,
 b\,\bigr)$ of $\A'$:
\begin{itemize}
 \item The pair $(q^{i}, v^{i})$ is that of the \emph{current state} and 
  what we call the \emph{internally accumulated acceptance value}.
 \item The set $\bigl\{\,(q^{1},v^{1}),\dotsc,(q^{n},v^{n})\,\bigr\}$  stands for the \emph{conjunction} of these pairs.
 \item The second component $v\in[0,1]$ of $\widetilde{q}$ is for
       keeping track of: the values at the
       leaves
       of the corresponding run tree, more precisely the smallest among such.
 \item The flag $b\in\{\ffalse,\ttrue\}$ is called an \emph{exposition
       flag}: it determines if
       the internally accumulated  acceptance values
       $v^{1},\dotsc,v^{n}$ should be exposed or not.  
       Note the definition of $F'$: the acceptance value of a state of
       $\A'$ 
       is nonzero only if the exposition flag $b$ is $\ttrue$. 
\end{itemize}
Let us comment on the definition of the transition function $\delta'$. 
  Starting from $\widetilde{q}
=\bigl(\,\bigl\{\,(q^{1},v^{1}),\dotsc,(q^{n},v^{n})\,\bigr\}, v,
 b\,\bigr)$---in which the ``current state'' is the
 conjunction $q^{1}\land q^{2}\land\cdots\land q^{n}$---we choose one disjunct 
 $q^{i}_{1}\land\cdots\land q^{i}_{l_{i}}\in \delta(q^{i},a)$ 
for
 each $q^{i}$ and the ``next state'' is
\begin{displaymath}
 (q^{1}_{1}\land\cdots\land q^{1}_{l_{1}})
 \land
  (q^{2}_{1}\land\cdots\land q^{2}_{l_{2}})
\land\cdots\land
(q^{n}_{1}\land\cdots\land q^{n}_{l_{n}})\enspace.
\end{displaymath}
If the exposition flag $b$ is $\ffalse$ then we keep accumulating
the acceptance values that we have seen since the last exposition,
 resulting
in the occurrence of $\max$
 in~(\ref{eq:defDeltaPrimeNotExternalizing}). If the flag is $\ttrue$
then the internally accumulated acceptance values are ``used'' (see the
 definition of $F'$), and these values must be ``forgotten'' so that we
  simulate a B\"uchi-like acceptance condition for $\A$. Therefore
in~(\ref{eq:defDeltaPrimeExternalizing}), there are
 no $v^{1},\dotsc,v^{n}$ occurring and we have a fresh start.
 \end{proof}

The state space $Q'$ of $\A'$ in the previous proof can actually be
smaller: we can identify two states $(Y,v,b)$ and $(Y',v,b)$ if $\min \{
v' \in V_{Q} \mid (q,v') \in Y \} = \min \{ v' \in V_{Q} \mid (q,v') \in
Y' \}$ holds for each $q \in Q$---this is the case for example when
$Y=\{(q,\frac{1}{2}),{(q,1)}\}$ and $Y'=\{(q,\frac{1}{2})\}$.  Therefore
we only need states $(Y,v,b)$ such that $\forall (q,v),(q',v') \in Y.\,
(q = q' \Rightarrow v = v')$, that is, $Y$ can be regarded as a partial
function. Summarizing, we can reduce the state space to $(V_{Q} \cup
\{ * \})^{Q} \times V_{\delta} \times \{ \ffalse, \ttrue \}$. The size
of the first component is $2^{|Q|\times \log |V_{Q}|}$, while it was
$2^{|Q|\times |V_{Q}|}$ before this optimization.

\subsection{Proof of Prop.~\ref{prop:ClosedUnderIncOperator}}
The proof is an adaptation of that of Prop.~\ref{prop:ABAtoNBA}. Here we
combine the usual construction of synchronous products of automata, with
the idea of exposition flags.
\begin{proof}
  Let $\A_{i} = (\Sigma,Q_{i},I_{i},\delta_{i},F_{i})$ for each $i \in
 \{ 1,\ldots,k\}$.  We define $f(\A_{1},\ldots,\A_{k}) =
 (\Sigma,Q,I,\delta,F)$ as follows. Its state space $Q$ is $Q=(\prod_{1
 \leq i \leq k}(Q_{i} \times V_{i})) \times \{ \ffalse,\ttrue \}$ where
 $V_{i} = \{ 0 \} \cup \{ v \in [0,1] \mid \exists q \in
 Q_{i}.\,F_{i}(q) = v \}$. The set $I$ of initial states is $I=\{
 ((q_{1},0),\ldots,(q_{k},0),\ffalse) \mid q_{1} \in I_{1},\ldots, q_{k}
 \in I_{k} \}$. The acceptance function is defined by
  \begin{equation}\label{eq:propClosedUnderIncOperator1}
    F \bigl((q_{1},v_{1}),\ldots,(q_{k},v_{k}),b\bigr) \;= \;
    \begin{cases}
      f(v_{1},\ldots,v_{k}) & \mbox{if } b = \ttrue \\
	  0 & \mbox{otherwise.}
    \end{cases}
  \end{equation}
  The transition function $\delta \colon Q \times \Sigma \to \Pcal(Q)$
 is defined as follows. Let $q = \bigl((q_{1},v_{1}),\ldots,(q_{k},v_{k}),b\bigr)
 \in Q$, and $a\in \Sigma$.
  \begin{equation}\label{eq:propClosedUnderIncOperator2}
    \delta (q,a) = 
    \begin{cases}
      \left(\displaystyle\prod_{1 \leq i \leq k} \bigl\{\, (q'_{i}, F_{i} (q'_{i})) \,\bigl|\bigr.\, q'_{i} \in \delta_{i} (q_{i},a) \,\bigr\} \right) \times \{ \ffalse,\ttrue \} & \mbox{if } b = \ttrue \\
	  \left(\displaystyle\prod_{1 \leq i \leq k} \bigl\{\, (q'_{i}, \max \{ v_{i},F_{i} (q'_{i}) \}) \,\bigl|\bigr.\, q'_{i} \in \delta_{i} (q_{i},a) \,\bigr\} \right) \times \{ \ffalse,\ttrue \} & \mbox{otherwise.}
    \end{cases}
  \end{equation}
 We shall prove that the automaton $f(\A_{1},\ldots,\A_{k})$ indeed
 satisfies the requirement. Recall that, by definition, a
 $[0,1]$-acceptance automaton has no dead ends.  Let $w \in
 \Sigma^{\omega}$ be an infinite word. 

 On the one hand, it follows
 easily from
 the above definition (in
 particular~(\ref{eq:propClosedUnderIncOperator1})) that if $\Lang(f(\A_{1},\ldots,\A_{k})) (w) = \overline{v}$,
 there exist $\overline{v}_{1},\ldots,\overline{v}_{k} \in [0,1]$ such that:
 $f(\overline{v}_{1},\ldots,\overline{v}_{k}) = \overline{v}$, and $\Lang(\A_{i}) (w) \geq \overline{v}_{i}$ for each $i
 \in \{ 1,\ldots,k\}$. Hence the monotonicity of $f$ yields $\Lang\bigl(f(\A_{1},\ldots,\A_{k})\bigr) (w) \leq
 f\bigl(\Lang(\A_{1})(w),\ldots,\Lang(\A_{k})(w)\bigr)$. 

 On the other hand, assuming that
 $\Lang(\A_{i}) (w) = \overline{v}_{i}$ for each $i \in \{ 1,\ldots,k\}$, 
 it is not hard to see that
 $\Lang\bigl(f(\A_{1},\ldots,\A_{k})\bigr) (w) \geq f(\overline{v}_{1},\ldots,\overline{v}_{k}) =
 f\bigl(\Lang(\A_{1})(w),\ldots,\Lang(\A_{k})(w)\bigr)$. 
 Here the intuition about the automaton 
  $f(\A_{1},\ldots,\A_{k})$, and especially its state
 $q = \bigl((q_{1},v_{1}),\ldots,(q_{k},v_{k}),b\bigr)
 \in Q$, is as follows.
 \begin{itemize}
  \item The automaton $f(\A_{1},\ldots,\A_{k})$ is essentially a
	synchronous product of $\A_{1},\ldots,\A_{k}$; the state
	$q_{i}\in Q_{i}$ is the current state of the constituent
	automaton $\A_{i}$.
  \item Each  constituent
	automaton $\A_{i}$ is additionally equipped with a register 
	for storing ``the greatest acceptance value that is recently seen.''
The value $v_{i}$ is the one stored in that register. 
  \item The flag $b\in \{\ffalse,\ttrue\}$ decides if the stored
	acceptance value $v_{i}$ is ``exposed'' or
	not. See~(\ref{eq:propClosedUnderIncOperator1}) where the
	acceptance value of the composed automaton 
	 $f(\A_{1},\ldots,\A_{k})$ is nonzero only if $b=\ttrue$.
	Also observe that, in~(\ref{eq:propClosedUnderIncOperator2}),
	the register $v_{i}$ is reset to the current acceptance value
	$F_{i}(q'_{i})$ when the register is exposed (i.e.\ $b=\ttrue$).
 \end{itemize}
Following this intuition, it is not hard to see that the claimed fact  $\Lang\bigl(f(\A_{1},\ldots,\A_{k})\bigr) (w) \geq
 f(\overline{v}_{1},\ldots,\overline{v}_{k}) $ is witnessed by a run
 such that: it does not expose the register values before
all the
 registers acquire the values $\overline{v}_{1},\ldots,\overline{v}_{k}$; and once they
 have all done so, the register values are exposed by setting $b=\ttrue$.

 From the above two inequalities, we conclude that
 $\Lang\bigl(f(\A_{1},\ldots,\A_{k})\bigr) (w) =
 f\bigl(\Lang(\A_{1})(w),\ldots,\Lang(\A_{k})(w)\bigr)$.
\end{proof}

\subsection{Proof of Lem.~\ref{lem:AphiEpsIsFiniteState}}
\label{pf:lemAphiEpsIsFiniteState}
\begin{proof}
 The state space $Q=\mathit{xcl}(\varphi)\times [0,1]^{+}$ of
 $\Ap_{\varphi,\varepsilon}$ is infinite for three reasons: 1) 
 the extended closure
$\mathit{xcl}(\varphi)$ contains $\varphi_{1}\U_{\eta^{+k}}\varphi_{2}$
 for unbounded $k\in \N$  (see~(\ref{eq:defXcl})); 2)  discount factors 
 occurring in $\vec{d}\in[0,1]^{+}$ are multiples of numbers from an infinite set 
 $\{\eta(0),\eta(1),\dotsc\}$; and 3) the length of a discount sequence
 $\vec{d}\in[0,1]^{+}$ is potentially unbounded. 

 We can easily see that the reason 3) is not a problem for us: in the construction of 
$\Ap_{\varphi,\varepsilon}$ (Def.~\ref{def:ABAforLTL}), the length of
 a discount sequence $\vec{d}$ grows only when we encounter negation
(i.e.\ in the definition of $\delta \bigl((\lnot
 \psi,\vec{d}),\sigma\bigr)$). 
 Therefore in a reachable state $(\psi,\vec{d})$ of $\Ap_{\varphi,\varepsilon}$, the
 length of $\vec{d}$ is bounded by the number of negation operators
 occurring in $\varphi$.

 To see that the reasons 1) and 2) are not problematic either, note that
 we obtain new states for these reasons only in the
 clause~(\ref{eq:defDeltaWithinEventHorizon}) of the definition of $\delta\bigl((\psi_{1} \U_{\eta} \psi_{2},
	 \vec{d}),\sigma\bigr)$. This clause is applied only when 
 	 $\eta(0)\cdot \prod_{i = 1}^{n} d_{i} > \varepsilon$, a
 condition satisfied by
 only finitely many reachable states of $\Ap$:
\begin{itemize}
 \item The discount function $\eta$ here is of the form
       $\eta=(\eta')^{+k}$,
       where $\eta'$ occurs in the original formula $\varphi$ and $k\in
       \N$. Since a discounting function $\eta'$ tends to $0$
       (Def.~\ref{def:discountFunction}),
       $\eta(0)=(\eta')^{+k}(0)=\eta'(k)$
       tends to $0$ as $k\to \infty$, too, making only finitely many $k$
       suitable.
 \item Each discount factor $d_{j}$ in $\vec{d}$ is a multiple
       $\eta_{1}(k_{1})\times \cdots\times \eta_{m}(k_{m})$, where
       $\eta_{i}$ is a discounting function occurring in $\varphi$ and
       $k_{i}\in\N$. They must at least satisfy
       $\eta_{i}(k_{i})>\varepsilon$: since $\eta_{i}$ tends to $0$, this allows only finitely many
       choices of $k_{i}$, for each $\eta_{i}$. Furthermore, the
       (necessary) condition that 
       $d_{j}=\eta_{1}(k_{1})\times \cdots\times
       \eta_{m}(k_{m})>\varepsilon$
       bounds the length $m$ of the multiple, too.   
\end{itemize}
%
%
%
\end{proof}

\subsection{Proof of Lem.~\ref{lem:correctnessOfA}}
\label{pf:lemcorrectnessOfA}
\begin{proof}
 In what follows
  let $Q$ denote the state space of $\A_{\varphi, \varepsilon}$;
 $\delta$ denote its transition function;
and $F$ denote its acceptance function.
For each $(\psi,\vec{d}) \in Q$,  we define an alternation $[0,1]$-acceptance automaton 
$\A_{\varphi, \varepsilon}^{(\psi,\vec{d})}$ by changing the initial
 state to $(\psi,\vec{d})$, that is,
$\A_{\varphi, \varepsilon}^{(\psi,\vec{d})} =
 (\mathcal{P}(\AP),Q,\{ (\psi,\vec{d}) \},\delta,F)$. Suppose
 that $\vec{d} = d_{1} d_{2} \ldots d_{n}$. 
 We prove the following more general statement, inductively on the
 construction of $\psi$:
  \begin{equation}\label{eq:ltlToNBA}
    \vec{d} \boxtimes {\sem{ \pi, \psi }} - \varepsilon
\; \leq\; \Lang(\A_{\varphi, \varepsilon}^{(\psi,\vec{d})})(\pi) 
\;\leq\; \vec{d} \boxtimes {\sem{ \pi, \psi }}
  \end{equation}
  for each $\pi \in (\mathcal{P}(\AP))^{\omega}$.

  The cases where $\psi = \true$, $p$, $\psi_{1} \land \psi_{2}$, $\lnot \psi'$ or $\X
 \psi'$ are straightforward. Here we only prove the case where
 $\psi=\lnot\psi'$. 
 By the definition of the automaton $\A_{\varphi,\varepsilon}$ we have
 \begin{math}
  \Lang(\A_{\varphi, \varepsilon}^{(\lnot\psi',\vec{d})})(\pi)
  =
  \Lang(\A_{\varphi, \varepsilon}^{(\psi',\vec{d}1)})(\pi)
 \end{math}, and the latter value lies in the interval
 \begin{math}
  \bigl[\,
(\vec{d}1)\boxtimes\sem{\pi,\psi'}
 - \varepsilon,\,
(\vec{d}1)\boxtimes\sem{\pi,\psi'}
\,
\bigr]
 \end{math}
 by the induction hypothesis. Now we obtain
\begin{align*}
 (\vec{d}1)\boxtimes \sem{\pi,\psi'}
 &\;=\; 
 \vec{d}\boxtimes (1-\sem{\pi,\psi'})
 \;=\; 
 \vec{d}\boxtimes \sem{\pi,\lnot\psi'}\enspace,
\end{align*}
 as required. Here the former equality is due to the definition of
 $\boxtimes$; the latter is the semantics of $\lnot\psi'$.



  Suppose  $\psi = \psi_{1} \U \psi_{2}$; we first deal with the case
 when $|\vec{d}|$ is odd. 
Let $\pi \in (\mathcal{P}(\AP))^{\omega}$. 
We note that, since  $|\vec{d}|$ is odd, the function
 $\vec{d}\boxtimes(\place)\colon [0,1]\to [0,1]$ is monotone and
 continuous (see~(\ref{eq:actionExplicitly})). This is used in:
\begin{equation}\label{eq:actionIsMonotoneAndConti}
   \begin{aligned}
    \vec{d} \boxtimes \sem{ \pi, \psi_{1} \U \psi_{2} } &= \vec{d} \boxtimes \sup_{i \in \N} \bigl\{\,\min \bigl\{\, \sem{ \pi^{i}, \psi_{2} }, \min_{0 \leq j \leq i-1} \sem{ \pi^{j}, \psi_{1} } \,\bigr\} \,\bigr\} \\
	&= \sup_{i \in \N} \bigl\{\,\min \bigl\{\, \vec{d} \boxtimes \sem{ \pi^{i}, \psi_{2} },\, \min_{0 \leq j \leq i-1} \bigl(\vec{d} \boxtimes \sem{ \pi^{j}, \psi_{1} }\bigr) \,\bigr\} \,\bigr\}\enspace.
  \end{aligned}
\end{equation}  
Now let us take a closer look at how the value $\Lang(\A_{\varphi, \varepsilon}^{(\psi_{1} \U
   \psi_{2},\vec{d})})(\pi) $ is defined for an alternating
 $[0,1]$-acceptance automaton $\A_{\varphi, \varepsilon}^{(\psi_{1} \U
   \psi_{2},\vec{d})}$. As seen in
 Def.~\ref{def:alternative_buchi}, the notions of run tree and path
 are Boolean; a non-Boolean value arises for the first time as the
 ``utility'' $F^{\infty}(\rho)$ of a path $\rho$ of a run tree. 
According to Def.~\ref{def:ABAforLTL} of $\A_{\varphi,\varepsilon}$
  (in particular the definition of $\delta \bigl((\psi_{1} \U
 \psi_{2},\vec{d}),\sigma\bigr)$),  any possible run tree $\tau$
 from the state $(\psi_{1} \U
   \psi_{2},\vec{d})$ is of one of the following forms:
\begin{itemize}
 \item the second disjunct ${\delta \bigl((\psi_{1},\vec{d}),\sigma\bigr)} \land
       {(\psi_{1} \U \psi_{2},\vec{d})}$ is chosen all the way
       (Fig.~\ref{fig:untilRunTree}, left),
       or
 \item the first disjunct $        {\delta
       \bigl((\psi_{2},\vec{d}),\sigma\bigr)} $ is eventually hit
       (Fig.~\ref{fig:untilRunTree}, right).
\end{itemize}
\begin{figure}[tbp] 
\includegraphics[width=\textwidth]{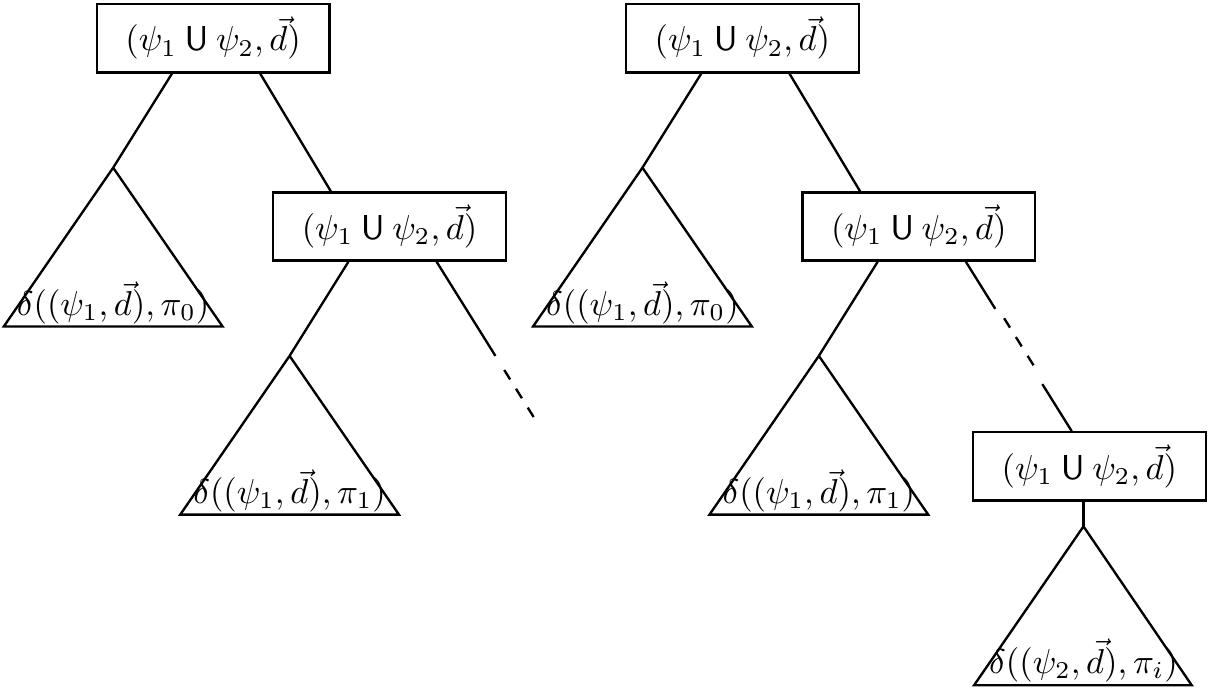}
 \caption{Possible run trees from the state $(\psi_{1} \U
   \psi_{2},\vec{d})$ in $\A_{\varphi,\varepsilon}$, when $|\vec{d}|$ is
 odd}
\label{fig:untilRunTree}
\end{figure}
In the former case, the utility
 $\min_{\rho\in\pathrm(\tau)}F^{\infty}(\rho)$ of such a run tree $\tau$
 is given by
\begin{math}
 \min \bigl\{\, 
 F(\psi_{1}\U \psi_{2},\vec{d}),\,
\inf_{j \in \N} \Lang(\A_{\varphi,
 \varepsilon}^{(\psi_{1},\vec{d})})(\pi^{j}) 
\,\bigr\}
\end{math}, 
where the first value $F(\psi_{1}\U \psi_{2},\vec{d})$ is induced by the
 rightmost path in Fig.~\ref{fig:untilRunTree}, left.  We have
$F(\psi_{1}\U \psi_{2},\vec{d})=0$ by definition
 (see~(\ref{eq:acceptanceFunc})); therefore the utility obtained in this
 case is $0$.

In the latter
 case, assume that the second disjunct 
 $        {\delta
       \bigl((\psi_{2},\vec{d}),\sigma\bigr)} $ 
 is hit at depth $i$. The tree's utility is then given by
\begin{math}
 \min \bigl\{\,
 \Lang(\A_{\varphi, \varepsilon}^{(\psi_{2},\vec{d})})(\pi^{i})
 ,\,
  \min_{0 \leq j \leq i-1} \Lang(\A_{\varphi,
 \varepsilon}^{(\psi_{1},\vec{d})})(\pi^{j})
\, \bigr\}
\end{math} where, again, the first value 
$\Lang(\A_{\varphi, \varepsilon}^{(\psi_{2},\vec{d})})(\pi^{i})$ arises
 from the rightmost path in Fig.~\ref{fig:untilRunTree}, right.

Putting all these together, we have
  \begin{align*}
   & \Lang(\A_{\varphi, \varepsilon}^{(\psi_{1} \U
   \psi_{2},\vec{d})})(\pi) 
\\
 &=
  \sup_{i\in\N}
\bigl(\;
  \min \bigl\{\,  
\Lang(\A_{\varphi,
   \varepsilon}^{(\psi_{2},\vec{d})})(\pi^{i})
,\, 
\min_{0 \leq j \leq i-1} \Lang(\A_{\varphi,
  \varepsilon}^{(\psi_{1},\vec{d})})(\pi^{j})
\, \bigr\}
\;\bigr)
\\
&\in
\left[
\begin{array}{r}
   \sup_{i\in\N}
 \bigl(\;
  \min \bigl\{\,  
 \vec{d}\boxtimes\sem{\pi^{i},\psi_{2}} -\varepsilon
 ,\, 
\min_{0 \leq j \leq i-1} 
 \vec{d}\boxtimes\sem{\pi^{j},\psi_{1}} -\varepsilon
 \, \bigr\}
 \;\bigr)\,,
\phantom{hoge}
\\
   \sup_{i\in\N}
 \bigl(\;
  \min \bigl\{\,  
 \vec{d}\boxtimes\sem{\pi^{i},\psi_{2}} 
 ,\, 
\min_{0 \leq j \leq i-1} 
 \vec{d}\boxtimes\sem{\pi^{j},\psi_{1}} 
 \, \bigr\}
 \;\bigr)
\end{array}
\right]
\\
&
\qquad\qquad\qquad\qquad\qquad\qquad\qquad\qquad\qquad\qquad
\text{by the induction hypothesis}
\\
&=
\bigl[\,
    \vec{d} \boxtimes \sem{ \pi, \psi_{1} \U \psi_{2} }
   -\varepsilon\,,\;
    \vec{d} \boxtimes \sem{ \pi, \psi_{1} \U \psi_{2} } 
\,\bigr]
\qquad\text{by~(\ref{eq:actionIsMonotoneAndConti}),}
  \end{align*}
as required.
  
  Suppose that $\psi = \psi_{1} \U \psi_{2}$ and that $|\vec{d}|$ is
 even. Let $\pi \in (\mathcal{P}(\AP))^{\omega}$. Since $\vec{d} \boxtimes (\place)$ is antitone and continuous, the second
 equality below holds.
  \begin{equation}\label{eq:untilEven}
  \begin{aligned}
    \vec{d} \boxtimes \sem{ \pi, \psi_{1} \U \psi_{2} } &= \vec{d}
   \boxtimes \sup_{i \in \N} \bigl\{\, \min \bigl\{\, \sem{ \pi^{i},
   \psi_{2} }, \min_{0 \leq j \leq i-1} \sem{ \pi^{j}, \psi_{1} }
   \,\bigr\} \,\bigr\} 
  \quad\text{by def.\ of $\sem{ \pi, \psi_{1} \U \psi_{2} }$}
\\
	&= \inf_{i \in \N} \bigl\{\, \max \bigl\{\, \vec{d} \boxtimes \sem{ \pi^{i}, \psi_{2} },\, \max_{0 \leq j \leq i-1} \bigl(\vec{d} \boxtimes \sem{ \pi^{j}, \psi_{1} }\bigr) \,\bigr\} \,\bigr\}\enspace.
  \end{aligned}
  \end{equation}
  We use the following observation. It is a quantitative adaptation of 
  the classic duality between the temporal operators $\U$ and $\R$
 (``release'').
  \begin{sublem}\label{sublem:quantitativeUntilAndRelease}
   Let
$a_{0}, a_{1},\dotsc$ and
 $b_{0}, b_{1},\dotsc$ all be real numbers in $[0,1]$. We have
  \begin{displaymath}
   \inf_{i\in\N}
\bigl(\,
\max\{b_{i}, \max_{0\le j\le i-1} a_{j} \}
\,\bigr)
   \;=\;
\Bigl\{\,
    \sup_{j\in\N}
    \bigl(\,
\min\{
    a_{j},
\min_{0\le i\le j}b_{i}
\}
\,\Bigr)
,
 \inf_{i\in\N} b_{i}
\,\Bigr\}
\enspace,
  \end{displaymath}
that is, denoting binary $\min$ and $\max$ by $\land$ and $\lor$:
  \begin{equation}\label{eq:goalofSublemquantitativeUntilAndRelease}
   \inf_{i\in\N}
\bigl(\,
b_{i}\lor (a_{0}\lor a_{1}\lor \cdots\lor a_{i-1})
\,\bigr)
   \;=\;
\Bigl(\,
    \sup_{j\in\N}
    \bigl(\,
    a_{j}\land 
    (b_{0}\land b_{1}\land\cdots\land b_{j})
\,\bigr)
\,\Bigr)
\lor
 \inf_{i\in\N} b_{i}
\enspace.
  \end{equation}
  \end{sublem}
  \begin{proof} (Of Sublem.~\ref{sublem:quantitativeUntilAndRelease})
   We distinguish two cases. Let us first assume that there exists
   $i\in\N$ such that $b_{i}<a_{0}\lor a_{1}\lor \cdots \lor a_{i-1}$.
   Let $k$ be the least number among such, that is, $k$ satisfies that
   \begin{equation}\label{eq:201410161346}
     b_{k} < a_{0}\lor a_{1}\lor \cdots \lor a_{k-1} 
	 \quad \mbox{ and } \quad 
	 \forall i\in[0,k-1].\; b_{i} \geq a_{0}\lor a_{1}\lor \cdots \lor a_{i-1}\enspace.
   \end{equation}
   Moreover, let $l \in [0,k-1]$ be a number such that $a_{l} = a_{0}\lor a_{1}\lor \cdots \lor a_{k-1}$. 
   We have 
  \begin{align*}
   &\inf_{i\in\N}
\bigl(\,
b_{i}\lor (a_{0}\lor a_{1}\lor \cdots\lor a_{i-1})
\,\bigr) 
\\
   &=\;
 b_{0}
\;\land\;
 (b_{1}\lor a_{0})
 \;\land\;
 (b_{2}\lor a_{0}\lor a_{1})
 \;\land\;
\cdots
\\
&
\qquad\qquad\qquad
 \;\land\;
 (b_{k}\lor a_{0}\lor \cdots\lor a_{k-1})
\land\,
   \inf_{i \geq k+1}
\bigl(\,
b_{i}\lor (a_{0}\lor a_{1}\lor \cdots\lor a_{i-1})
\,\bigr)
\\
   &=\;
b_{0}\land b_{1}\land\cdots\land b_{k-1} \land a_{l} 
\,\land\,
   \inf_{i \geq k+1}
\bigl(\,
b_{i}\lor (a_{0}\lor a_{1}\lor \cdots\lor a_{i-1})
\,\bigr)
\quad\text{by def.\ of $k$, (\ref{eq:201410161346}).}
  \end{align*}
  Since we have
  $a_{l} \leq b_{i}\lor (a_{0}\lor a_{1} \lor \cdots\lor a_{i-1})$
 for each $i \in[k+1,\infty)$,
  \begin{displaymath}
   a_{l} \;\leq\;
   \inf_{i \geq k+1}
\bigl(\,
b_{i}\lor (a_{0}\lor a_{1}\lor \cdots\lor a_{i-1})
\,\bigr)
  \end{displaymath}
  and we obtain
  \begin{equation}\label{eq:infAtFirstCase}
   \inf_{i\in\N}
\bigl(\,
b_{i}\lor (a_{0}\lor a_{1}\lor \cdots\lor a_{i-1})
\,\bigr) \\
   \;=\; 
   b_{0}\land b_{1}\land\cdots\land b_{k-1} \land a_{l}
\enspace.
  \end{equation}
  Now we compare the last value $b_{0}\land b_{1}\land\cdots\land b_{k-1} \land a_{l}$
   with 
 the right-hand side of our goal~(\ref{eq:goalofSublemquantitativeUntilAndRelease}).
By the definition of $k$ and $l$, for each $j \in [0,k-1]$, we have 
   \begin{displaymath}
     a_{j} \leq a_{0}\lor a_{1}\lor \cdots \lor a_{k-1} = a_{l}
	 \quad \mbox{ and } \quad 
	 \forall i\in[j+1,k-1].\; a_{j} \leq a_{0}\lor a_{1}\lor \cdots \lor a_{i-1} \leq b_{i}\enspace,
   \end{displaymath}
yielding
\begin{align*}
 a_{j}&\;\le\; a_{l}\land b_{j+1}\land b_{j+2}\land\cdots\land
 b_{k-1}\enspace, \quad\text{and hence}
\\
 a_{j}\land 
    (b_{0}\land b_{1}\land\cdots\land b_{j}) 
 &\;\le\;
	b_{0}\land b_{1}\land\cdots\land b_{j} \land b_{j+1}\land\cdots\land b_{k-1} \land a_{l}\enspace.
\end{align*}
The last inequality holds for each $j \in [k,\infty)$, too:
   \begin{align*}
   a_{j}\land 
    (b_{0}\land b_{1}\land\cdots\land b_{j}) 
	&\,\leq\,
	b_{0}\land b_{1}\land\cdots\land b_{k-1} \land b_{k} \\
	&\,\leq\,
	b_{0}\land b_{1}\land\cdots\land b_{k-1} \land (a_{0}\lor a_{1}\lor \cdots\lor a_{k-1}) \\
	&\qquad\qquad\qquad\qquad\qquad\qquad\qquad\qquad
	\text{by def.\ of $k$, (\ref{eq:201410161346})} \\
	&\,=\,
	b_{0}\land b_{1}\land\cdots\land b_{k-1} \land a_{l}
	\quad\text{by def.\ of $l$.}
   \end{align*}
 Consequently
   \begin{equation}\label{eq:supLeqInf}
    \sup_{j\in\N}
    \bigl(\,
    a_{j}\land 
    (b_{0}\land b_{1}\land\cdots\land b_{j})
\,\bigr)
 \,\leq\,
 b_{0}\land b_{1}\land\cdots\land b_{k-1} \land a_{l}\enspace.
   \end{equation}
We turn to the other part $     \inf_{i\in\N} b_{i} $
of the right-hand side of~(\ref{eq:goalofSublemquantitativeUntilAndRelease}).
   By the definition of $k$ and $l$, we have $b_{k} \,\leq\, a_{0}\lor a_{1}\lor \cdots\lor a_{k-1} \,=\, a_{l}$. Therefore
   \begin{equation}\label{eq:infLeqInf}
   \begin{aligned}
     \inf_{i\in\N} b_{i} 
	\;\leq\;
	b_{0}\land b_{1}\land\cdots\land b_{k-1} \land b_{k}
	\;\leq\;
	b_{0}\land b_{1}\land\cdots\land b_{k-1} \land a_{l}\enspace.
   \end{aligned}
   \end{equation}
   By~(\ref{eq:supLeqInf}) and~(\ref{eq:infLeqInf}), 
   \begin{equation}\label{eq:leqAtFirstCase}
   \Bigl(\,
    \sup_{j\in\N}
    \bigl(\,
    a_{j}\land 
    (b_{0}\land b_{1}\land\cdots\land b_{j})
\,\bigr)
\,\Bigr)
\lor
 \inf_{i\in\N} b_{i} 
 \;\leq\;
 b_{0}\land b_{1}\land\cdots\land b_{k-1} \land a_{l}\enspace,
   \end{equation}
   on the one hand. 
   On the other hand,  since $l\in[0,k-1]$,
   \begin{equation}\label{eq:geAtFirstCase}
   \begin{aligned}
   \Bigl(\,
    \sup_{j\in\N}
    \bigl(\,
    a_{j}\land 
    (b_{0}\land b_{1}\land\cdots\land b_{j})
\,\bigr)
\,\Bigr)
\lor
 \inf_{i\in\N} b_{i} 
 &\,\geq\,
    \sup_{j\in\N}
    \bigl(\,
    a_{j}\land 
    (b_{0}\land b_{1}\land\cdots\land b_{j})
\,\bigr) \\
 &\,\geq\,
 a_{l}\land 
    (b_{0}\land b_{1}\land\cdots\land b_{l}) \\
 &\,\geq\,
 a_{l} \land (b_{0}\land b_{1}\land\cdots\land b_{k-1})\enspace.
   \end{aligned}
   \end{equation}
   By~(\ref{eq:leqAtFirstCase}) and~(\ref{eq:geAtFirstCase}), 
   \begin{equation}\label{eq:infSupEqualAtFirstCase}
   \Bigl(\,
    \sup_{j\in\N}
    \bigl(\,
    a_{j}\land 
    (b_{0}\land b_{1}\land\cdots\land b_{j})
\,\bigr)
\,\Bigr)
\lor
 \inf_{i\in\N} b_{i} 
 \,=\,
 b_{0}\land b_{1}\land\cdots\land b_{k-1} \land a_{l}\enspace.
   \end{equation}
   By~(\ref{eq:infAtFirstCase}) and~(\ref{eq:infSupEqualAtFirstCase}), 
  \begin{displaymath}
   \inf_{i\in\N}
\bigl(\,
b_{i}\lor (a_{0}\lor a_{1}\lor \cdots\lor a_{i-1})
\,\bigr)
   \;=\;
\Bigl(\,
    \sup_{j\in\N}
    \bigl(\,
    a_{j}\land 
    (b_{0}\land b_{1}\land\cdots\land b_{j})
\,\bigr)
\,\Bigr)
\lor
 \inf_{i\in\N} b_{i}
\enspace.
  \end{displaymath}
This establish the claim, in our first case where there exists
   $i\in\N$ such that $b_{i}<a_{0}\lor a_{1}\lor \cdots \lor a_{i-1}$.

In the other case we assume  that $b_{i} \geq a_{0}\lor a_{1}\lor \cdots \lor a_{i-1}$ for each $i\in\N$. By this assumption, 
  $
b_{i}\lor (a_{0}\lor a_{1}\lor \cdots\lor a_{i-1})
   \;=\; 
   b_{i}
  $
  for each $i \in \N$.
  Therefore
  \begin{equation}\label{eq:infAtSecondCase}
   \inf_{i\in\N}
\bigl(\,
b_{i}\lor (a_{0}\lor a_{1}\lor \cdots\lor a_{i-1})
\,\bigr) \\
   \;=\; 
   \inf_{i\in\N} b_{i}
\enspace.
  \end{equation}
 
Let us now fix $j\in\N$. For each $i \in [j+ 1,\infty)$ we have 
  \begin{math}
    a_{j} \leq a_{0}\lor a_{1}\lor \cdots\lor a_{i-1} \leq b_{i}
  \end{math}, where the latter inequality holds because of the assumption.
 Therefore $a_{j}\le \inf_{i\geq j+1} b_{i}$; this is used in
  \begin{align*}
    a_{j}\land 
    (b_{0}\land b_{1}\land\cdots\land b_{j})
	&\;\leq\;
	(\inf_{i \geq j+1} b_{i}) \land (b_{0}\land b_{1}\land\cdots\land b_{j}) 
	\;=\;
	\inf_{i\in\N} b_{i}\enspace.
  \end{align*}
   This holds for any $j\in\N$; therefore
  $
  \sup_{j\in\N}
    \bigl(\,
    a_{j}\land 
    (b_{0}\land b_{1}\land\cdots\land b_{j})
\,\bigr)
	\,\leq\,
	\inf_{i\in\N} b_{i}$.
 This yields
	\begin{math}
   \Bigl(\,
    \sup_{j\in\N}
    \bigl(\,
    a_{j}\land 
    (b_{0}\land b_{1}\land\cdots\land b_{j})
\,\bigr)
\,\Bigr)
\lor
 \inf_{i\in\N} b_{i} 
 \,=\,
 \inf_{i\in\N} b_{i} 
   \end{math}, which is combined with~(\ref{eq:infAtSecondCase}) and 
   proves the claim~(\ref{eq:goalofSublemquantitativeUntilAndRelease}). 
  This concludes the proof of
   Sublem.~\ref{sublem:quantitativeUntilAndRelease}. 
  \end{proof}

 We turn back to the proof of Lem.~\ref{lem:correctnessOfA}.
By letting $a_{j}=\vec{d} \boxtimes \sem{ \pi^{j}, \psi_{1} }$ 
and $b_{i}=\vec{d} \boxtimes \sem{ \pi^{i}, \psi_{2} }$ in
 Sublem.~\ref{sublem:quantitativeUntilAndRelease}, we obtain
  \begin{equation}\label{eq:distrInfSup}
  \begin{aligned}
	& \inf_{i \in \N} \bigl(\, \max \bigl\{\, \vec{d} \boxtimes \sem{ \pi^{i}, \psi_{2} },\, \max_{0 \leq j \leq i-1} \bigl(\vec{d} \boxtimes \sem{ \pi^{j}, \psi_{1} }\bigr) \,\bigr\} \,\bigr)  \\
	&=\max \bigl\{\, \sup_{j \in \N} \bigl(\, \min \bigl\{\, \vec{d} \boxtimes \sem{ \pi^{j}, \psi_{1} },\, \min_{0 \leq i \leq j} \bigl(\vec{d} \boxtimes \sem{ \pi^{i}, \psi_{2} }\bigr) \,\bigr\} \,\bigr),\, \inf_{i \in \N} (\vec{d} \boxtimes \sem{ \pi^{i}, \psi_{2} }) \,\bigr\}\enspace.
  \end{aligned}
  \end{equation}

  By~(\ref{eq:untilEven}) and~(\ref{eq:distrInfSup}), we have
  \begin{equation}\label{eq:actionIsAntiMonotoneAndConti}
\begin{aligned}
&     \vec{d} \boxtimes \sem{ \pi, \psi_{1} \U \psi_{2} } 
\\
&= \max \bigl\{\, \sup_{j \in \N} \bigl\{\, \min \bigl\{\, \vec{d} \boxtimes \sem{ \pi^{j}, \psi_{1} },\, \min_{0 \leq i \leq j} \bigl(\vec{d} \boxtimes \sem{ \pi^{i}, \psi_{2} }\bigr) \,\bigr\} \,\bigr\},\, \inf_{i \in N} (\vec{d} \boxtimes \sem{ \pi^{i}, \psi_{2} }) \,\bigr\}\enspace.
\end{aligned}  
\end{equation}
  Let us now look at the value 
\begin{math}
 \Lang(\A_{\varphi, \varepsilon}^{(\psi_{1} \U
   \psi_{2},\vec{d})})(\pi) 
\end{math}. 
We analyze  possible run trees $\tau$ starting from the state
$(\psi_{1}\U\psi_{2},\vec{d})$, much like in the previous case where
 $|\vec{d}|$ is odd (in the current case it is even). It is easily seen
 from Def.~\ref{def:ABAforLTL} that $\tau$ is of one of the forms shown
 in Fig.~\ref{fig:untilRunTreeEven}.
\begin{figure}[tbp] 
\includegraphics[width=\textwidth]{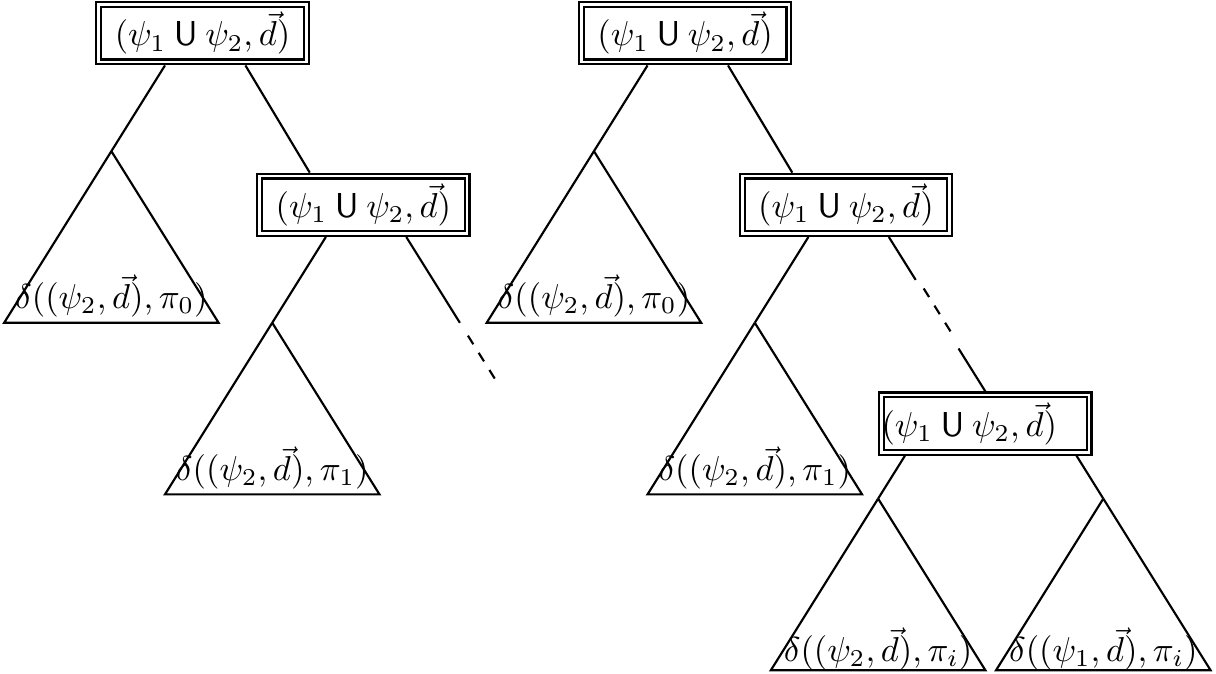}
 \caption{Possible run trees from the state $(\psi_{1} \U
   \psi_{2},\vec{d})$ in $\A_{\varphi,\varepsilon}$, when $|\vec{d}|$ is
even. The double-lined nodes  have
 the acceptance value $1$.}
\label{fig:untilRunTreeEven}
\end{figure}
\begin{itemize}
 \item 
  If $\tau$ is of the form in Fig.~\ref{fig:untilRunTreeEven} on the left, its utility
  $\min_{\rho\in\pathrm(\tau)}F^{\infty}(\rho)$ 
  is 
 \begin{math}
  \inf_{j \in \N} \Lang(\A_{\varphi, \varepsilon}^{(\psi_{2},\vec{d})})(\pi^{j})
 \end{math}; note that the rightmost path's value of $F^{\infty}$ is $1$
 and hence does not appear here. 
 \item 
  If $\tau$ is of the form in Fig.~\ref{fig:untilRunTreeEven} on the
       right, its utility   $\min_{\rho\in\pathrm(\tau)}F^{\infty}(\rho)$ 
  is given by
  \begin{math}
   \min \bigl\{\, \Lang(\A_{\varphi, \varepsilon}^{(\psi_{1},\vec{d})})(\pi^{i}),\, \min_{0 \leq j \leq i} \Lang(\A_{\varphi, \varepsilon}^{(\psi_{2},\vec{d})})(\pi^{j}) \,\bigr\}
  \end{math} where $i$ is the depth of the last occurrence of the node $(\psi_{1}\U\psi_{2},\vec{d})$.
\end{itemize}
The value 
$\Lang(\A_{\varphi, \varepsilon}^{(\psi_{1} \U
   \psi_{2},\vec{d})})(\pi) $
is defined as the supremum of these utilities.  Therefore:
  \begin{align*}
   & \Lang(\A_{\varphi, \varepsilon}^{(\psi_{1} \U
   \psi_{2},\vec{d})})(\pi) 
\\
 &= 
  \max \bigl\{\, \sup_{i \in \N} 
  \bigl(\; \min \bigl\{\, \Lang(\A_{\varphi, \varepsilon}^{(\psi_{1},\vec{d})})(\pi^{i}),\, \max_{0 \leq j \leq i} \Lang(\A_{\varphi, \varepsilon}^{(\psi_{2},\vec{d})})(\pi^{j}) \,\bigr\} \;\bigr),\; \inf_{j \in \N} \Lang(\A_{\varphi, \varepsilon}^{(\psi_{2},\vec{d})})(\pi^{j}) \,\bigr\} 
  \\
&\in
\left[
\begin{array}{r}
  \displaystyle \max \bigl\{\,
   \sup_{j\in\N}
 \bigl(\;
  \min \bigl\{\,  
 \vec{d}\boxtimes\sem{\pi^{j},\psi_{1}}
 ,\, 
\min_{0 \leq i \leq j} 
 \vec{d}\boxtimes\sem{\pi^{i},\psi_{2}}
 \, \bigr\}
 \;\bigr)\;,
 \inf_{i \in \N}
 \vec{d}\boxtimes\sem{\pi^{i},\psi_{2}} 
 \,\bigr\} - \varepsilon\,,
\phantom{hoge}
\\
  \displaystyle \max \bigl\{\,
   \sup_{j\in\N}
 \bigl(\;
  \min \bigl\{\,  
 \vec{d}\boxtimes\sem{\pi^{j},\psi_{1}}
 ,\, 
\min_{0 \leq i \leq j} 
 \vec{d}\boxtimes\sem{\pi^{i},\psi_{2}}
 \, \bigr\}
 \;\bigr)\;,
 \inf_{i \in \N}
 \vec{d}\boxtimes\sem{\pi^{i},\psi_{2}} 
 \,\bigr\}
\end{array}
\right]
\\
&
\qquad\qquad\qquad\qquad\qquad\qquad\qquad\qquad\qquad\qquad
\text{by the induction hypothesis}
\\
&=
\bigl[\,
    \vec{d} \boxtimes \sem{ \pi, \psi_{1} \U \psi_{2} }
   -\varepsilon\,,\;
    \vec{d} \boxtimes \sem{ \pi, \psi_{1} \U \psi_{2} } 
\,\bigr]
\qquad\text{by~(\ref{eq:actionIsAntiMonotoneAndConti}),}
\end{align*}
concluding the case when  $\psi = \psi_{1} \U \psi_{2}$ and $|\vec{d}|$ is
 even. 
  
  Suppose that $\psi = \psi_{1} \U_{\eta^{+k}} \psi_{2}$ and that
  ${|\vec{d}|}$ is odd. 
 We prove the claim by induction on $k$, going backwards, decrementing
 $k$ starting from the event
 horizon towards $k=0$. As the base case, assume that $k$ is big enough and we are
 beyond the event horizon, that is,  $\eta (k) \cdot \prod_{i =
  1}^{n} d_{i} \leq \varepsilon$.
Let $\pi \in
  (\mathcal{P}(\AP))^{\omega}$ and
  $\vec{d}=d_{1}\dotsc d_{n}$. 
 Then we have
 $\sem{ \pi, \psi } \cdot \prod_{i = 0}^{n} d_{i}
=
\sem{ \pi,  \psi_{1} \U_{\eta^{+k}} \psi_{2} } \cdot \prod_{i = 0}^{n} d_{i}
\leq \varepsilon$, by   Lem.~\ref{lem:maxTruthValOfDiscountedUntil}
 and that $\eta^{+k}(0)=\eta(k)$. It follows
 from~(\ref{eq:actionExplicitly}) that we have 
$0 \leq \vec{d} \boxtimes \sem{ \pi, \psi } - \vec{d} \boxtimes 0 \leq
 \varepsilon$
(note that $n=|\vec{d}|$ is odd). Therefore
\begin{align*}
 \Lang(\A_{\varphi, \varepsilon}^{(\psi_{1} \U_{\eta^{+k}}
 \psi_{2},\vec{d})})(\pi)
&=
 \vec{d}\boxtimes 0
 \qquad\text{by~(\ref{eq:defDeltaBeyondEventHorizon})}
\\
&\in
\bigl[\,\vec{d} \boxtimes \sem{ \pi, \psi }-\varepsilon,\,
\vec{d} \boxtimes \sem{ \pi, \psi }
\,\bigr]\enspace.
\end{align*}

Now, as the step case, assume that  $\eta (k) \cdot \prod_{i =
  1}^{n} d_{i} > \varepsilon$ and that the claim has been shown for
 $k+1$. The  analogue below of
 $\psi_{1}\U\psi_{2}\cong \psi_{2}\lor \bigl(\psi_{1}\land \X(\psi_{1}\U\psi_{2}) \bigr)$ follows easily from Def.~\ref{def:LTLdDSemantics}:
\begin{displaymath}
 \sem{\pi,\psi_{1} \U_{\eta^{+k}} \psi_{2}}
\;=\;
 \max\bigl\{\,
 \eta^{+k}(0)\cdot
\sem{\pi,\psi_{2}},\,
 \min\bigl\{\,
 \eta^{+k}(0)\cdot
\sem{\pi,\psi_{1}},\,
 \sem{\pi^{1}, \psi_{1} \U_{\eta^{+(k+1)}} \psi_{2}}
\,\bigr\}
\, \bigr\}\enspace.
\end{displaymath}
Therefore
\begin{equation}\label{eq:201410152355}
 \begin{aligned}
& \vec{d}\boxtimes
 \sem{\pi,\psi_{1} \U_{\eta^{+k}} \psi_{2}}
\\
&=
 \max\bigl\{\,
 \vec{d}\boxtimes\bigl(
\eta^{+k}(0)\cdot
\sem{\pi,\psi_{2}}
\bigr)
,\,
\\
&\qquad
 \min\bigl\{\,
 \vec{d}\boxtimes\bigl(
 \eta^{+k}(0)\cdot
\sem{\pi,\psi_{1}}
\bigr)
,\,
 \vec{d}\boxtimes
 \sem{\pi^{1}, \psi_{1} \U_{\eta^{+(k+1)}} \psi_{2}}
\,\bigr\}
\, \bigr\}
\\
&=
 \max\bigl\{\,
 \bigl(\vec{d}
\odot
\eta^{+k}(0)\bigr)
\boxtimes
\sem{\pi,\psi_{2}}
,\,
\\&\qquad
 \min\bigl\{\,
\bigl( \vec{d}\odot  \eta^{+k}(0)
\bigr)
\boxtimes
\sem{\pi,\psi_{1}}
,\,
 \vec{d}\boxtimes
 \sem{\pi^{1}, \psi_{1} \U_{\eta^{+(k+1)}} \psi_{2}}
\,\bigr\}
\, \bigr\}\enspace,
\end{aligned}\end{equation}
where the first equality is due to the monotonicity of
 $\vec{d}\boxtimes(\place)$, and the second is by~(\ref{eq:odotAndBoxtimes}).
Now
\begin{align*}
& 
 \Lang(\A_{\varphi, \varepsilon}^{(\psi_{1} \U_{\eta^{+k}}
 \psi_{2},\vec{d})})(\pi)
 \\
&=
 \max\bigl\{\,  
 \Lang(\A_{\varphi, \varepsilon}^{(
 \psi_{2},\vec{d}\odot \eta^{+k}(0))})(\pi)
 ,\,
 \min\bigl\{\,
 \Lang(\A_{\varphi, \varepsilon}^{(
 \psi_{1},\vec{d}\odot \eta^{+k}(0))})(\pi)
,\, 
 \Lang(\A_{\varphi, \varepsilon}^{(\psi_{1} \U_{\eta^{+(k+1)}}
 \psi_{2},\vec{d})})(\pi^{1})
\,\bigr\}
\,\bigr\}
\end{align*}
by Def.~\ref{def:ABAforLTL}. 
By the induction hypothesis (the claim has been shown for simpler
 formulas as well as $\psi_{1} \U_{\eta^{+(k+1)}}
 \psi_{2}$), a lower bound of the above value is given by
\begin{align*}
& \max\Bigl\{\,  
\Bigl( \bigl(\vec{d}
\odot
\eta^{+k}(0)\bigr)
\boxtimes
\sem{\pi,\psi_{2}}
\Bigr)
-\varepsilon
 ,\,
\\
&\qquad
 \min\Bigl\{\,
\Bigl( 
\bigl( \vec{d}\odot  \eta^{+k}(0)
\bigr)
\boxtimes
\sem{\pi,\psi_{1}}
\Bigr)
-\varepsilon,\,
\Bigl(
 \vec{d}\boxtimes
 \sem{\pi^{1}, \psi_{1} \U_{\eta^{+(k+1)}} \psi_{2}}
\Bigr)
-\varepsilon
\,\Bigr\}
\,\Bigr\}
\\
&=
\vec{d}\boxtimes
 \sem{\pi,\psi_{1} \U_{\eta^{+k}} \psi_{2}}
-\varepsilon
\qquad\text{by~(\ref{eq:201410152355}).}
\end{align*}
Similarly an upper bound $\vec{d}\boxtimes
 \sem{\pi,\psi_{1} \U_{\eta^{+k}} \psi_{2}}
$ is obtained by the induction hypothesis
 and~(\ref{eq:201410152355}). This proves the claim.

  

The remaining case where 
 $\psi = \psi_{1} \U_{\eta^{+k}} \psi_{2}$ and 
  ${|\vec{d}|}$ is even is similar to the last case. 
We describe only the base case of induction, where $k$ is  big enough so that
 $\eta (k) \cdot \prod_{i =
  1}^{n} d_{i} \leq \varepsilon$. By
 Lem.~\ref{lem:maxTruthValOfDiscountedUntil}
 we have $\sem{\pi,\psi}\in [0, \eta^{k}(0)]$; therefore
\begin{displaymath}\textstyle
 0
\;\le\;
 \eta(k)-\sem{\pi,\psi}
\;\le\; \eta(k)
 \;\le\; \varepsilon /\prod_{i =
  1}^{n} d_{i}\enspace.
\end{displaymath}
By~(\ref{eq:actionExplicitly}) and that $n$ is even, we have
\begin{align*}
 \vec{d}\boxtimes\sem{\pi,\psi} -\vec{d}\boxtimes \eta(k)
 \;&=\;
\bigl(
\prod_{i =
  1}^{n} d_{i} 
\bigr)\cdot
\bigl(\,\eta(k)-\sem{\pi,\psi}\,\bigr)
\;\in \; [0,\varepsilon]\enspace.
\end{align*}
Hence
\begin{align*}
 \Lang(\A_{\varphi, \varepsilon}^{(\psi_{1} \U_{\eta^{+k}}
 \psi_{2},\vec{d})})(\pi)
&=
 \vec{d}\boxtimes \eta(k)
 \qquad\text{by~(\ref{eq:defDeltaBeyondEventHorizon})}
\\
&\in
\bigl[\,\vec{d} \boxtimes \sem{ \pi, \psi }-\varepsilon,\,
\vec{d} \boxtimes \sem{ \pi, \psi }
\,\bigr]\enspace.
\end{align*}
This concludes the proof.
\end{proof}

\subsection{Proof of Lem.~\ref{lem:fromOptimalInATimesKToOptimalInA}}
\label{pf:lemfromOptimalInATimesKToOptimalInA}
\begin{proof}
 It follows easily from the definition that there is a bijective
 correspondence between: a run $\zeta=
 (q'_{0}, s'_{0})\,\bullet\,
 (q'_{1}, s'_{1})\,\bullet\,
 \dotsc
$ of $\A\times\K$; and a pair $(\xi,\rho)$ of
 a path $\xi=s'_{0}s'_{1}\dotsc\in\pathrm(\K)$ of $\K$ and 
 a run $\rho$ over $\lambda(\xi)$ of $\A$. Moreover, 
 the acceptance value of $\zeta$ in $\A\times\K$ is equal to
 that of $\rho$ in $\A$. The claim follows immediately.
\end{proof}

\subsection{Proof of Thm.~\ref{thm:mainWithoutPropositional}}
\label{pf:thmmain}
\begin{proof}
 \begin{align*}
  \sem{s_{0}s_{1}\dotsc,\varphi}
  \;&\geq\;
  \Lang(\Ana_{\varphi,\varepsilon})\bigl(\,\lambda(s_{0})\lambda(s_{1})\dotsc\,\bigr)
  \qquad\text{by Cor.~\ref{cor:NBAforLTL}}
 \\
 &=\;
  \max_{\xi\in\pathrm(\K)}
  \Lang(\Ana_{\varphi,\varepsilon})\bigl(\,\lambda(\xi)\,\bigr)
  \qquad\text{by Lem.~\ref{lem:fromOptimalInATimesKToOptimalInA}}
 \\
 &\geq\;
  \sup_{\xi\in\pathrm(\K)}
\sem{\xi,\varphi}
-\varepsilon
  \qquad\text{by Cor.~\ref{cor:NBAforLTL}.}
 \end{align*}
The solution $s_{0}s_{1}\dotsc$  thus obtained arises from a lasso
 computation of $\Ana_{\varphi,\varepsilon}\times\K$ (by the algorithm
 in Lem.~\ref{lem:lassoOptimalityForQuantitativeAcceptAutom}), hence is
 ultimately periodic.
\end{proof}

\subsection{Proof of Prop.~\ref{prop:sizeOfAltAutom}}
\label{pf:propsizeOfAltAutom}
\begin{proof}
  In the proof of Lem.~\ref{lem:correctnessWithPropositional}, we construct $\A_{\varphi,\varepsilon}^{\vec{d}}$ inductively.
 We shall therefore prove, inductively on the construction on $\varphi$, that the size of the state space of $\A_{\varphi,\varepsilon}^{\vec{d}}$ is singly exponential in $|\langle \varphi
 \rangle|$ and in the length of the description of $\varepsilon$.
 
 In the case where $\varphi =
 \true,p,\varphi_{1}\land\varphi_{2},\lnot\varphi',\X\varphi'$ or
 $\varphi_{1} \U \varphi_{2}$, the claim is obvious.
 
 Suppose that $\varphi = \varphi_{1} \U_{\expo_{\lambda}^{+k}}
 \varphi_{2}$ where $\lambda \in (0,1)$. Let $k_{\max} = \lceil
 \log_{\lambda}\varepsilon\rceil + 1$. Recall that
 the construction of 
 $\A_{\varphi_{1} \U_{\expo_{\lambda}^{+k}}
 \varphi_{2},\varepsilon}^{\vec{d}}$ in
 Lem.~\ref{lem:correctnessWithPropositional}
 is by backward induction on $k$, from $k=k_{\max}$ to $k=0$. 
 In the base case when $k = k_{\max}$,  we have
 $\expo_{\lambda}^{+k} (0) \leq \varepsilon$ (beyond the event
 horizon);  in this case
 the size of the state space of $\A_{\varphi,\varepsilon}^{\vec{d}}$ is
 one. In the step case, the state space of
 $\A_{\varphi_{1} \U_{\expo_{\lambda}^{+k}}
 \varphi_{2},\varepsilon}^{\vec{d}}$ is the union of: 
 those of the two automata for $\varphi_{1}$ and $\varphi_{2}$;
 that of the automaton $\A_{\varphi_{1} \U_{\expo_{\lambda}^{+(k+1)}} \varphi_{2},\varepsilon}^{\vec{d}}$; and 
 the singleton of the initial state of
 $\A_{\varphi,\varepsilon}^{\vec{d}}$. 
 Overall, the state space of
 $\A_{\varphi_{1} \U_{\expo_{\lambda}^{+k}}
 \varphi_{2},\varepsilon}^{\vec{d}}$ increases as $k$ decreases, and the
 maximum is when $k=0$---in which case the state space of
 $\A_{\varphi_{1} \U_{\expo_{\lambda}}
 \varphi_{2},\varepsilon}^{\vec{d}}$ is roughly 
$\mathcal{O} (k_{\max})  = \mathcal{O} (\lceil
 \log_{\lambda}\varepsilon\rceil + 1)$ copies of
those of the two automata for $\varphi_{1}$ and $\varphi_{2}$. 
 Now we appeal to the fact used in~\cite{AlmagorBK14} that the value
 $k_{\max} \sim \log_{\lambda} \varepsilon = \log \varepsilon / \log \lambda$ is
 polynomial in the length of the description of $\lambda$---hence in 
$|\langle\varphi\rangle|$---and
 $\varepsilon$.\footnote{It is not explicit in~\cite{AlmagorBK14} what
 is meant by the description length of $\lambda\in(0,1)$. For the claimed
 fact to be true---that
 $\log_{\lambda} \varepsilon = \log \varepsilon / \log \lambda$ is
 polynomial in the length of the description of $\lambda$---we expect it to be $a+b$ where $\lambda=a/b$.
 For example, when $\lambda=1-\frac{1}{b}$,  we have
 \begin{math}
    \log_{\lambda} \varepsilon 
	= \frac{\log \varepsilon}{\log \lambda} 
	= \frac{\log \varepsilon}{\log (1-\frac{1}{b})} 
    = \frac{-\log \epsilon}{\log b - \log (b-1)} 
	\leq b \cdot (-\log \varepsilon)
   \end{math}
  where for the last inequality we used $(\log x)'=\frac{1}{x}$. 
  This is linear in $b$.
 }
 By this fact and the induction hypothesis, the size of the state space of $\A_{\varphi,\varepsilon}^{\vec{d}}$ is singly exponential in $|\langle \varphi
 \rangle|$ and in the length of the description of $\varepsilon$.
 
 Suppose that $\varphi = \varphi_{1} \oplus \varphi_{2}$. Since
 $(\vec{d}\boxtimes v_{1} - \varepsilon) \oplus (\vec{d}\boxtimes v_{2}
 - \varepsilon) = \vec{d}\boxtimes(v_{1} \oplus v_{2}) - \varepsilon$,
 we have $\A_{\varphi,\varepsilon}^{\vec{d}}$ coincide with
 $\A_{\varphi_{1},\varepsilon}^{\vec{d}} \oplus
 \A_{\varphi_{2},\varepsilon}^{\vec{d}}$---where the latter is defined
 in Prop.~\ref{prop:ClosedUnderIncOperator}. (We note that the
 construction  in Prop.~\ref{prop:ClosedUnderIncOperator} can be readily
 adapted to
 \emph{alternating} $[0,1]$-acceptance automata, too.) Hence the size of the state space of $\A_{\varphi,\varepsilon}^{\vec{d}}$ is polynomial in those of $\A_{\varphi_{1},\varepsilon}^{\vec{d}}$ and $\A_{\varphi_{2},\varepsilon}^{\vec{d}}$. By the induction hypothesis, the size of the state space of $\A_{\varphi,\varepsilon}^{\vec{d}}$ is singly exponential in $|\langle \varphi
 \rangle|$ and in the length of the description of $\varepsilon$.
\end{proof}

\subsection{Proof of Thm.~\ref{thm:mainComplexity}}
\begin{proof}
 The construction in Prop.~\ref{prop:ABAtoNBA} (from $\A_{\varphi,\varepsilon}$ to
 $\Ana_{\varphi,\varepsilon}$) results in $\Ana_{\varphi,\varepsilon}$
 that is exponentially bigger than $\A_{\varphi,\varepsilon}$; the size of 
 the product $\Ana_{\varphi,\varepsilon}\times\K$
 (Def.~\ref{def:productAutomata}) is linear in those of
 $\Ana_{\varphi,\varepsilon}$ and $\K$; and finding
 an optimal run by
 Lem.~\ref{lem:lassoOptimalityForQuantitativeAcceptAutom} is in
 NLOGSPACE. Combined with Prop.~\ref{prop:sizeOfAltAutom}, the overall
 complexity is EXPSPACE in $|\langle \varphi
 \rangle|$ and NLOGSPACE in the size of $\K$.
\end{proof}

\subsection{Proof of Thm.~\ref{thm:mainComplexityWithoutPropositional}}
Firstly we give an alternative proof to the following statement (that is
a restriction of Prop.~\ref{prop:sizeOfAltAutom}). It is used in the
proof of Thm.~\ref{thm:mainComplexityWithoutPropositional}.

\begin{sublem}[size of  $\A_{\varphi,\varepsilon}$, for $\LTLd{\Dexp,\emptyset}$]\label{sublem:sizeOfAltAutomWithoutPropositional}
  Let $\varphi$ be an $\LTLd{\Dexp,\emptyset}$ formula and $\varepsilon\in(0,1)\cap\Q$
 be a positive rational number.  The size of the state space of
 the alternating $[0,1]$-acceptance automaton
 $\A_{\varphi,\varepsilon}$ is singly exponential in $|\langle \varphi
 \rangle|$ and in the length of the description of $\varepsilon$. \qed
\end{sublem}
\begin{proof} (Of Sublem.~\ref{sublem:sizeOfAltAutomWithoutPropositional})
   Recall that a state of $\A_{\varphi,\varepsilon}$ is a pair
 $(\psi,\vec{d})$ 
 of $\psi \in \mathit{xcl} (\varphi)$ and $\vec{d} \in [0,1]^{+}$. 
 We first claim that the number of different $\psi$'s is polynomial in 
 $|\langle\varphi\rangle|$ and $\log\varepsilon$. The claim is obvious
 except for the number of the formulas $\psi$ of the form
 \begin{math}
 \psi_{1} \U_{\eta^{+i}} \psi_{2}
 \end{math}, for varying $i\in\N$.
 Let
 $\lambda_{0}$ be the maximum number 
 in $\varphi$ 
 used as the base of an exponential discounting function.
 For each subformula $\psi_{1} \U_{\eta} \psi_{2}$ of $\varphi$,
 the numbers $i$ for which we have a state
 \begin{math}
  \bigl(\,\psi_{1} \U_{\eta^{+i}} \psi_{2},
 \vec{d}
 \,\bigr)
 \end{math}
 in $\A_{\varphi,\varepsilon}$ is bounded by $1 + \lceil \log_{\lambda_{0}}
 \varepsilon\rceil$. Now we appeal to the fact used in~\cite{AlmagorBK14} that the value
 $\log_{\lambda_{0}} \varepsilon = \log \varepsilon / \log \lambda_{0}$ is
 polynomial in the length of the description of $\lambda_{0}$---hence in 
 $|\langle\varphi\rangle|$---and
 $\varepsilon$.

 Our second claim is that the number of different $\vec{d}$'s occurring
 in states of  $\A_{\varphi,\varepsilon}$ is exponential in
 $|\langle\varphi\rangle|$ and the description length of 
 $\varepsilon$, hence is the bottleneck in complexity. The length of a discount sequence $\vec{d}$ 
 is bounded by the number of negations in $\varphi$, therefore by
  $|\langle\varphi\rangle|$. Each entry $d_{i}$ is a multiple
 $\lambda_{i_{1}}\lambda_{i_{2}}\dotsc \lambda_{i_{m}}$ of 
 different discounting bases $\lambda_{j}$ (there are 
 at most   $|\langle\varphi\rangle|$-many such), and since its value 
 must be bigger than $\varepsilon$,  the length $m$ of such a multiple is at most
 $\log_{\lambda_{0}}\varepsilon$. Therefore
 the number of candidates for $d_{i}=\lambda_{i_{1}}\lambda_{i_{2}}\dotsc \lambda_{i_{m}}$ is bounded by
 $|\langle\varphi\rangle|^{\log_{\lambda_{0}}\varepsilon}$; appealing to
 the fact (see~\cite{AlmagorBK14}) that $\log_{\lambda} \varepsilon = \log \varepsilon / \log \lambda$ is
 polynomial in the length of the description of $\lambda$ and
 $\varepsilon$, we obtain the claim.
\end{proof}

\begin{proof} (Of Thm.~\ref{thm:mainComplexityWithoutPropositional}, sketch) 
 We describe how to avoid 
the exponential blowup in the translation from 
$\A_{\varphi,\varepsilon}$ to  $\Ana_{\varphi,\varepsilon}$.

 Looking at the construction of Prop.~\ref{prop:ABAtoNBA} in case of
 $\A=\A_{\varphi,\varepsilon}$, we have $V_{Q}=\{0,1\}$, therefore
\begin{equation}\label{eq:201410162348}
 Q'\;=\;
{\mathcal{P} (Q \times 2)} \times V_{\delta} \times 2
 \;\cong\;
 \bigl(\mathcal{P} (Q)\bigr)^{2} \times V_{\delta} \times 2\enspace.
\end{equation}
Here the original state space $Q$ is bounded by 
 $\mathit{xcl}_{\varepsilon} (\varphi)\times
 |\langle\varphi\rangle|^{\log_{\lambda_{0}}\varepsilon}$, where
 \begin{displaymath}
  \mathit{xcl}_{\varepsilon} (\varphi) 
 \;=\; \mathit{xcl}(\varphi) \setminus \{ \varphi_{1} \U_{\eta} \varphi_{2} \in \mathit{xcl} (\varphi) \mid \eta(0) < \varepsilon\}
 \end{displaymath} is a finite set and the second component 
$ |\langle\varphi\rangle|^{\log_{\lambda_{0}}\varepsilon}$ is from the
 proof of Prop.~\ref{sublem:sizeOfAltAutomWithoutPropositional}.

The optimization lies in the reduction of
$\mathcal{P} (Q)$
that occurs in~(\ref{eq:201410162348}) to 
\begin{equation}\label{eq:QfiveTimes}
\bigl((Q \times Q^{2} \times Q^{2}) \cup \{ \bullet \} \bigr)^{\mathit{xcl}_{\varepsilon} (\varphi)} \enspace,
\end{equation} hence from a double exponential to a single
 exponential; 
recall from the proof  
 of Prop.~\ref{sublem:sizeOfAltAutomWithoutPropositional} that $Q$ is exponential and $  \mathit{xcl}_{\varepsilon}
 (\varphi) $ is polynomial, in $|\langle\varphi\rangle|$ and the
 description length of $\varepsilon$. 

The reduction is done concretely
 as follows. Given a set 
\begin{equation}\label{eq:aSetToBeSuppressed}
 \bigl\{\,
(\psi,\vec{d_{1}}),\,
(\psi,\vec{d_{2}}),\,\dotsc,
(\psi,\vec{d_{m}})
\,\bigr\}
\end{equation}
 of states of $Q$ with a common first component $\psi$, we suppress the
 set into the
 function 
 \begin{equation}\label{eq:piecewiseLinear}
   (\vec{d_{1}}\land\cdots\land\vec{d_{m}})\boxtimes(\place)
   \;\colon\;
   v\;\longmapsto \;
   \min\{\vec{d_{1}}\boxtimes v,\cdots,\vec{d_{m}}\boxtimes v\}\
 \end{equation}
  that does the same job. The latter is a piecewise linear function on
 $[0,1]$ and hence is presented as a disjunction of pairs $(f_{i}, [l_{i},
 r_{i}])$ of  a linear function $f_{i}$ and its domain (here $l_{i},
 r_{i}\in(0,1)$). Now $f_{i}$ is represented by some discount sequence
 so
 there are at most $|Q|$-many of them. A point $l_{i}\in [0,1]$  is expressed as the cross point of two linear functions,
 each represented by a discount sequence.
The same goes for $r_{i}$. 
 Moreover, disjunction is taken out of a single state in the resulting
 automaton---from alternating to non-alternating we only need to bundle
 up states in conjunction. 
In summary, to express 
 the piecewise linear function in~(\ref{eq:piecewiseLinear}) we need: $Q$ to
 represent $f_{i}$; $Q^{2}$ to represent $l_{i}$; and
$Q^{2}$ to represent $r_{i}$, resulting in $Q\times Q^{2}\times Q^{2}$
 in~(\ref{eq:QfiveTimes}). 

 We consider all those sets in the form
 of~(\ref{eq:aSetToBeSuppressed}), therefore we need 
$Q\times Q^{2}\times Q^{2}$ for each formula $\psi\in\mathit{xcl}_{\varepsilon}
 (\varphi)$. The set $\{\bullet\}$ is in~(\ref{eq:QfiveTimes}) to take
 care of the case when the set~(\ref{eq:aSetToBeSuppressed}) for the
 formula $\psi$ is empty.
\end{proof}

\section{Reduction of Fuzzy Automata to $[0,1]$-Acceptance Automata}
\label{appendix:fuzzyAndZeroOne}
A generalization of $[0,1]$-acceptance automaton
  is naturally obtained by making transitions also
$[0,1]$-weighted. The result is called \emph{fuzzy automaton}
  and studied e.g.\ in~\cite{Rahonis05}.
  Here we show that this generalization does not add expressivity. In
  fact we prove a more general result, parametrizing $[0,1]$
  into a general semiring 
  $\mathbb{K}$ (under certain conditions).


We
follow~\cite{DrosteP07} and
impose certain conditions on a semiring $K$ of weights.
\begin{defi}[\cite{DrosteP07}]\label{def:conditionsOnSemiring}
A tuple $\mathbb{K}=(K,\leq,+,\cdot,0,1)$ is called an \emph{ordered semiring} if $(K,+,\cdot,0,1)$ is a semiring, $(K,\leq)$ is a partially ordered set and both $+$ and $\cdot$ are monotonic.

 An ordered semiring $\mathbb{K}=(K,\leq,+,\cdot,0,1)$ is said to be
 \emph{lattice-complete} if: $(K,\leq)$ is a complete lattice; the units
 $0,1$ of $+,\cdot$ satisfy $0 \leq x \leq 1$ for each $x \in K$; and  
  \begin{displaymath}
    y + \sup_{i \in I} x_{i} \;=\; \sup_{i \in I} (y + x_{i})
  \end{displaymath}
  for each family $(x_{i})_{i \in I}$ and each $y \in K$. We define an
 infinite sum, as usual, by
  \begin{displaymath}
    \sum_{i \in I} x_{i} \;=\; \sup_{F \in {\mathcal{P}_{\mathrm{fin}}} (I)} \sum_{i \in F} x_{i}
  \end{displaymath}
  where ${\mathcal{P}_{\mathrm{fin}}} (I)$ is the set of finite subsets
 of $I$.

 A semiring is \emph{locally finite} if the underlying monoid $(K,\cdot,1)$ is locally
 finite, that is: for each finite subset $F\subseteq K$, the submonoid of $(K,\cdot,1)$ generated by $F$ is finite.
\end{defi}
The notion of $\mathbb{K}$-weighted (B\"uchi) automaton is
studied in~\cite{DrosteP07}, from which the following definition is taken.
\begin{defi}[$\mathbb{K}$-acceptance (B\"uchi) automaton, $\mathbb{K}$-weighted (B\"uchi) automaton]
  Let $(K,\leq,+,\cdot,0,1)$ be a lattice-complete semiring. A \emph{$\mathbb{K}$-acceptance (B\"uchi) automaton} is a tuple $\A = (\Sigma,Q,I,\delta,F)$, where $\Sigma$ is a finite  alphabet,
 $Q$ is a finite set of states, $I \subseteq Q$ is a set of initial
 states, $\delta : Q \times \Sigma \rightarrow \mathcal{P}(Q)$ is a
 transition function and $F : Q \rightarrow K$ is 
 a function that assigns an \emph{acceptance value} to each state.
  We define the language $\Lang(\A) : \Sigma^{\omega} \rightarrow K$ of $\A$ as
  \begin{displaymath}
    \Lang(\A)(w) = \sum_{ \rho \in \run(w)} \max \{ F(q) \mid q \in \mathrm{Inf}(\rho) \}\enspace.
  \end{displaymath}
  
  A \emph{$\mathbb{K}$-weighted (B\"uchi) automaton} is a tuple $\A =
 (\Sigma,Q,I,\delta,F)$, where $\Sigma$ is a finite  alphabet, $Q$ is a
 finite set of states, $I : Q \rightarrow K$ is a function assigns an
 \emph{initial weight} to each state, $\delta : Q \times \Sigma
 \rightarrow {K}^{Q}$ is a ($\mathbb{K}$-weighted) transition function
 and $F : Q \rightarrow K$ is a function assigns an \emph{acceptance value} to each state.
  We define the language $\Lang(\A) : \Sigma^{\omega} \rightarrow K$ of $\A$ by
  \begin{displaymath}
    \Lang(\A)(w) = \sum_{q_{0} q_{1} \ldots \in Q^{\omega}} \inf_{n \in \N} \sup_{i \geq n} \bigl(\, I(q_{0}) \cdot \delta (q_{0},w_{0})(q_{1}) \cdot\,\cdots\,\cdot \delta (q_{i-1},w_{i-1})(q_{i}) \cdot F(q_{i}) \,\bigr)\enspace.
  \end{displaymath}
\end{defi}

These notions specialize to $[0,1]$-acceptance automaton and fuzzy automaton~\cite{Rahonis05}
 by taking the fuzzy semiring $([0,1],\max,\min,0,1)$ as $\mathbb{K}$ in the above definitions.

Locally finiteness of a semiring~\cite{DrosteP07} is central in the
following result. Its proof is not hard but
 the result is not explicit in~\cite{DrosteP07} or elsewhere.
\begin{lem}\label{lem:nondeterminization}
  Let $\mathbb{K}=(K,\leq,+,\cdot,0,1)$ be a lattice-complete semiring
 and $\A = (\Sigma,Q,I,\delta,F)$ be a $\mathbb{K}$-weighted
 automaton.
 If
 $\mathbb{K}$ is locally finite (Def.~\ref{def:conditionsOnSemiring}),
 there exists a $\mathbb{K}$-acceptance
 automaton $\A' = (\Sigma,Q',I',\delta',F')$ such that
 \begin{math}
  \Lang(\A) = \Lang(\A')
 \end{math}.
\end{lem}
\begin{proof}
 Let $(F,\cdot,1)$ be the submonoid of $(K,\cdot,1)$ generated by
 the (finite) set of weights of transitions occurring in $\A$, that is,
 $\{\delta (q,a)(q')\mid  q,q' \in Q, a \in \Sigma\}$. The set $F$ is
 finite since $\mathbb{K}$ is locally finite.
 We now define $\A' = (\Sigma,Q',I',\delta',F')$ as follows.
 \begin{align*}
  Q' &= Q \times F\enspace,
  &
  I' &= I \times \{ 1 \}\enspace,
  \\
  \delta' \bigl(\,(q,k),\,a\,\bigr)
  &=
  \bigl\{\;
  \bigl(\,q',\;k \cdot \delta(q,a)(q')\,\bigr)
  \;\bigl|\bigr.\;
  q'\in Q
  \;\bigr\}\enspace,
  &
  F'(q,k) &=  k \cdot F(q)\enspace.
 \end{align*}
 The proof of $\Lang(\A)=\Lang(\A')$ is straightforward.
\end{proof}

It is straightforward that the fuzzy semiring $([0,1],\max,\min,0,1)$ is locally
finite. This leads to:
\begin{cor}
  Let $\A$ be a fuzzy automaton. There exists a $[0,1]$-acceptance 
  automaton $\A'$ such that 
 \begin{math}
  \Lang(\A) = \Lang(\A')
 \end{math}.
\qed
\end{cor}

The main results of~\cite{Rahonis05,DrosteP07} concern the
characterization of so-called \emph{$\omega$-rational formal power
series} over $\mathbb{K}$---those which are generated by
$\omega$-regular-like expressions---by $\mathbb{K}$-\emph{weighted} B\"uchi
automata. Lem.~\ref{lem:nondeterminization} therefore gives us another
characterization by $\mathbb{K}$-\emph{acceptance} B\"uchi automata.



\end{document}